%% file: Main.tex
\documentclass[11pt]{article}

\usepackage[margin=0.75in]{geometry}
\usepackage{graphicx}
\usepackage{color}
\usepackage{epsfig}
\usepackage{epstopdf}

\usepackage{subcaption}

\usepackage{booktabs}
\usepackage{mathrsfs} 
\usepackage{amsmath}
\usepackage{amssymb}
\usepackage{amsthm}
\usepackage{amsfonts,amssymb}

\usepackage{rotating}
\usepackage{verbatim}
\usepackage{hyperref}
\usepackage{bm}
\usepackage[normalem]{ulem}
\usepackage{bbm}
\usepackage{cleveref}
\usepackage{algorithm}
\usepackage{algpseudocode}
\usepackage{float}

\usepackage[authoryear]{natbib}

\usepackage{authblk}

\def\independenT#1#2{\mathrel{\rlap{$#1#2$}\mkern2mu{#1#2}}}
\newcommand\independent{\protect\mathpalette{\protect\independenT}{\perp}}

\DeclareMathOperator*{\argmin}{arg\,min}
\def\R{\mathbb{R}}
\newcommand{\One}[1]{\mathbbm{1}\left\{#1\right\}}

\newcommand{\eps}{\varepsilon}

\newcommand{\eqd}{\overset{d.}{=}}
\newcommand{\mc}{\mathcal}
\newcommand{\mb}{\mathbf}
\newcommand{\bs}{\boldsymbol}

\newcommand{\bmu}{\bs{\mu}}

\newcommand{\bOmega}{\bs{\Omega}}
\newcommand{\mX}{\mb{X}}
\newcommand{\mx}{\mb{x}}
\renewcommand{\P}{\mathbf{P}}
\newcommand{\E}{\mathbf{E}}

\newcommand{\Cor}{\operatorname{Cor}}
\newcommand{\tp}{^{\top}}

\newcommand{\indp}{\independent}

\usepackage{xspace}
\newcommand{\nameVV}{$M^1 P_1$\xspace}
\newcommand{\nameDP}{Bonf\xspace}

\newcommand{\nameERG}{Erd\H{o}s–R\'{e}nyi-Gilbert}

\theoremstyle{thmstyleone}%

\newtheorem{condition}{Condition}

\numberwithin{equation}{section}

\newtheorem{theorem}{Theorem}
\newtheorem{lemma}{Lemma}

\newtheorem{remark}{Remark}
\newtheorem{corollary}{Corollary}
\newtheorem{proposition}{Proposition}
\newtheorem{definition}{Definition}

\numberwithin{equation}{section}

 \title{
Goodness-of-Fit Tests for High-Dimensional Gaussian Graphical Models via Exchangeable Sampling
}
 
\author[1]{Xiaotong Lin}
\author[1]{Weihao Li}
\author[2]{Fangqiao Tian}
\author[1]{Dongming Huang}
\affil[1]{Department of Statistics and Data Science, National University of Singapore}
\affil[2]{Department of Mathematics, National University of Singapore}
\date{}

\begin{document}
\maketitle


\begin{abstract}
\input{Abstract.tex}
\end{abstract}




\section{Introduction}\label{sec: intro}
\label{sec:Introduction}
\input{1Introduction.tex}


\section{Sampling Exchangeable Copies}\label{sec: Sampling}
\input{2Sampling.tex}


\section{Goodness-of-Fit Test} \label{sec: Gof}
\input{3GGM}


\section{Power Analysis}\label{sec: power theory}
\input{4Power}


\section{Simulation Studies}\label{sec: Simulation}
\input{5Simulation}


\section{Real-World Examples}
\label{sec: Application}
\input{6Application}

\section{Discussion}
\label{sec: Discussion}
\input{7Discussion}

\section*{Acknowledgements}
D. Huang was partially supported by NUS Start-up Grant A-0004824-00-0 and Singapore Ministry of Education AcRF Tier 1 Grant A-8000466-00-00.


\bibliography{ref}
\bibliographystyle{unsrtnat}

\bigskip

\appendix

\noindent\textbf{\LARGE Appendix}
\smallskip
\newline

\section*{Notation}
\input{Notation}

\section{Connections with existing works}\label{sec: existing work}

\input{10AppendixLiterature}

\section{Details on Algorithms}\label{app: full algorithm}

\input{AppendixAlgorithm}

\section{Details on power theory}\label{app: full theory}
\input{AppendixPower}

\section{Details on Simulations in Section~\ref{sec: Simulation}}\label{app: full simulation}
\input{AppendixSimulation}
\section{Proofs}
\label{sec: Proof}
\input{8Proof.tex}


\section{Additional Simulations}
\label{sec: Appendix Simulation}
\input{9AddSim}

\section{Details of Real-World Examples}\label{sec: Appendix Application}
\input{9AppendixApplication}

\end{document}

%% file: Abstract.tex
We introduce a general framework for testing goodness-of-fit for Gaussian graphical models in both the low- and high-dimensional settings. 
This framework is based on a novel algorithm for generating exchangeable copies by conditioning on sufficient statistics. 
This framework provides exact finite-sample error control regardless of the dimension and allows flexible choices of test statistics to improve power. 
We explore several candidate test statistics and conduct extensive simulation studies to demonstrate their finite-sample performance compared to existing methods. 
The proposed tests exhibit superior power, particularly in cases where the true precision matrix deviates from the null hypothesis due to many small nonzero entries.
To justify theoretically, we consider a high-dimensional setting where the proposed test achieves rate-optimality under two distinct signal patterns in the precision matrix: (1) dense patterns with many small nonzero entries and (2) strong patterns with at least one large entry. 
Finally, we illustrate the usefulness of the proposed test through real-world applications.

%% file: 1Introduction.tex
Gaussian graphical models (GGMs) have been widely used for analysing complex dependencies among high-dimensional variables in various fields. 
By using graphs to represent conditional independence in multivariate Gaussian distributions, GGMs provide visual interpretations for relationships among variables. For example, GGMs have been applied to recover gene association networks \citep{schafer_empirical_2005, dobra_variable_2009, castelletti_bayesian_2021}, to model functional brain networks
\citep{belilovsky_testing_2015, smith_network_2011, varoquaux_brain_2010}, and to model cross-sectional and time-series data in psychological studies \citep{epskamp_gaussian_2018}.

Consider a $p$-dimensional random vector $X=(X_{1}, X_{2}, \cdots, X_{p})\tp$ following a multivariate normal distribution $\mathbf{N}_p(\bs{\mu}, \bs{\Sigma})$ with mean $\bs{\mu}$ and covariance matrix $\mathbf{\Sigma}$. 
The GGM concerns the precision matrix $\mathbf{\Omega}=\mathbf{\Sigma}^{-1}$,  where an off-diagonal entry $\omega_{ij} = 0$ implies $X_i$ and $X_j$ are conditionally independent given the remaining variables \citep{Dempster1972}.
Using this property, the GGM imposes constraints on the pattern of zero entries of $\mathbf{\Omega}$ to satisfy its conditional independence structure of the distribution. 
The structure is represented by a graph $G=(\mathcal{V}, \mathcal{E})$, where $\mathcal{V}=\left\{1, 2, \cdots, p\right\}$ is the node set and $\mathcal{E}$ is 
the edge set that include $(i,j)$ if $i\neq j$ and $\omega_{ij}$ is allowed to be nonzero. Formally, the $p$-dimensional GGM with graph $G=(\mc{V},\mc{E})$ is defined as: 
\begin{equation}\label{eq: model MG}
\mc{M}_{G}=\left\{ \mathbf{N}_p(\bmu, \bOmega^{-1}) : \bmu \in \mathbb{R}^p, \bOmega\succ \mathbf{0}, \bOmega_{i,j}=0 \text{ if } i\neq j \text{ and } (i,j) \notin \mc{E} \right\}. 
\end{equation}
Additional details on GGMs can be found in 
\citet {lauritzen_graphical_1996} and \citet{anderson_introduction_2007}.

This paper considers goodness-of-fit testing in GGMs, which tests whether the data-generating distribution belongs to the GGM with a given candidate graph.

\subsection{Motivation}\label{sec: motivation}

To apply a GGM for data analysis, it is important to determine the graph. 
This involves estimating the graph structure and assessing the adequacy of the GGM. 
Graph selection has been studied extensively, with important advancements such as the Graphical Lasso (GLasso) \citep{yuan_model_2007} and Neighbourhood Lasso \citep{meinshausen_high-dimensional_2006}. 
However, goodness-of-fit (GoF) testing for GGMs has received comparatively less attention.

GoF testing evaluates whether a prespecified graph sufficiently captures all true edges. 
Unlike graph selection, which aims to \textit{identify as many true edges as possible while avoiding false ones}, GoF testing 
aims to determine whether \textit{a specific candidate graph includes all true edges even if false edges are present}. 
The definition of Type-I error also differs: in graph selection, it means \textit{selecting any false edge}, while in GoF testing, it means \textit{rejecting the candidate graph when it already contains all true edges}. These conceptual differences lead to distinct methodologies, which are summarised in Table~\ref{tab:comparison}.
Appendix~A.1 provides a literature review of both problems. 

\begin{table}[!t]
    \centering 
    \caption{Comparison of goodness-of-fit (GoF) testing and graph selection for graphical models}\label{tab:comparison}
     \begin{tabular}{@{}p{3cm}@{\hspace{0.5cm}}p{5.5cm}@{\hspace{0.5cm}}p{5.5cm}@{}}
    \toprule
        \hline
        \textbf{Aspect} & \textbf{GoF Testing} & \textbf{Graph Selection} \\
       \addlinespace
        Goal & Assess a candidate graph $G_0$ before downstream analyses & Identify the underlying graph $G$  to represent conditional independence\\
        \addlinespace
       Focus & Test if $G_0$ includes all true edges, regardless of false edges & Select as many true edges as possible while avoiding false edges \\
       \addlinespace
        Null Hypothesis & Single hypothesis that $G_0$ includes all true edges & Multiple individual hypotheses about the absence of each edge \\
        \addlinespace
      Output & A testing decision or a p-value & A graph with selected edges \\
         \addlinespace
        Error Control & Rejection probability when $G_0$ includes all true edges & False discovery rate of selected edges \\ 
  \bottomrule
    \end{tabular}
\end{table}

In GoF testing, a candidate graph may be derived from domain expertise or historical data. 
Validating a candidate graph is important before using it for subsequent inference. 
For example, in analysing U.S. precipitation data, assessing whether the GGM based on geographic adjacency is appropriate can enhance the credibility of downstream analyses. 
In the model-X framework, if we are sure that the distribution of predictors belongs to a GGM, 
then we can reliably implement a conditional knockoff filter to select important variables for a response of interest with false discovery control. Appendix~A.2 provides more details on this matter.

\subsection{Problem formulation}
Suppose the rows of the observed data $\mathbf{X}$ are independent and identically distributed (i.i.d.) samples from some population $P$. 
Let $[p]:=\{1,2,\ldots, p\}$. 
Given a graph $G_0=([p], \mc{E}_0)$, the goal is to test the null hypothesis that
\begin{eqnarray}\label{eq: GoF null}
H_0: P\in \mc{M}_{G_0}, 
\end{eqnarray}
where $\mc{M}_{G_0}$ is the GGM w.r.t. $G_0$, as defined in \eqref{eq: model MG}.

The statement $P\in \mc{M}_{G_0}$ allows $G_0$ to include additional false edges, which is more flexible than the \textit{faithfulness} that requires each edge in $G_0$ corresponds to an off-diagonal nonzero entry of the precision matrix. 
In addition, the statement does not concern combinatorial properties of the graph, such as connectivity and the maximum degree. 
Both faithfulness testing and combinatorial inference are very different from GoF testing, and Appendices~A.3 and A.4 discuss these differences.

While several GoF tests for GGMs exist, none are fully satisfactory. 
Classical tests like the generalised likelihood ratio test are inapplicable in high dimensions. 
The multiplicity adjustment of pairwise tests by \citet{drton_multiple_2007} and the procedures by \citet{verzelen_tests_2009} remain suboptimal in certain situations, especially when signals are \textit{dense but weak}---that is, when the true precision matrix has many small nonzero entries absent from the null model. Furthermore, they cannot incorporate prior structural information. These limitations motivate the development of more powerful and flexible GoF tests.

To overcome the aforementioned challenges, we construct a valid randomisation test using arbitrary test statistics. Suppose we can sample $M$ copies $\widetilde{\mathbf{X}}^{(m)}$ of the observed data $\mathbf{X}$ so that $\left\{ \mathbf{X}, \widetilde{\mathbf{X}}^{(1)} , \cdots, \widetilde{\mathbf{X}}^{(M)} \right\}$ are exchangeable under the null hypothesis. 
Let $T(\mb{X})$ be a statistic chosen such that larger (positive) values are regarded as evidence against the null, and $T(\widetilde{\mathbf{X}}^{(m)})$ represents the corresponding test statistic of each copy. 
If we define a p-value by
\begin{equation}\label{eqn:pval-1}
\textnormal{pVal}(\mathbf{X},\widetilde{\mathbf{X}}^{(1)},\dots,\widetilde{\mathbf{X}}^{(M)}) = \frac{1}{M+1}\left(1+\sum_{m=1}^M \One{T(\widetilde{\mathbf{X}}^{(m)})\ge T(\mathbf{X})}\right),
\end{equation}
then it is guaranteed that $\mathbf{P}\left(\textnormal{pVal}\leq \alpha \right)\leq \alpha$ under the null hypothesis, where $\alpha$ is any predefined significance level. 
The general idea of this approach has been well-explored and is known as Monte Carlo testing (see Appendix A.5 for a review). 
The validity of the p-value defined in \eqref{eqn:pval-1} does not rely on the choice of $T(\cdot)$, and the computation does not require any knowledge about the sampling distribution of $T(\mb{X})$. 
Therefore, users have the freedom to choose any test statistic to improve power, such as by incorporating prior knowledge about potential alternatives.

The main requirement for this approach is the exchangeability among the observed data $\mb{X}$ and the generated copies $\widetilde{\mb{X}}^{(m)}$. 
This could be achieved by \textit{co-sufficient sampling}, which generates i.i.d. Monte Carlo samples from the conditional distribution given the value of a sufficient statistic; see Appendix~A.6 for a review. 
Although co-sufficient sampling has been applied to many parametric models, it is not straightforward to implement co-sufficient sampling for GGMs. 
The idea of conditioning on sufficient statistics has also been explored for model-X knockoff generation \citep{huang_relaxing_2020}. However, their methods cannot satisfy the exchangeability requirement due to the fundamental differences between knockoff and exchangeable copies; see Appendix A.5 for a detailed discussion.
Moreover, it remains unclear how to choose the test statistic to achieve high power.

\subsection{Main contributions}

To fill the aforementioned gaps, we provide the following contributions: 
\begin{enumerate}
\item In Section \ref{sec: Sampling}, we propose an efficient algorithm for generating exchangeable copies of the observed data under GGMs, which leverages the geometry of Gaussian conditional distributions and the properties of Markov chains.

\item In Section~\ref{sec: Gof}, we propose the MC-GoF test for \eqref{eq: GoF null} using the exchangeable copies.  
This method yields a valid p-value with any test statistic function and can be applied with any dimension. 
Furthermore, we propose several test statistic functions and extend the method to test the local Markov property.

\item 
Section~\ref{sec: power theory} presents a theoretical framework for the power properties of the MC-GoF test in high dimensions.
We show that the MC-GoF test can achieve rate-optimality under two distinct alternative hypotheses about how the precision matrix violates the null hypothesis:  
(1) a dense alternative with many small nonzero entries, and (2) a strong alternative with at least one large entry. 
In both cases, the signal strength required by the MC-GoF test to have full power matches the lower bound on the separation rate for any consistent test.

\item 
We conduct extensive simulations to investigate the finite-sample performance of the MC-GoF test across various graph structures in Section~\ref{sec: Simulation}. 
Our proposed method is shown to be competitive in general, and significantly outperforms existing methods when the signal pattern is dense but weak. 
\end{enumerate} 

We exemplify the usefulness of our methods through their applications to real-world datasets in Section~\ref{sec: Application}.
Section~\ref{sec: Discussion} discusses future directions, including adapting our methodologies to other graphical models. 
The supplementary material provides complete details of the literature review, algorithms, theoretical derivations, proofs, simulation results, and data analyses.

%% file: 2Sampling.tex
Suppose $G=(\mc{V}, \mc{E})$  is a given graph and the rows of $\mb{X}$ are $n$ i.i.d. samples from a distribution in $\mc{M}_{G}$. 
We propose an algorithm to generate copies $\widetilde{\mathbf{X}}^{(m)}$ ($1\leq m\leq M$) so that $\mb{X},\widetilde{\mathbf{X}}^{(1)}, \ldots, \widetilde{\mathbf{X}}^{(M)}$ are exchangeable. 
Our method is based on the idea of conditioning on sufficient statistics and the properties of Markov chains. 

Let $\mathbf{x}_{i\cdot}\tp $ be the $i$-th row of $\mathbf{x}$ and $f_{\bmu, \bOmega}(\mathbf{x})$ be the density function for $\mb{X}$. Since $\bOmega_{i,j}=0$ for $(i,j)\notin \mc{E}$, a sufficient statistic is given by $\psi_{G}(\mX)$, where $\psi_{G}(\cdot)$ is defined as $\psi_{G}(\mx):=\left( \sum_{i=1}^{n}\mx_{i\cdot}, ~~ (\mx^{\top}\mx)_{i,j}: i=j \text{ or } (i,j)\in \mc{E} \right)$ for $\mx \in \mathbb{R}^{n\times p}$.

Let $\Psi=\psi_{G}(\mX)$ be the observed sufficient statistic. 
If the value of $\psi_{G}(\mx)$ is fixed to be $\Psi$, the density function $f_{\bmu, \bOmega}(\mathbf{x})$ is free of $\mx$. 
Therefore, the conditional distribution of $\mX$ given $\Psi$ is uniform on a support defined as $\mc{X}_{\Psi} : = \{ 
\mx\in \mathbb{R}^{n\times p}: \psi_{G}(\mx)=\Psi 
\}$.
Note that $\mc{X}_{\Psi}$ is an algebraic variety in a high-dimensional space and, except for special cases such as the one presented in Appendix~B.1, it generally lacks a simple structure. As a result, it is difficult to sample directly from the uniform distribution on $\mc{X}_{\Psi}$. Instead, we construct Markov chains to generate exchangeable copies using the idea in \citet{besag_generalized_1989}, where we sequentially update the columns of the data matrix one at a time. 
It should be clarified that we are not using the Markov Chain Monte Carlo method to sample approximately from the conditional distribution; instead, we construct Markov chains to generate exchangeable copies.

Our sampling method will be presented in two steps: we first introduce the one-step update in Section~\ref{sec: residual rotation}, and then introduce the construction of the Markov chains in Section~\ref{sec: sample MC}.

\subsection{Residual Rotation}\label{sec: residual rotation}

In our method, we update each column $\mathbf{X}_i$ by sampling from its conditional distribution given the other columns and the sufficient statistic $\Psi$. 
Under the Gaussian graphical model, this conditional distribution is uniform over a certain linear constraint set in $\mathbb{R}^{n}$. 
To sample from this uniform distribution, we hold the projection of $\mathbf{X}_i$ onto the column space of $[\mathbf{1}_n, \mathbf{X}_{N_i}]$ fixed and then rotate the residual component uniformly at random. 
We name this procedure the \textit{residual rotation} and summarise it in Algorithm~\ref{alg: residual rotation}, with derivations provided in Appendix~B.2. 
Proposition~\ref{prop: residual rotation} states the desired properties, namely, the output $\widetilde{\mathbf{X}}_{i}$ and the input $\mathbf{X}_{i}$ are conditionally exchangeable given $\mathbf{X}_{-i}$.

\begin{algorithm}[H]
\caption{Sampling one column via residual rotation}\label{alg: residual rotation}
\hspace*{\algorithmicindent} \textbf{Input:}  $n\times p$ data matrix $\mathbf{X}$, index $i$ of the variable to sample, neighbourhood $N_i$ of $i$\\
\hspace*{\algorithmicindent} If $n\leq |N_i|+1$, output  $\widetilde{\mathbf{X}}_{i}=\mathbf{X}_{i}$; otherwise, proceed to the following steps. \\
\hspace*{\algorithmicindent} \textbf{Step 1:} Apply least squares linear regression to $\mathbf{X}_{i}$ on $\left[\bs{1}_{n}, \mathbf{X}_{N_i}\right]$ \\
\hspace*{\algorithmicindent} \textbf{Step 2:} Obtain the fitted vector $\bs{F}$ and the residual vector $\bs{R}$ from the regression \\
\hspace*{\algorithmicindent} \textbf{Step 3:} 
Apply linear regression to a standard normal $n$-vector  on $\left[\bs{1}_{n}, \mathbf{X}_{N_i}\right]$ and obtain the residual $\widetilde{\bs{R}}$ \\ 
\hspace*{\algorithmicindent} \textbf{Step 4:}  Compute and output  $\widetilde{\mathbf{X}}_{i}=\bs{F}+\widetilde{\bs{R}}\frac{\|\bs{R}\|}{\|\widetilde{\bs{R}}\| }$
\end{algorithm}

\begin{proposition}\label{prop: residual rotation}
    Suppose the rows of $\mathbf{X}$ are i.i.d. samples from a distribution in $\mc{M}_{G}$. 
    Let $\widetilde{\mathbf{X}}_{i}$ be the output of Algorithm~\ref{alg: residual rotation} with index $i$, $N_i$ be its neighbourhood, and  $\widetilde{\mathbf{X}}$ be a matrix formed by replacing the $i$-th column of $\mathbf{X}$ with $\widetilde{\mathbf{X}}_{i}$.
    Then, $\psi_{G}(\mathbf{X})=\psi_{G}(\widetilde{\mathbf{X}})$ and the conditional distribution of $\widetilde{\mathbf{X}}$ given $\mathbf{X}=\mathbf{x}$  is the same as that of $\mathbf{X}$ for almost every  $\mathbf{x}\in \mathbb{R}^{n\times p}$.  
Furthermore, if $n\geq 3+|N_i|$, then $\widetilde{\mathbf{X}}_{i}\neq \mathbf{X}_{i}$, a.s.
\end{proposition}

\subsection{Sampling Exchangeable Copies from Markov Chains}\label{sec: sample MC}
To generate exchangeable copies of $\mathbf{X}$, we construct forward and backward Markov chains using the residual rotation method. First, we run the residual rotation across all variables in turn to create a forward Markov chain $\mc{C}_1=(\mathbf{X}, \overrightarrow{\mathbf{X}}^{(1)}, \ldots, \overrightarrow{\mathbf{X}}^{(p)})$. Then, starting from this endpoint, we run a similar procedure backward with the 
reversed order of variables to obtain a backward Markov chain $\mc{C}_2=(\overrightarrow{\mathbf{X}}^{(p)}, \overleftarrow{\mathbf{X}}^{(p-1)}, \ldots, \overleftarrow{\mathbf{X}}^{(0)})$. 
If the distribution of $\mathbf{X}$ belongs to $\mc{M}_{G}$, Proposition~\ref{prop: residual rotation} implies that
$$
(\overrightarrow{\mathbf{X}}^{(p)}, \overleftarrow{\mathbf{X}}^{(p-1)}, \ldots, \overleftarrow{\mathbf{X}}^{(0)})\eqd
(\overrightarrow{\mathbf{X}}^{(p)}, \overrightarrow{\mathbf{X}}^{(p-1)}, \ldots,\overrightarrow{\mathbf{X}}^{(1)}, \mathbf{X} ).
$$
Therefore, conditioning on $\overrightarrow{\mathbf{X}}^{(p)}$, we have that $\mathbf{X}$ and  $\overleftarrow{\mathbf{X}}^{(0)}$ are i.i.d. 
We can independently conduct the backward sampling multiple times to obtain multiple copies $\widetilde{\mathbf{X}}^{(m)}$ for $m = 1, 2, \dots, M $ that are exchangeable jointly with $\mathbf{X}$. Derivations and illustrations are provided in Appendix~B.3. 

Our approach is inspired by the \textit{parallel} method introduced by \citet{besag_generalized_1989}. 
The starting point of the backward chains is called the \textit{hub} since it is the common centre of all chains. 
When generating each chain, we can repeat the iteration for $L$ times to prolong the chain, but $L=1$ already suffices to provide satisfactory results. More discussions on the choice of $L$ can be found in Appendix~B.3.2. 
The order $\mc{I}$ of variables to apply residual rotation can be customised as long as the backward chain uses the exact reverse order of updates.

The general procedure is summarised in Algorithm~\ref{alg: exchangeable}. It is well-suited to high-dimensional settings where $p\gg n$. 
First, the output copies have columns that differ from those of $\mathbf{X}$ if $n-2$ exceeds the degree of variables, regardless of $p$. In addition, it is computationally efficient because it supports independent sampling of $M$ Markov chains via parallel computing, and each chain has computational complexity $O(L p n d^2)$ scaling linearly in $p$.

\begin{algorithm}[hbtp]
\caption{Sampling exchangeable copies for a GGM}\label{alg: exchangeable}
\hspace*{\algorithmicindent} \textbf{Input:}  $n\times p$ data matrix $\mathbf{X}$, graph $G$, number of copies $M$,  number of iterations $L$ (set to 1 by default), permutation $\mc{I}$ of a subset $\mc{T}\subset [p]$. \\
 \hspace*{\algorithmicindent} \textbf{Step 1:} Start from $\mathbf{X}$ and run Algorithm~\ref{alg: residual rotation} according to the order $\mc{I}$ for $L$ times to generate $\mathbf{X}^{(hub)}$. \\ 
 \hspace*{\algorithmicindent} \textbf{Step 2:} For $m = 1,2, \cdots, M$, independently start from $\mathbf{X}^{(hub)}$ and run Algorithm~\ref{alg: residual rotation} according to the reversed order of $\mc{I}$ for $L$ times to generate $\widetilde{\mathbf{X}}^{(m)}$. \\
\hspace*{\algorithmicindent} \textbf{Output:} $\widetilde{\mathbf{X}}^{(1)}, \ldots, \widetilde{\mathbf{X}}^{(M)}$. 
\end{algorithm}

Proposition~\ref{prop: exchangeable} summarises the desirable exchangeability of the output of Algorithm~\ref{alg: exchangeable}. 
It serves as the foundation for the test in Section~\ref{sec: Gof}, as exchangeability can be used to construct exactly valid randomisation tests (see \citet[Chapter 4]{davison_hinkley_1997}).

\begin{proposition}\label{prop: exchangeable}
Let $G$ be a graph and suppose the rows of $\mathbf{X}$ are i.i.d. samples from a distribution in $\mc{M}_{G}$. 
Let $\{\widetilde{\mathbf{X}}^{(m)}\}_{m=1}^{M}$ be generated by Algorithm~\ref{alg: exchangeable} with $\mc{T}\subseteq [p]$. 
Then (1) $\mathbf{X}, \widetilde{\mathbf{X}}^{(1)}, \ldots, \widetilde{\mathbf{X}}^{(M)}$ are exchangeable, (2) $\widetilde{\mathbf{X}}_{-\mc{T}}^{(m)}=\mathbf{X}_{-\mc{T}}$ for all $m$, and (3) $\mathbf{X}_{\mc{T}}, \widetilde{\mathbf{X}}_{\mc{T}}^{(1)}, \ldots, \widetilde{\mathbf{X}}_{\mc{T}}^{(M)}$ are conditionally exchangeable given $\mathbf{X}_{-\mc{T}}$.

\end{proposition}

%% file: 3GGM.tex
We will exploit the exchangeable copies to construct goodness-of-fit tests for Gaussian graphical models in Section~\ref{sec: GoF pvalue} and then propose several test statistics in Section~\ref{sec: GoF test statistic}. In Section~\ref{sec: GoF local}, we adapt the method for testing the local Markov property.

\subsection{Construction of p-values}\label{sec: GoF pvalue}

To test the null hypothesis $H_0$ that $\mathbf{X}$ follows a distribution in $\mc{M}_{G_0}$ for a graph $G_0$, we generate $\widetilde{\mathbf{X}}^{(1)}, \ldots, \widetilde{\mathbf{X}}^{(M)}$ using Algorithm~\ref{alg: exchangeable}. 
Under $H_0$, Proposition~\ref{prop: exchangeable} ensures $T(\mathbf{X}), T(\widetilde{\mathbf{X}}^{(1)}), \ldots, T(\widetilde{\mathbf{X}}^{(M)})$ are exchangeable for any chosen test statistic $T$. As a result, the p-values computed as in \eqref{eqn:pval-1} are exactly valid at any significance level $\alpha$.

\begin{algorithm}[htbp]
\caption{Monte Carlo goodness-of-fit (MC-GoF) test for GGMs}\label{alg:gof}
 \hspace*{\algorithmicindent} \textbf{Input:} $n\times p$ data matrix $\mathbf{X}$, graph $G_0$, test statistic $T(\cdot)$, number of copies $M$ \\
 \hspace*{\algorithmicindent} \textbf{Step 1:} Apply Algorithm~\ref{alg: exchangeable} with $G=G_0$ and $\mc{I}=(1,2,\ldots,p)$ to obtain $\widetilde{\mathbf{X}}^{(1)}, \ldots, \widetilde{\mathbf{X}}^{(M)}$ \\
 \hspace*{\algorithmicindent} \textbf{Step 2:} Calculate the statistics  $T(\mathbf{X}), T(\widetilde{\mathbf{X}}^{(1)}), T(\widetilde{\mathbf{X}}^{(2)}), \cdots, T(\widetilde{\mathbf{X}}^{(M)}) $ \\
 \hspace*{\algorithmicindent} \textbf{Step 3:} Compute and output the (one-sided) p-value in \eqref{eqn:pval-1}

 \end{algorithm}

We summarise the testing procedure in Algorithm~\ref{alg:gof} and call it the \textit{Monte Carlo goodness-of-fit} (MC-GoF) test for GGMs. 
The test statistic $T(\cdot)$ can be any function where larger values provide stronger evidence against $H_0$; extensions to tie-breaking and two-sided rules can be found in Appendix~B.3.1. Proposition~\ref{prop: GoF valid} established the validity of the MC-GoF test for $H_0$. 
\begin{proposition}\label{prop: GoF valid}
For any test statistic $T(\cdot)$, the p-value obtained by Algorithm~\ref{alg:gof} is valid:
for any level $\alpha \in (0,1)$, we have  $\P(\textnormal{pVal}_{T}\leq \alpha)\leq \alpha$ if $P\in \mc{M}_{G_0}$. 
\end{proposition}

Based on this result, the MC-GoF test guarantees finite-sample Type-I error control with no constraint on the dimension $p$ and no need to know about the sampling distribution of the test statistic. 
It also bypasses the need to estimate any unknown parameters so $\mc{M}_{G_0}$ is allowed to have much more unknown parameters than the sample size $n$, which is particularly advantageous for high-dimensional settings where $p\gg n$. 
Specifically, the MC-GoF test remains applicable even when all degrees of $G_0$ reach $n/2$, while many existing methods impose sparsity assumptions that the maximum degree is at the order of $o(\sqrt{n/\log p})$.

Although the MC-GoF test does not require $G_0$ to be sparse, we focus on graphs with maximum degree $d$ at most moderate relative to $n$ (e.g., $d \leq n/2$) for practical reasons: 
extremely dense graphs are difficult to estimate with limited data, and they are rarely encountered in applications of graphical models.

The number of copies $M$ does not affect the validity of the test, but it should not be too small as this may reduce the power. 
For significance level $\alpha=0.05$, setting $M=100$, or more generally $M\geq 2/\alpha$, is sufficient, as further increasing $M$ yields minimal benefit. 
Setting the number of the iteration $L=1$ already achieves desirable power properties in simulations, even for moderately dense null graphs. Detailed discussions on the choice of $M$ and $L$ are provided in Appendix~B.3.2. 

\subsection{Choice of test statistic} \label{sec: GoF test statistic}

The MC-GoF testing framework (Algorithm~\ref{alg:gof}) ensures valid Type-I error control for any chosen test statistic $T(\cdot)$. However, to achieve high power, $T(\cdot)$ should effectively capture deviations from the null hypothesis. 
This section introduces three classes of test statistics and a fourth principle to incorporate prior knowledge. 
We provide an overview of the motivations and aggregation schemes. Detailed derivations are deferred to Appendix~B.4.

Consider alternatives where the population belongs to a GGM, but the true graph $G$ includes some edges absent from the null graph $G_0=(\mc{V}, \mc{E}_0)$ in \eqref{eq: GoF null}. 
The problem reduces to testing  
$H_0: \forall (i,j)\notin \mc{E}_0, ~\omega_{ij}=0$ versus $H_a: \exists (i,j)\notin \mc{E}_0, ~\omega_{ij}\neq 0$. 
We propose three strategies to measure the ``distributional difference'' between $H_0$ and $H_a$: 
\begin{enumerate}
    \item Partial-Correlation-Based Statistics: These statistics measure the conditional dependence between pairs of variables not connected in $G_0$. The motivation is that if $G$ contains edges that $G_0$ omits, some partial correlations will be nonzero. See more details in Appendix~B.4.1. 
    \item Conditional-Regression-Based Statistics: 
 Based on the idea of nodewise regressions, these statistics evaluate the regression of each variable on its neighbours in $G_0$ and test whether adding any omitted variable improves the fit. See more details in Appendix~B.4.2. 
 \item Likelihood-Ratio-Based Statistics: 
 These statistics compare the likelihoods under the null model with $G_0$ and a full model fitted using regularisation (e.g., GLasso). See more details in Appendix~B.4.3. 
\end{enumerate}

For partial-correlation-based and conditional-regression-based approaches, we first derive pairwise or nodewise measures that reflect potential departures from $G_0$ and then combine these local measures into one statistic. Two aggregation schemes are:
\begin{enumerate}
    \item[(i)] Sum-Type Aggregation: 
Summing up local measures is effective if
we consider the \textit{dense-but-weak} alternative, where there are many small deviations from $G_0$. 
 For example, from partial correlation measures, we define the SRC statistic as 
  $T_{\text{SRC}}(\mathbf{X}) = \sum_{(i,j)\notin \mc{E}_0} \hat{\gamma}_{ij}^2$, 
  where $\hat{\gamma}_{ij}$ is a residual correlation measure derived from regressing $i$ and $j$ on the union of their neighbours in $G_0$. 
 Other statistics of this type are PRC and ERC defined in Appendix~B.4.1, and F$_{\Sigma}$ in Appendix~B.4.2. 
 \item[(ii)] Max-Type Aggregation:
 Taking the maximum of local measures is effective if we consider the \textit{strong-but-sparse} alternative, where there are one or a few large deviations from $G_0$. 
For example, from partial correlation measures, the MRC statistic is defined as
$T_{\text{MRC}}(\mathbf{X}) = \max_{(i,j)\notin \mc{E}_0} \hat{\gamma}_{ij}^2$ in Appendix~B.4.1. 
Another example of this type is F$_{\max}$ defined in Appendix~B.4.2, where its relation to the $M^1 P_1$ and $M^1 P_2$ procedures proposed in \cite{verzelen_tests_2009} is further discussed. 
\end{enumerate}

The likelihood ratio approach inherently provides a single global statistic, so it does not require an additional aggregation step. 
For example, in Appendix~B.4.3, the GLR-$\ell_1$ statistic is defined as 
$T_{\text{GLR}-\ell_1}= 2 \log \ell(\widehat{\bOmega}_{ML})$, 
where $\widehat{\bOmega}_{GL}$ is a modified GLasso estimator that penalises only edges absent in $G_0$ and is tailored for high-dimensional settings.

When additional knowledge suggests that certain edges not in $G_0$ are more likely to exist, one can modify the test statistic to reflect this information. For example, if $w_{ij}\geq 0$ encodes prior beliefs about the existence of an edge $(i,j)$, then a weighted version of the SRC statistic is defined as
$T_{\text{SRC}-w}(\mathbf{X}) = \sum_{(i,j)\notin \mc{E}_0} w_{ij}\hat{\gamma}_{ij}^2$.
Appendix~B.4.4 discusses this idea in detail, and Appendix~F.4 illustrates the advantage of this approach with numerical examples.
More advanced methods, such as Bayesian approaches, demand future exploration.

\subsection{Testing local Markov property}\label{sec: GoF local}

Instead of a global test, some applications of GGMs may require a localised evaluation, such as testing the \textit{local Markov property} for a subset of nodes. 
For node $i$, this property states that $X_i$ is independent of other variables outside its neighbourhood $N_i$ given $X_{N_i}$. 
Testing the local Markov property for a subset $\mathcal{T}\subseteq [p]$ involves verifying 
\begin{equation}\label{eq: local null GoF}
    H_0: X \text{ is normal and } X_i\indp X_{B_i} \mid X_{N_i}, 
\end{equation}
where $B_i$ excludes $i$ and $N_i$. 
When $\mc{T}=[p]$, \eqref{eq: local null GoF} becomes the global hypothesis \eqref{eq: GoF null}. 

To obtain a valid p-value for such a local test, we modify Step 1 of Algorithm~\ref{alg:gof} by using a permutation of $\mathcal{T}$ only (rather than a permutation of $[p]$). 
We call this modified algorithm the \textit{local MC-GoF test} for GGMs and prove its validity for testing \eqref{eq: local null GoF} in Appendix~E.4. 
We can also adjust test statistics to focus on nodes in $\mathcal{T}$; for example, a localised version of 
SRC can be defined as $T_{\text{SRC}}^{(loc)}(\mathbf{X}) = \sum_{i\in \mc{T}}\sum_{j\in B_{i}} \hat{\gamma}_{ij}^2$.

This localised framework is more efficient in terms of both computation and power when only a part of the graph is of interest. 
It also serves as a diagnostic tool for identifying  nodes that contribute to a global model misspecification.

%% file: 4Power.tex
This section establishes a theoretical framework for analysing the power properties of the MC-GoF test in high-dimensional settings, where both $p$ and $n$ grow to infinity, with $p$ potentially much larger than $n$. 
We will quantify when the MC-GoF test can distinguish between the null hypothesis and structured alternatives, and will establish rate-optimal detection boundaries.

\subsection{Setup and Framework}
The power analysis involves understanding the following two aspects regarding the observed statistic $T(\mathbf{X})$ and the Monte Carlo copies $T(\widetilde{\mathbf{X}}^{(m)})$: \\
(A) the empirical distribution of $T(\widetilde{\mathbf{X}}^{(m)})$, $m\in [M]$; \\
(B) the distribution of $T(\mathbf{X})$ under a meaningful alternative hypothesis. \\
It is challenging to analyse Aspect (A) under general null graphs due to the complexity of Markov chains involved in the generation of $\widetilde{\mathbf{X}}^{(m)}$. 
To achieve a tractable theory, we assume the null graph $G_0=([p], \mc{E}_0)$ belongs to the class of \textit{clique-star shaped graphs}: the graph is a split graph with a clique $\mc{H}$ and an independent set $\mc{C}$ such that every node in $\mc{C}$ is connected to all nodes in $\mc{H}$. 
Such structures simplify the null model so we can explicitly characterise the distributions of the generated copies. Remark~\ref{rem:clique-star-shaped} discusses the necessity and novelty of this simplification.

To study Aspect (B), we focus on normal populations and 
consider two distinct signal patterns:  (1) the \textbf{dense-but-weak} pattern, where many off-diagonal entries of $\{\mathbf{\Omega}_{ij}:(i,j)\notin \mc{E}_0\}$ are nonzero but each is too small to detect individually, and (2) the \textbf{strong-but-sparse} pattern, where most off-diagonal entries of $\{\mathbf{\Omega}_{ij}:(i,j)\notin \mc{E}_0\}$ are zero except for at least one significant.  
These alternatives bridge between Aspect (B) and recent advances in correlation matrix inference.

Concretely, suppose the $n$ rows of $\mathbf{X}$ are i.i.d. sampled from $P=\mathbf{N}_p(\bs{\mu}, \mathbf{\Omega}^{-1})$. 
For a set $\mathbf{\Theta}$ of precision matrices, we write $\mathbf{\Omega}\in \mathbf{\Theta}$ to indicate $P\in \{\mathbf{N}_p(\bs{\mu}, \mathbf{\Omega}^{-1}): \bs{\mu}\in \mathbb{R}^{p}, \mathbf{\Omega}\in \mathbf{\Theta}\}$. 
Consider the null $H_0: P\in \mc{M}_{G_0}$ under the following condition. 
\vspace{-1em}
\begin{condition}\label{cond: G0 clique-star}
$G_0=([p], \mc{E}_0)$ is clique-star shaped with clique $\mc{H}$ and independent set $\mc{C}$ such that $|\mc{H}|\leq n-3$ and $q=|\mc{C}|$. 
Furthermore, either $p\geq n$ or $p=q$ holds. 
\end{condition}
Under Condition~\ref{cond: G0 clique-star}, $P\in \mc{M}_{G_0} \Leftrightarrow$ the submatrix $\mathbf{\Omega}_{\mc{C}}$ is a diagonal matrix, and so does its inverse.  
We measure deviations from the null using the following quantity: 
$$
D(\mathbf{\Omega}_{\mc{C}}, s) = \|\mathbf{R}(\mathbf{\Omega}_{\mc{C}}^{-1}) - \mathbf{I}_q\|_s,
$$
where $\mathbf{R}(\mathbf{M})=\text{diag}(\mathbf{M})^{-1/2}\mathbf{M}\text{diag}(\mathbf{M})^{-1/2}$ is the diagonal normalisation and the symbol $s\in \{F, \infty\}$ denotes the Frobenius norm or the maximum norm of a matrix, respectively. For both $s=F$ and $s=\infty$, we have $D(\mathbf{\Omega}_{\mc{C}}, s)=0 \Leftrightarrow P\in \mc{M}_{G_0}$. 
Remark~\ref{rem:power-metric} provides detailed discussions and concrete examples for this metric. 
As we will see, the power of the MC-GoF test depends on the magnitude of $D(\mathbf{\Omega}_{\mc{C}}, s)$. 

\subsection{Theoretical Results}
For dense alternatives where many of the off-diagonal entries of $\mathbf{\Omega}_{\mc{C}}$ are nonzero, we use $D(\mathbf{\Omega}_{\mc{C}}, F)$ to quantify the deviation from the null. 
For strong alternatives where some of the off-diagonal entries of $\mathbf{\Omega}_{\mc{C}}$ are separated from zero, we use $D(\mathbf{\Omega}_{\mc{C}}, \infty)$ to quantify the deviation. 
Our goal is to characterise how large these deviations must be (relative to $n$, $p$, and $q$) for the MC-GoF test to achieve high power, and to show that the rates of these deviations cannot be improved by any other test.
For any $\alpha\in (0,1)$ and any statistic $T$, 
we denote the power of the $\alpha$-level MC-GoF test with statistic $T$ as 
$$
\beta_{n}(P; \alpha, T):=\mathbb{P}\left( \frac{1}{M+1}\left[1+\sum_{m=1}^M \mathbf{1}\left\{T\left(\widetilde{\mathbf{X}}^{(m)}\right) \geq T(\mathbf{X})\right\}\right] \leq \alpha \right),
$$
where $\widetilde{\mathbf{X}}^{(m)}$ are generated by Algorithm~\ref{alg:gof} with $G_0$. 
We now present the main results. All formal statements and further discussion are provided in Appendix~C.

\textbf{Dense alternatives.} Consider the alternative  
$H_a: \mathbf{\Omega}\in \mathbf{\Theta}_{n1}(b)$,  
where $\mathbf{\Theta}_{n1}(b) := \left\{\mathbf{A}\in \mathbb{S}_+^p: D(\mathbf{A}_C, F) \geq b\sqrt{\frac{q}{n}}, \quad \lambda_{\max}(\mathbf{A}_C)/\lambda_{\min}(\mathbf{A}_C)\leq b_0 \right\}$ for a fixed $b_0\geq 1$. 
In this case, the SRC statistic is suitable and can be expressed as $T_1(\mathbf X)= \sum_{i,j \in \mc{C}, i\neq j}\hat{\gamma}_{ij}^2$. 

\begin{theorem}\label{thm:simple-dense-power}
    Suppose Condition~\ref{cond: G0 clique-star} holds and $\lim \frac{q}{n}=\gamma>0$. Then, the followings hold:
    \begin{enumerate}
        \item [(a)]For any $\epsilon>0$, there exists a constant $b$ such that if $\mathbf{\Omega}\in \mathbf{\Theta}_{n1}(b)$, it holds $\liminf_{n} \beta_{n}(P;\alpha, T_1)\geq 1-\epsilon$, 
provided that $M\geq \max(2\alpha^{-1},\log(2\epsilon^{-1}))$. 
        \item [(b)]
Suppose the condition number of $\mathbf{\Omega}_{\mc{C}}$ is bounded and $D(\mathbf{\Omega}_{\mc{C}}, F)\to\infty$. 
If $\alpha_n$ and $M$ scale appropriately, then $\lim_{n\rightarrow\infty}\beta_{n}(P; \alpha_n, T_1) = 1$.
        \item [(c)]
Let $0<\alpha<\beta<1$. Suppose $q / n \leq \kappa$ for some constant $\kappa<\infty$. 
Then, there exists a constant $b$ such that
$$
\limsup _{n \rightarrow \infty}~~\left\{ \sup_{\phi} \inf_{\mathbf{\Omega} \in \mathbf{\Theta}_{n1}(b)} \mathbb{E}(\phi) \right\}< \beta.
$$
where $\sup_{\phi}$ is taken over any $\alpha$-level test $\phi$ for $H_0: P\in \mc{M}_{G_0}$. 
    \end{enumerate}
\end{theorem}

\textbf{Strong alternatives.} 
Consider the alternative  
$H_a: \mathbf{\Omega}\in \mathbf{\Theta}_{n2}(b), \quad \text{where } \mathbf{\Theta}_{n2}(b)=\{\mathbf{A}\in \mathbb{S}_+^p: D(\mathbf{A}_C, \infty) \ge b\sqrt{\log(q)/n}\}.$ 
In this case, the MRC statistic is suitable and can be expressed as $T_2(\mathbf X)= \max_{i,j\in \mc{C},i\neq j} \hat{\gamma}_{ij}^2$.

\begin{theorem}
    Suppose Condition~\ref{cond: G0 clique-star} holds and that $\log q/n \to 0$. Then, the followings hold: 
    \begin{enumerate}
        \item [(a)]
    For any $\epsilon>0$, if $n>\max(10/\alpha, 32\log(16/\eps),8/\eps)$ and 
$M>\max(2\alpha^{-1},\log(2\epsilon^{-1})$, then $
  \underset{\mathbf{\Omega}\in \mathbf{\Theta}_{n2}(16)} {\inf}\beta_{n}(P;\alpha, T_2)\geq 1-\epsilon$. 
        \item [(b)] If $\alpha_n$ and $M$ scale appropriately, then $\lim_{n\rightarrow\infty}\inf_{\mathbf{\Omega}\in \mathbf{\Theta}_{n2}(16)} \beta_{n}(P; \alpha_n , T_2) = 1$. 
        \item [(c)]Suppose $\log (q) / n \leq \kappa$ for some constant $\kappa<\infty$. 
Then, for any constant $b\in (0,  \min(1,\kappa^{-1}))$, the following holds:
$$
\limsup_{n \rightarrow \infty}~~ \left\{ \sup_{\phi} ~~ \inf_{\mathbf{\Omega} \in \mathbf{\Theta}_{n2}(b)} \mathbb{E}(\phi) \right\}\leq \alpha, 
$$
where $\sup_{\phi}$ is taken over all $\alpha$-level test for $H_0: P\in \mc{M}_{G_0}$. 

    \end{enumerate}

\end{theorem}

\textbf{Union Alternative.}
We finally consider the scenario where the true precision $\mathbf{\Omega}$ lies in a union of dense and strong alternative sets. Consider the alternative 
$H_a: \mathbf{\Omega}\in \mathbf{\Theta}_{n3}(b)=\mathbf{\Theta}_{n1}(b)\cup \mathbf{\Theta}_{n2}(b)$. 
What makes this interesting is that $\mathbf{\Omega}$ can deviate from the null either via many small entries or at least one large entry, but not both simultaneously. Accordingly, we define a combined statistic $T_3(\mathbf{X})=\max\left( \frac{n}{2q}[T_1(\mathbf{X})-q(q-1)/n],~nT_2(\mathbf{X})-4\log q+\log\log q \right)$.

\begin{theorem}
     
Suppose Condition~\ref{cond: G0 clique-star} holds and $q/n \to \gamma \in (0,\infty)$. Then, the followings hold:
    \begin{enumerate}
        \item [(a)]For any $\epsilon>0$, there exists a constant $b$ such that if $\mathbf{\Omega}\in \mathbf{\Theta}_{n3}(b)$, it holds $\liminf_{n} \beta_{n}(P;\alpha, T_3)\geq 1-\epsilon$, 
provided that $M\geq \max(2\alpha^{-1},\log(2\epsilon^{-1}))$. 

        \item [(b)]Let $0<\alpha<\beta<1$. Suppose $q / n \leq \kappa$ for some constant $\kappa<\infty$. Then, there exists a constant $b$ such that
$$
\limsup _{n \rightarrow \infty}~~\left\{ \sup_{\phi} \inf_{\mathbf{\Omega} \in \mathbf{\Theta}_{n3}(b)} \mathbb{E}(\phi) \right\}< \beta.
$$
where $\sup_{\phi}$ is taken over any $\alpha$-level test $\phi$ for $H_0: P\in \mc{M}_{G_0}$.

    \end{enumerate}

\end{theorem}

For each alternative, the theorem first shows that the MC-GoF test, with a suitable statistic, achieves high power when the signal is sufficiently large, and then provides lower bounds on the detection limit that match the required signal levels up to a constant. 
Hence, the MC-GoF test is both powerful and rate-optimal.

Figure~\ref{fig:phase_transition} summarises the feasibility of distinguishing the null from the alternative based on $D(\mathbf{\Omega}_{\mc{C}}, F)$ and $D(\mathbf{\Omega}_{\mc{C}}, \infty)$. It illustrates the regimes where the MC-GoF test achieves rate-optimal detection.

\begin{figure}[hbtp]
    \centering
     \caption{
   Goodness-of-fit detection boundaries and powerful tests. The four quadrants represent different signal regimes: top-left (dense-but-weak), top-right (dense-and-strong), bottom-right (strong-but-sparse), and bottom-left (weak-and-sparse).
    }
    \includegraphics[width=0.8\linewidth]{phase_transition.pdf}
    \label{fig:phase_transition}
\end{figure}

Our theoretical analysis focuses on the large-sample power with clique-star-shaped null graphs. 
To complement the theory, we conduct extensive simulation studies in Section~\ref{sec: Simulation} and Appendix~F to evaluate the finite-sample performance of the MC-GoF test under general null graphs. 
These simulations demonstrate that the MC-GoF test remains robust and powerful across diverse graph structures. This suggests that the theoretical insights obtained for clique-star-shaped null graphs apply more broadly.

\begin{remark}\label{rem:clique-star-shaped}
The reduction to clique-star-shaped null graphs is a necessary and strategic choice for developing a rigorous theory for the MC-GoF test.
This reduction makes the characterisation of the MC-GoF test analytically tractable, which allows us to derive results on the power properties of the test that would otherwise remain inaccessible.
As the analysis of alternative hypotheses is clear by focusing on these null graphs, we established the rate-optimal results for two distinct and meaningful signal patterns that, to the best of our knowledge, have not previously been examined in the literature on goodness-of-fit testing for GGMs.
Furthermore, deriving these results for clique-star-shaped null graphs already demands substantial theoretical effort, which suggests that extending the theory to general null graph structures would be even more challenging. 
Nonetheless, the insights and benchmarks gained here provide a foundation for future extensions. 
\end{remark}
\begin{remark}\label{rem:power-metric}

We provide some intuitions about the metric $D(\mathbf{\Omega}_{\mc{C}}, s)$, which is the key quantity for characterising the power of the MC-GoF test. 
Suppose $X=(X_1,\ldots, X_p)^\top \sim \mathbf{N}_p(\mathbf{0}, \mathbf{\Omega}^{-1})$ and $\mc{H}=[p]\setminus\mc{C}$. We have
\begin{align*}
D(\mathbf{\Omega}_{\mc{C}}, \infty)&=\max_{i,j\in \mc{C}} ~~ \left|\operatorname{cor}\left(X_i, X_j\mid X_{\mc{H}}\right)\right|,\\
D(\mathbf{\Omega}_{\mc{C}}, F)^2 & \geq  \frac{1}{\lambda} \sum_{i,j\in \mc{C}}\mathbf{\tilde{\Omega}}_{i,j}^2, 
\end{align*}
where $\lambda$ is the condition number of $\mathbf{\Omega}_{\mc{C}}$ and 
$\mathbf{\tilde{\Omega}}_{\mc{C}}=\mathbf{R}(\mathbf{\Omega}_C)$. 
These suggest that $D(\mathbf{\Omega}_{\mc{C}}, \infty)$ measures the largest magnitude of the partial correlations that violate the null hypothesis, and the square of $D(\mathbf{\Omega}_{\mc{C}}, F)$ serves as a proxy for the sum of squared entries of the normalised $\mathbf{\Omega}_{\mc{C}}$. 
To illustrate the difference between $D(\mathbf{\Omega}_{\mc{C}}, \infty)$ and $D(\mathbf{\Omega}_{\mc{C}}, F)$, we provide two simple examples. 
\paragraph{Example (Dense but weak).}
Let $b$ be a positive constant and $\rho=b/\sqrt{n(q-1)}$. Define $\mathbf{J}=(1-\rho)\mathbf{I}_q+\rho \mathbf{1}_q\mathbf{1}_q^\top$ where $\mathbf{1}_q$ is the $q$-vector with all entries being $1$. 
Suppose $\mathbf{\Omega}_{\mc{C}}=\mathbf{J}^{-1}$, then we have $D(\mathbf{\Omega}_{\mc{C}},F)=b\sqrt{q/n}$ while $D(\mathbf{\Omega}_{\mc{C}},\infty)=\rho=b/\sqrt{n(q-1)}\ll \sqrt{\log(q)/n}$. 
According to Figure~\ref{fig:phase_transition}, 
the precision matrix in this example lies in the regime where the signals are dense but weak (the bottom right region). 

\paragraph{Example (Strong but sparse).}
Let $\rho=b\sqrt{\log(q)/n}$ and define $\mathbf{J}=\mathbf{I}_q+\rho \mathbf{E}$ where $\mathbf{E}=1_{\{i=1,j=2\}}+1_{\{i=2,j=1\}}$. 
Suppose $\mathbf{\Omega}_{\mc{C}}=\mathbf{J}^{-1}$, then we have
$D(\mathbf{\Omega}_{\mc{C}},\infty)=b\sqrt{\log(q)/n}$ while 
$D(\mathbf{\Omega}_{\mc{C}},F)=b\sqrt{2\log(q)/n}\ll \sqrt{q/n}$. 
According to Figure~\ref{fig:phase_transition}, 
the precision matrix in this example lies in the regime where the signals are strong but sparse  (the upper left region). 

\end{remark}

%% file: 5Simulation.tex
We conduct numerical experiments to evaluate the finite-sample performance of the proposed MC-GoF test compared to two baseline methods: the \nameVV procedure by \citet{verzelen_goodness--fit_2010} and the Bonferroni-adjusted approach (\nameDP) from \citet{drton_multiple_2007}. 
We implement the MC-GoF test with four test statistic functions: $\text{PRC}$, $\text{ERC}$, $\text{F}_{\Sigma}$, and GLR-$\ell_1$ defined in Appendix~B.4. 
To conserve space, we present an overview here and refer readers to Appendix~D for full simulation details, along with an illustration of valid Type I error control of all methods and an examination of the $M^1P_2$ procedure. 

We first consider three representative classes of graphs: band graphs, hub graphs, and \nameERG~random graphs. 
In each setting, we define 
a true distribution with precision matrix $\boldsymbol{\Omega}$ and a null graph $G_0$ with different configurations of population parameters. 
To clarify the impact of signals, we focus on scenarios where $G_0$ is a subgraph of the true graph so that any violation of the null reveals itself clearly through additional edges with varying signal strengths. 
We measure these signals by parameters $s$ or $\xi$ that control the Kullback–Leibler divergence between the true distribution and the null model.  Each scenario isolates specific aspects of graph structure, signal strength, dimension, and sample size.

Simulation results show that the MC-GoF test consistently outperforms or matches the power of \nameVV and \nameDP. 
Some more observations are 

(1) The MC-GoF test with F$_{\Sigma}$ is consistently among the top performers, making it our recommended default.  

(2) The power of the MC-GoF test tends to increase as  the dimension $p$ increases. In contrast, the baseline methods \nameVV and \nameDP suffer from a loss of power in the high-dimensional cases. This indicates that our proposed method is more applicable to high-dimensional problems than existing methods.

(3) When many edges are absent from the null but each is relatively weak, baseline methods struggle to detect deviations from the null, while the MC-GoF test achieves high power. 
This observation matches our theoretical results in Theorem~\ref{thm:simple-dense-power}, which shows the rate-optimality of the MC-GoF test for detecting dense-but-weak signals. 

(4) Additional simulations in Appendix~F.3 demonstrate that our tests are highly competitive with $M^1P_1$ for detecting strong-but-sparse signals.

We also extend the experiments to more complex graph structures. 
Appendix~F.2 presents experiments with tree, spatial, small-world, and scale-free graphs, which may be more commonly seen in real-world applications. 
Appendix~F.1 includes an example with moderately dense null graphs, where the minimum degree is not small relative to the sample size. 
These additional experiments consider cases where the null graph is not a subgraph of the true graph and can be relatively dense. 
Across all these challenging scenarios, MC-GoF tests, especially with $F_{\Sigma}$, significantly outperform the baseline methods.
These findings reinforce the superior performance and broad applicability of our proposed method.

%% file: 6Application.tex
This section empirically illustrates the effectiveness of our proposed methods through case studies in environmental science and finance. 

\subsection{Average Daily Precipitation in the United States}

Modelling daily precipitation is essential in various fields, such as agriculture \citep{naylor2007assessing}, hydrology \citep{hamududu2017assessing}, and climate\citep{cloke2009ensemble}. 
Here, we intend to model the average daily precipitation in each year across the contiguous United States. 
We collected data from 48 states from 1979 to 2011, excluding several states due to their geographical isolation, from the North American Land Data Assimilation System. 
This dataset has $n=33$ observations, and the dimension of variables is $p=48$. 
Appendix~G.1 reports details of data preprocessing as well as data analyses, for which we provide a summary below.

The null hypothesis states that the precipitation across states can be modelled by a GGM with $G_0$, the graph representing geographic adjacency. The maximum and median degrees of $G_0$ are 8 and 4 respectively. The testing procedures are identical to those in Section~\ref{sec: Simulation}, namely, the MC-GoF tests with statistics F$_{\Sigma}$, GLR-${\ell_1}$, PRC, and ERC, as well as the benchmark methods \nameVV and \nameDP. 

We first conducted a simulation study to confirm that all methods control the Type-I error, and then reported the p-values of the GoF test applied to real data. The MC-GoF test using the F$_{\Sigma}$ and the GLR-${\ell_1}$ statistic reject the null hypothesis with p-values smaller than 0.05, which suggests that a GGM defined by the geographic adjacency of the states falls short of modelling the precipitation. 
A plausible explanation is that the data are strongly associated with precipitation in neighbouring regions outside the U.S. (Canada, Mexico, and the oceans), and the absence of the precipitation data in these regions breaks the conditional independence between non-adjacent states. 
Besides, the other tests fail to reject the null hypothesis.

We further examine how sample size affects the GoF tests by analysing random subsamples. 
In summary, the MC-GoF test with the GLR-${\ell_1}$ statistic remains robust performance across varying subsample sizes, while the F$_{\Sigma}$ statistic rejects the null more often as the sample size increases. 
Other methods show little improvement as more data are included.
This example demonstrates the effectiveness of our MC-GoF test in detecting the deficiency of GGMs for modelling high-dimensional data.

\subsection{Modelling stock returns}

In stock market analysis, building a graphical model for stock returns can uncover dependencies among stocks and identify clusters or sector-specific interactions. 
We apply our MC-GoF test to assess the GGM for modelling 103 large-cap U.S. stocks spanning 11 industry sectors. 
The data consists of 92 weekly average returns from September 1, 2020, to June 30, 2022, obtained from Yahoo Finance. 
The dataset has $n = 92$ observations, and the dimension is $p=103$. 
Appendix~G.2 reports technical details of data preprocessing and results of statistical analyses, for which we provide a summary below.

To estimate a suitable graph $G$, we use a modified GLasso that leaves within-sector edges unpenalised. This approach incorporates industry knowledge to produce a super-graph $\widehat{G}$ that may include extra edges but still yields an adequate GGM. The resulting graph has a maximum node degree of 30, a median degree of 19, and about 18.4\% of node pairs connected.

We validate the adequacy of the GGM with $\widehat{G}$ using the MC-GoF test with statistic $\mathrm{F}_{\Sigma}$ (the overall winner in the simulation studies), as well as two benchmark methods ($M^1P_1$ and Bonf). 
All the resulting p-values exceed 0.3, which confirms that the GGM with $\widehat{G}$ is sufficient for modelling the stock data. 
This validation supports the suitability of the estimated graph for downstream tasks such as model-X inference for portfolio analysis or risk assessment.

%% file: 7Discussion.tex
In this work, we introduced the MC-GoF test for Gaussian graphical models (GGMs), which provides finite-sample Type-I error control in any dimension. 
This test is a randomisation test based on a novel algorithm for generating exchangeable copies alongside the observation under the GGM. 
Additionally, the test is flexible in choosing the test statistic function to accommodate different alternatives and to incorporate prior information. 
These advantages make the MC-GoF test especially effective in high-dimensional inference.

We developed a theoretical framework to analyse the power of the MC-GoF test in high-dimensional settings. 
We considered the \textit{dense-but-weak} alternative and the \textit{strong-but-sparse} alternative separately, where the MC-GoF test is proved to be not only powerful but also rate-optimal. 
Numerical studies further demonstrate the superior performance of the MC-GoF test compared to existing GoF tests, especially  for dense-but-weak signals.

Our work raises several future research directions. 
First, while we focused on GGMs, similar techniques may extend to other first-order Markov random fields with suitable sufficient statistics. For instance, discrete undirected graphical models might be handled by conditional sampling schemes inspired by our approach. 
Future research may extend these ideas to exponential family graphical models \citep{wainwright2008graphical}.
Second, while Bayesian inference offers a general framework for embedding prior information into test statistics in our method, as illustrated in Appendix~B.4.4, 
efficient implementations remain open challenges. 
Third, while our theory focuses on clique-star-shaped null graphs, extending these results to more general null graphs would be a natural next step. Relaxing these structural constraints introduces complex Markov chain dynamics and intricate dependency patterns, which require the development of new theoretical tools. 
Finally, while the MC-GoF test can validate assumptions of downstream tasks such as model-X variable selection, it would be valuable to develop a unified framework that integrates the MC-GoF test with model-X inference.

%% file: Notation.tex

We use the following notations in the main text and this supplementary material. 
 For any set $U$, let $|U|$ be its cardinality. 
 Define $[p]:=\{1,\dots,p\}$ for any positive integer $p$. 
 For $S\subseteq [p]$ and a vector $X$, let $X_S$ be the sub-vector $(X_{j})_{j \in S}$;
 and for a matrix $\mathbf{X}$, let $\mathbf{X}_j$ be the $j$th column and $\mathbf{X}_S$ the submatrix formed by the columns corresponding to $S$. 
 Also, let $X_{-S}:=(X_j)_{j\notin S}$ and $\mathbf{X}_{-S}:=\mathbf{X}_{S^c}$. 
 For a square matrix $\mathbf{A}$, $\text{diag}(\mathbf{A})$ is the diagonal matrix with the same diagonal entries as $\mathbf{A}$.

We write $\mathbf{A}\succ\mathbf{0}$ if $\mathbf{A}$ is positive definite. 
Let $\textbf{N}_p(\bmu,\mathbf{\Sigma})$ denote a $p$-dimensional normal distribution with mean $\bmu$ and covariance $\mathbf{\Sigma}$, and define $\mathbf{\Omega}:=\mathbf{\Sigma}^{-1}$ with entries $\omega_{ij}$.
We consider undirected graphs without loops or multiple edges. 
For a graph $G=(\mc{V}, \mc{E})$ with node set $\mc{V}$ and edge set $\mc{E}$, denote the neighborhood of node $i$ by $N_i=\{j\in  \mc{V}: j\neq i, ~~ (i,j)\in \mc{E} \}$. 
We use the terms `node' $(j \in[p])$ and `variable' $\left(X_j\right)$ interchangeably. We write $\mc{M}_{G}$ for the GGM w.r.t. $G$.

%% file: 10AppendixLiterature.tex
\subsection{Comparison between GoF testing and graph selection} \label{rem: graph selection}

In the past two decades, high-dimensional GGM estimation has gained a significant amount of attention. 
The current primary methodologies can be categorized into two types. 
The first involves utilizing penalized likelihood with an $\ell_1$ penalty on entries of the precision matrix for variable selection \citep{yuan_model_2007,dAspremont2008,friedman_sparse_2008}. 
In particular, \citet{friedman_sparse_2008} propose an efficient algorithm called Graphical Lasso (GLasso) for solving the $\ell_1$ penalized likelihood estimation. Rates of convergence and model selection consistency have then been developed \citep{rothman_sparse_2008,lam2009,ravikumar_high-dimensional_2011}. 
In addition to the $\ell_1$ penalized likelihood, several related methods have also been introduced; for example, the neighborhood selection \citep{meinshausen_high-dimensional_2006}, graphical Dantzig selector \citep{yuan2010high}, and CLIME \citep{cai_constrained_2011}. 
The second branch of research in GGM estimation focuses on edge selection through multiple-testing procedures to control either the family-wise error rate or the false discovery rate; see, for example, \citet{schafer_empirical_2005,wille_low-order_2006,drton_model_2004,drton_multiple_2007,drton_sinful_2008,liu_gaussian_2013,li2021ggm}.

While the estimation of precision matrices has been well-explored, research on goodness-of-fit hypothesis testing within high-dimensional GGMs is relatively less developed. 
Nonetheless, several advancements have been made in this area. 
For instance, edge selection methods with family-wise error control, such as those in \citet{drton_multiple_2007}, can be adapted for GoF testing. 
\citet{verzelen_tests_2009} introduced a general GoF test by connecting the local Markov property and conditional regression of a Gaussian random variable.
However, this test is unable to achieve high power when the signals are dense but weak.  
\citet{bodnar_exact_2023} proposed finite-sample tests for particular precision matrix structures (e.g., block-diagonal, AR(1), factor structures), but their approach requires that $n > p$. 
These above tests focus on specific test statistics, while our method is flexible to choose any test statistic.

There are fundamental distinctions between GoF testing and graph selection. GoF testing aims to determine whether a specific candidate graph includes all true edges, regardless of the presence of additional false edges. In contrast, graph selection aims to identify as many true edges as possible while minimizing the inclusion of false ones. These differences extend to the definition of Type-I error: in GoF testing, it corresponds to rejecting the null hypothesis when the candidate graph already contains all true edges, whereas in graph selection, it corresponds to the inclusion of any false edge in the estimated graph. Consequently, the methodologies and theoretical frameworks for GoF testing are very different from those for graph selection. These distinctions, along with their practical implications, are summarized in Table~\ref{tab:comparison} in the main text.

\subsection{Connection between GoF testing and model-X inference} \label{rem: modelX}



GoF testing is important in the \textit{model-X framework} for high-dimensional inference.
Unlike traditional methods that model the conditional distribution of a response $Y$ given the predictor $X$, the model-X framework shifts the focus to modeling the distribution $F_X$ of $X$. 
For methods like the model-X knockoff filter \citep{candes_panning_2017}, which identifies important predictor variables for $Y$ while controlling the false discovery rate, accurate modeling of $F_X$ is essential. 
In particular, when $F_X$ is modeled by a GGM and the underlying graph is a subgraph of a known graph $G$, the model-X knockoff filter can be efficiently implemented \citep{huang_relaxing_2020}. 
The validity of this approach relies on the assumption that $F_X\in \mc{M}_{G}$, which can be assessed statistically through GoF testing. 

\subsection{Comparison between GoF and faithfulness}\label{rem: faithfulness}

In GGMs, the population $P$ is said to be \textit{faithful} to a graph if $\omega_{ij}\neq 0\Leftrightarrow $ $i$ and $j$ are connected. 
It is associated with the following model
$$
\overline{\mc{M}}_{G}=\left\{ \mathbf{N}_p(\bmu, \bOmega^{-1}) : \bmu \in \mathbb{R}^p, \bOmega\succ \mathbf{0}, \bOmega_{i,j}=0 \text{ if and only} i\neq j \text{ and } (i,j) \notin \mc{E} \right\}. 
$$
Compared to \eqref{eq: model MG}, we have $\overline{\mc{M}}_{G}\subsetneq \mc{M}_{G}$. 
Therefore, 
faithfulness testing is different from GoF testing. 
Recently, \citet{le_testing_2022} introduced asymptotic tests for testing the faithfulness to a prespecified graph, which is more restricted than the hypothesis for GoF testing in \eqref{eq: GoF null}. 
In addition, their theory requires the sparsity assumption that the maximum degree is $o(\sqrt{n/\log p})$, and thus, their tests do not apply in the settings of this paper.

We adopt the hypothesis \eqref{eq: GoF null} as the null hypothesis for GoF testing in GGMs because our focus lies in validating the adequacy of the GGM $\mc{M}_{G_0}$. 
Once validated, the model $\mc{M}_{G_0}$ can be used in subsequent analyses, such as implementing model-X knockoff filter to identify important predictor variables, even if $G_0$ may include false edges.

\subsection{Comparison between GoF testing and combinatorial testing} \label{rem: combtest}

\citet{neykov_combinatorial_2018} studied testing for some characteristics of the graph structure such as the connectivity, the presence of a cycle, and the maximum degree. 
These tests target specific global characteristics of the underlying graph, rather than assessing the goodness-of-fit for a prespecified graph, where the mapping between node and variable is fixed. 
In principle one could formulate a  structural property as a composite null hypothesis that includes a large set of graphs, where these graphs are invariant to changes in the mapping between nodes and variables. 
However, since this composite null hypothesis could contain an exponentially growing number of graphs, how to extend a GoF test to handle such a hypothesis remains an open challenge. 

\subsection{Review of Monte Carlo tests} \label{rem:MCtest}

Our GoF test belongs to the class of Monte Carlo tests. 
For problems where sampling from the null distribution is impracticable, 
\citet{besag_generalized_1989} proposed two methods to generate exchangeable copies of the data using Markov chains. 
A related line of research lies in the use of approximate sufficiency for asymptotic GoF testing, which could be expedient if finding an appropriate sufficient statistic is hard. 
\citet{barber_testing_2021} considered sampling data conditional on an asymptotically efficient estimator and derived bounds on the finite-sample Type-I error inflation of their method. 
Recently, \citet{zhu2023approximate} developed an extension to constrained and penalized maximum likelihood estimation.
These methods are approximative in nature and require additional assumptions about the distribution, the dimension of variables, and the sample size, and thus are not directly comparable to our method.

\subsection{Review of co-sufficient sampling}\label{rem:css literature}


Our GoF test for GGMs is closely related to the framework of co-sufficient sampling.  
Under this framework, GoF testing is achieved by sampling copies from the conditional distribution of the data given a sufficient statistic. This idea can be traced back at least to \citet{bartlett_properties_1937}, and the general problem of Monte Carlo computation of the conditional expectation given a sufficient statistic has been explored by various researchers \citep{agresti1992survey,engen_stochastic_1997,lindqvist_counterexample_2003,lindqvist_monte_2005,lindqvist2007conditional}. 
\citet{diaconis1998algebraic} proposed Markov chain algorithms for sampling from discrete exponential families conditional on a sufficient statistic. 
\citet{lockhart_use_2007} and \citet{lockhart2009exact} considered using the Gibbs sampler for co-sufficient sampling for the Gamma distribution and the von Mises distribution. \citet{gracia2005transformations} and \citet{o2006conditional} studied the conditional sampling for inverse-Gaussian distributions, and 
\citet{santos2019metropolis} proposed a Metropolis-Hasting algorithm for exponential families with doubly transitive sufficient statistics.
To the best of our knowledge, existing literature has not considered co-sufficient sampling for GGMs with dimensions of variables larger than sample sizes, as considered by the current paper. 
While \cite{huang_relaxing_2020} studied conditioning on sufficient statistics to construct high-dimensional model-X knockoff variables, their construction is limited to knockoff generation and does not extend to co-sufficient sampling. We provide a further discussion on the comparison between knockoff and exchangeable samples in Appendix~\ref{rem:knockoff}.

\subsection{Comparison between knockoffs and exchangeable copies} \label{rem:knockoff}

Algorithm~\ref{alg: exchangeable} generates copies $\widetilde{\mathbf{X}}^{(m)}$ of $\mathbf{X}$ such that $\{\mathbf{X}, \widetilde{\mathbf{X}}^{(1)}, \ldots, \widetilde{\mathbf{X}}^{(M)}\}$ are exchangeable, that is, 
the joint distribution is invariant to permuting any subset of these datasets.
Sampling these copies $\widetilde{\mathbf{X}}^{(m)}$ may appear similar to the knockoff construction \citep{barber2015controlling,candes_panning_2017}, but there is a crucial difference. 
We use the model-X knockoff framework proposed in \cite{candes_panning_2017} as an example, and the discussion below easily extends to the fixed-X knockoff framework used in \cite{barber2015controlling} and \cite{li2021ggm}.

In the model-X knockoff framework, a single dataset $\tilde{\mathbf{X}}$ is constructed to replicate the dependencies in the original dataset $\mathbf{X}$. 
The key property of the knockoff construction, known as the swap invariance property, requires that for each variable $j$, swapping the $j$-th column of $\mathbf{X}$ with the corresponding column of $\tilde{\mathbf{X}}$ leaves their joint distribution unchanged. 
Formally, for any $j \in \{1,2,\ldots,p\}$, we have the distributional equivalence  
$$[\mathbf{X}, \tilde{\mathbf{X}}]_{\text{swap}(j)} \stackrel{d}{=} [\mathbf{X}, \tilde{\mathbf{X}}].
$$  
The swap invariance property is the foundation of the FDR control of the knockoff filter in variable selection. 

To connect with the focus of the current paper, the swap invariance property also guarantees that the original and knockoff datasets, as a pair, are exchangeable. 
However, when multiple knockoffs $\{\widetilde{\mathbf{X}}^{(1)}, \ldots, \widetilde{\mathbf{X}}^{(M)}\}$ are generated, the pairwise exchangeability does not automatically extend to the joint exchangeability. 
Since each knockoff dataset is generated by conditioning on the observed data $\mathbf{X}$, $\mathbf{X}$ plays a distinguished role that the knockoffs do not share. 
This distinction is critical: while the swap invariance property ensures pairwise exchangeability between  $\mathbf{X}$ and $\widetilde{\mathbf{X}}^{(m)}$, it does not guarantee joint exchangeability across multiple knockoffs and $\mathbf{X}$ together.

In fact, exchanging the position of $\mathbf{X}$ with any of the knockoffs $\widetilde{\mathbf{X}}^{(m)}$ will change the joint distribution.
For instance, 
$$
[\mb{X},\widetilde{\mathbf{X}}^{(1)}, \ldots, \widetilde{\mathbf{X}}^{(M)}] \text{ and } 
[\widetilde{\mathbf{X}}^{(1)},\mb{X}, \ldots, \widetilde{\mathbf{X}}^{(M)}]
$$
do not follow the same distribution in general.

%% file: AppendixAlgorithm.tex
\subsection{Special sampling with complete graphs}\label{rem:complete-graph}
    There is a special case where sampling from the uniform distribution of $\mc{X}_{\Psi}$ is simple: when $G$ is a complete graph, which means that every node is connected to all other nodes. In this case, $\mc{E}=\{(i,j):\forall i, j \in [p], i\neq j\}$, and the minimal sufficient statistic $\psi_{G}$ consists of the sample mean and the sample Gram matrix. 
    The conditional distribution of $\mathbf{X}$ is the uniform distribution on all $n\times p$ matrices with the same sample mean and sample Gram matrix as the observed ones. 
    We can sample from this conditional distribution using the construction  $\widetilde{\mathbf{X}}=\mathbf{1}\mathbf{1}^\top \mathbf{X}+\mathbf{\Gamma}\mathbf{H}\mathbf{\Gamma}^\top\mathbf{X}$, where $\mathbf{\Gamma}$ is an $n\times (n-1)$ orthonormal matrix that is orthogonal to the vector $\mathbf{1}$ and $\mathbf{H}$ is a uniform random $(n-1)\times (n-1)$ orthonormal matrix.
    The astute reader may notice a connection between this simple generation and the fixed-X knockoff construction in \cite{barber2015controlling}. 
    The difference is that the generated copy $\widetilde{\mathbf{X}}$ is not a knockoff matrix and the matrix $\mathbf{H}$ needs to be random. 
 The case where $G$ is the complete graph only serves for pedagogical purposes,  since the model $\mc{M}_{G}$ includes all $p$-dimensional multivariate normal distributions and is not interesting in graphical modeling.

\subsection{Derivation of Residual Rotation}\label{rem:rr}

The Lebesgue density of the joint distribution of $\mathbf{X}$ is given by
\begin{eqnarray}\label{eq: pdf GGM}
f_{\bmu, \bOmega}(\mathbf{x})=(2 \pi)^{- np / 2} \operatorname{det}(\bOmega)^{n/2} \exp \left(-\frac{1}{2}\sum_{i=1}^{n}(\mathbf{x}_{i}-\bs{\mu})\tp \bOmega(\mathbf{x}_{i}-\bs{\mu})\right), ~~\forall \mx\in \mathbb{R}^{n\times p}, 
\end{eqnarray} 
where $\mathbf{x}_{i}\tp $ is the $i$th row of $\mathbf{x}$.

Consider updating one column $\mathbf{X}_{i}$ according to its conditional distribution given all other columns $\mathbf{X}_{-i}$ as well as the sufficient statistic $\Psi$. 
We first study the conditional distribution of $\mathbf{X}_{i}$ given $\mathbf{X}_{-i}$ and then the conditional distribution given $(\mathbf{X}_{-i},\Psi)$. 

By the properties of the multivariate normal distribution\citep{anderson_introduction_2007}, 
the conditional distribution of $\mathbf{X}_{i}$ given $\mathbf{X}_{-i}$ is $\mathbf{N}_{n}(\mu_{i} \bs{1}_n + \left(\mathbf{X}_{-i} - \bs{1}_{n} \bs{\mu}_{-i}^\top \right) \bOmega_{-i, i}  \bOmega_{i,i}^{-1}, ~~\bOmega_{i,i}^{-1} \mathbf{I}_{n} )$, where $\mu_i$ is the $i$-th entry of $\bs{\mu}$. Under the model $\mc{M}_{G}$, if $j$ is not a neighbor of $i$, then $\bOmega_{j, i}=0$. Therefore, $\mathbf{X}_{-i} \bOmega_{-i, i}\bOmega_{i,i}^{-1}=\mathbf{X}_{N_i} \bs{\alpha} $ for some $|N_i|$-dimensional parameter $\bs{\alpha}$ that depends on $\bOmega$.

For this conditional distribution, a sufficient statistic is given by $$\psi_{G}^{i}(\mathbf{X}):=\left[ \mathbf{X}_{i}^{\top} \bs{1}_{n}, \mathbf{X}_{i}^{\top} \mathbf{X}_{i}, \mathbf{X}_{i}^{\top} \mathbf{X}_{N_i}\right].$$ 
Denote by $\Psi^{i}$ the observed value of this statistic.  
Recall the definition of $\psi_{G}$ that
\begin{eqnarray}\label{eq: suff stat}
    \psi_{G}(\mx):=\left( \sum_{i=1}^{n}\mx_{i}, (\mx^{T}\mx)_{i,j}: i=j \text{ or } (i,j)\in \mc{E} \right),  ~~\forall  \mx \in \mathbb{R}^{n\times p},
\end{eqnarray}
and note that $\psi_{G}^{i}(\mathbf{x})$ is contained by $\psi_{G}(\mathbf{x})$. Therefore, according to the Lebesgue density of the joint distribution in 
\eqref{eq: pdf GGM}, the conditional distribution of $\mathbf{X}_{i}$ given $(\mathbf{X}_{-i}, \Psi)$ is the uniform distribution on 
\begin{equation}\label{eq: conditional support}
\mc{X}_{\Psi}^{i}: = \{ 
\mx_{i}\in \mathbb{R}^{n}: \psi_{G,i}(\mx)=\Psi^{i} 
\}. 
\end{equation}

Algorithm~\ref{alg: residual rotation} in the main text provides a simple way to sample from this uniform distribution. 
The idea of this algorithm is to fix the projection of $\mathbf{X}_{i}$ onto the column space of $\left[\bs{1}_{n}, \mathbf{X}_{N_i}\right]$ and to rotate the residual uniformly and independently. 
We name this procedure the \textit{residual rotation}.
Given $(\mathbf{X}_{-i}, \Psi)$, the output $\widetilde{\mathbf{X}}_{i}$ is a conditionally i.i.d. copy of $\mathbf{X}_{i}$. 
As long as $n\geq |N_i|+3$, the output $\widetilde{\mathbf{X}}_{i}$ is different from $\mathbf{X}_{i}$ almost surely.  

\subsection{Details on Algorithm~\ref{alg: exchangeable}}\label{rem:alg2}

To generate exchangeable copies of $\mathbf{X}$, we introduce a method that updates the individual columns one at a time to generate Markov chains. 
This method combines the residual rotation (Algorithm~\ref{alg: residual rotation}) in Section~\ref{sec: residual rotation} of the main paper and the parallel method introduced in \citet{besag_generalized_1989}.

When Algorithm~\ref{alg: residual rotation} is executed repeatedly with index running from $1$ to $p$, we obtain a Markov chain $\mc{C}_1=(\mathbf{X}, \overrightarrow{\mathbf{X}}^{(1)}, \ldots, \overrightarrow{\mathbf{X}}^{(p)})$, where the endpoint $\overrightarrow{\mathbf{X}}^{(p)}$ is an entirely new sample and the arrow pointing right overhead in the notation $\overrightarrow{\mathbf{X}}^{(i)}$ emphasizes that this Markov chain is going forward. 
 
This is illustrated in Figure~\ref{fig:rotation}, where each row shows one step of residual rotation.

\begin{figure}[!ht]
    \centering
    \includegraphics[width=1\textwidth]{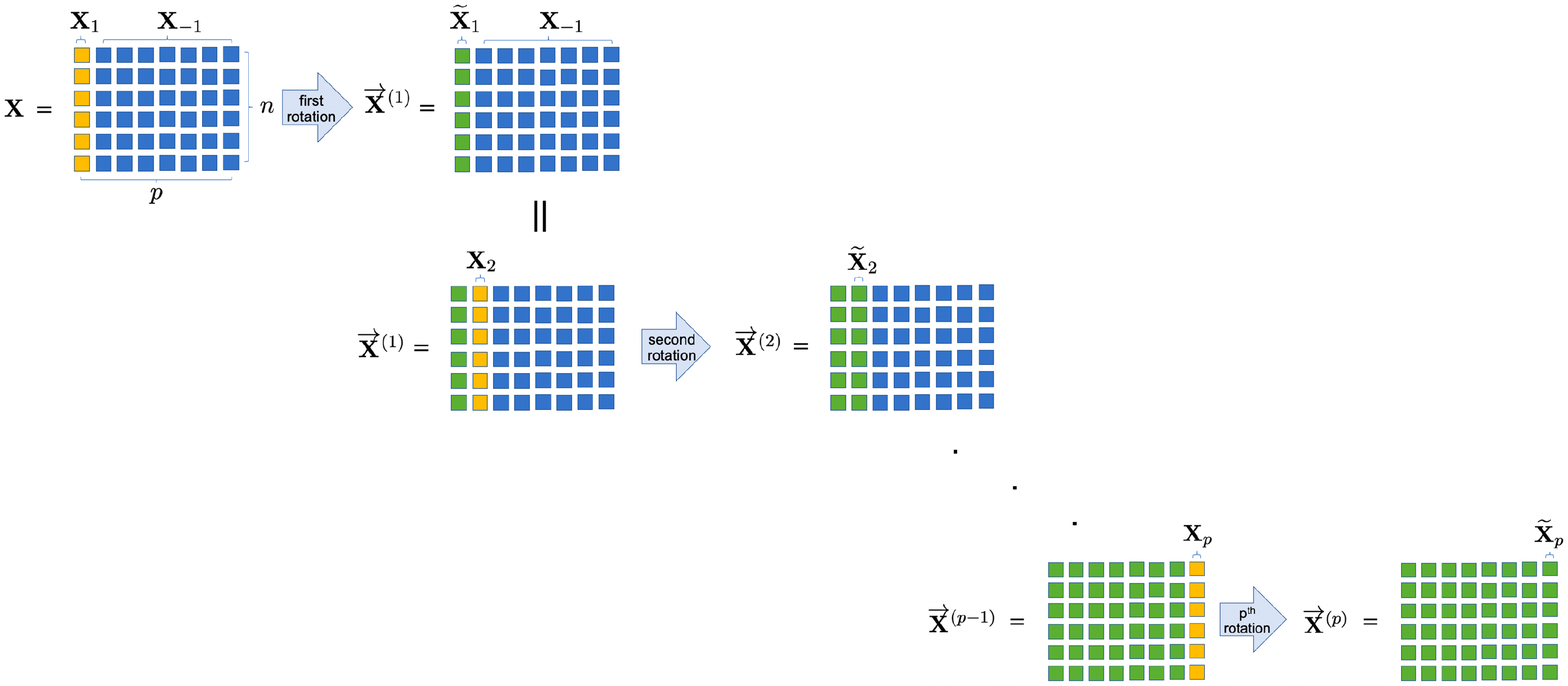}
    \caption{Forward sampling process. 
    Each mini square represents an entry of the data matrix. 
    At each step, Algorithm~\ref{alg: residual rotation} is applied to the yellow column $\mathbf{X}_i$, producing the green column. 
    The updated matrix is used as input for the next step. 
    Iterating column by column yields a fully updated matrix.  
    }
    \label{fig:rotation}
\end{figure}

If we start from $\overrightarrow{\mathbf{X}}^{(p)}$ and independently run Algorithm~\ref{alg: residual rotation} backward along the reversed order (from $p$ to $1$), we obtain a second Markov chain $\mc{C}_2=(\overrightarrow{\mathbf{X}}^{(p)}, \overleftarrow{\mathbf{X}}^{(p-1)}, \ldots, \overleftarrow{\mathbf{X}}^{(0)})$, where the arrow pointing left overhead in the notation $\overleftarrow{\mathbf{X}}^{(i)}$ emphasizes that this chain is going backward.

If the distribution of $\mathbf{X}$ belongs to $\mc{M}_{G}$, the backward Markov chain $\mc{C}_2$ has the same distribution as the reversed chain of $\mc{C}_1$, as seen by applying Proposition~\ref{prop: residual rotation} at each residual rotation. 
In other words, we have 
$$\label{eqref}
(\overrightarrow{\mathbf{X}}^{(p)}, \overleftarrow{\mathbf{X}}^{(p-1)}, \ldots, \overleftarrow{\mathbf{X}}^{(0)})\eqd
(\overrightarrow{\mathbf{X}}^{(p)}, \overrightarrow{\mathbf{X}}^{(p-1)}, \ldots,\overrightarrow{\mathbf{X}}^{(1)}, \mathbf{X} ). 
$$
Therefore, conditioning on $\overrightarrow{\mathbf{X}}^{(p)}$, we can conclude that $\mathbf{X}$ and  $\overleftarrow{\mathbf{X}}^{(0)}$ are independent and identically distributed, and thus $\mathbf{X}$ and  $\overleftarrow{\mathbf{X}}^{(0)}$ are conditionally exchangeable. Since the backward sampling can be conducted independently multiple times, we can obtain multiple versions of $\overleftarrow{\mathbf{X}}^{(0)}$ (denoted as $\widetilde{\mathbf{X}}^{(m)}$ for $m = 1, 2, \dots, M $) that are exchangeable jointly with $\mathbf{X}$.

In the generation of each Markov chain, we can repeat the residual rotation with the index running from $1$ to $p$ (or from $p$ to $1$) for $L$ iterations so that each Markov chain is lengthened from $p$ to $p L$. 
Unlike traditional MCMC, which requires sufficiently long chains to approximate target distributions, our algorithm generates exchangeable copies even with $L=1$. 
Although increasing $L$ seems beneficial for reducing the correlation between the endpoint and the starting point of the chain, further analyses in Appendix~\ref{app: L} show that choosing $L=1$  suffices to provide satisfactory results.

The sampling procedure is illustrated in Figure \ref{fig:MC} and 
is summarized in Algorithm~\ref{alg: exchangeable}. 
In Algorithm~\ref{alg: exchangeable} in the main paper, the residual rotation is applied following any order $\mc{I}$, and the endpoint of the forward Markov chain $\mc{C}_1$ is called the \textit{hub} since it is the common center of all Markov chains. 
Our Markov chain construction is based on the \textit{parallel} method introduced by \citet{besag_generalized_1989}, which has been revisited recently by \citet{berrett_conditional_2020,barber_testing_2021,howes2023markov}.

\begin{figure}[hbtp]
    \centering
    \includegraphics[width=.7\textwidth]{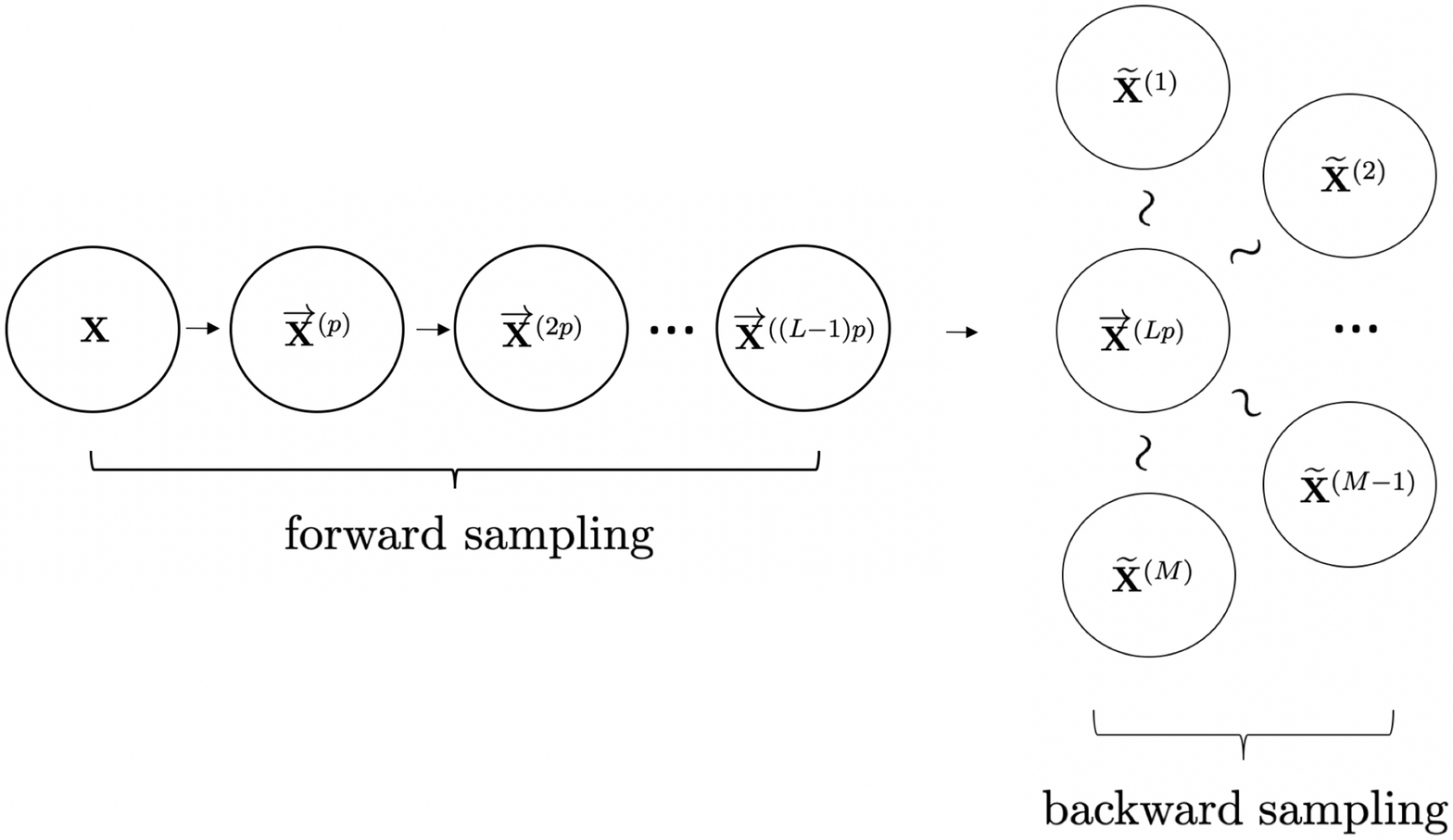}
    \caption{
    Sampling exchangeable copies of $\mathbf{X}$. 
    Left: Forward sampling from $\mathbf{X}$, and each right arrow represents L-time forward residual rotation. 
    Right: Backward sampling from the hub $\overrightarrow{\mathbf{X}}^{(Lp)}$, and each wavy line represents a backward sampling of a Markov chain of length $Lp$.}
    \label{fig:MC}
\end{figure}

By setting the argument $\mc{I}$ as a permutation of $[p]$, 
Algorithm~\ref{alg: exchangeable} allows updating the columns in any chosen order. 
Exchangeability is maintained as long as the backward sampling in Step 2 follows the reverse sequence of $\mc{I}$. 
For example, we can arrange the nodes in the graph based on their degrees such that  $d_{j_1}\leq d_{j_{2}}\leq \ldots, \leq d_{j_p}$, and we set $\mc{I}=(j_1, j_2, \ldots, j_p)$. 
Besides, $\mc{I}$ can be chosen as a permutation of a subset of $[p]$ so that some columns of the output remain unchanged. 
This approach is useful for constructing the test for the local Markov property in Section~\ref{sec: GoF local} of the main paper. 

\subsubsection{Extensions to randomized and two-sided tests}\label{remark:p-value}

In Algorithm~\ref{alg:gof}, the function $T(\cdot)$ can be any test statistic function such that larger values are regarded as evidence against the null hypothesis. 
Here, we discuss how we can use randomized tests to boost power and two-sided tests for more general choices of test statistics. 
The p-values defined below are both valid under the condition of Proposition~\ref{prop: GoF valid}. We include the justification for completeness in Appendix~\ref{app:Proof of Lemma 3}. 
\begin{enumerate}
    \item 
The p-value defined in Algorithm~\ref{alg:gof} is conservative when there are ties. Alternatively, one can break the ties randomly. 
We simplify the notation $ T(\mathbf{X})$ and $ T(\widetilde{\mathbf{X}}^{(m)})$   as $T_0$ and $T_m$.  
Let $\kappa=\sum_{i=1}^{M} \One{T_i =T_0}$, and let $S$ be a random integer from $\left\{1,2, \cdots, \kappa+1 \right\}$. 
We define the randomized p-value as follows:
$$
\frac{1}{1+M}\left(S  + \sum_{i=1}^{M} \One{T_i >T_0} \right). $$

\item 
    Algorithm~\ref{alg:gof} is stated in terms of a one-sided test, which treats larger values of $T$ as evidence for rejecting the null hypothesis. 
    We can also consider a two-sided test that rejects the null when $T$ is either too large or too small. 
    With $S$ and $\kappa$ defined in the first point, we define the two-sided p-value as follows:
\[
\frac{2}{1+M}\left[ S+\min\left(  \sum_{i=1}^{M} \One{T_i >T_0},\sum_{i=1}^{M} \One{T_i <T_0}  \right)   \right].
\]
\end{enumerate}

\subsubsection{Choices of parameters in Algorithm~\ref{alg:gof} }\label{rem:ml}

The number $M$ of copies in Algorithm~\ref{alg: exchangeable} has a minor impact on the test. 
Proposition~\ref{prop: GoF valid} guarantees that the validity of the p-value is unaffected by $M$. 
However, $M$ should not be too small in order to achieve good power. 
Given a significance level $\alpha$, setting $M \geq 2/\alpha$ ensures that the computed p-values can be smaller than $\alpha/2$. 
Both theory and empirical results suggest that as long as $M$ exceeds a moderate threshold, further increases have little impact. Specifically, the theoretical results in Section~\ref{sec: power theory} of the main paper show that taking $M = 2 \max(\alpha^{-1}, \log n)$ is sufficient to achieve essentially full power when the alternative is well separated from the null (see Corollaries~\ref{cor:dense consistent} and \ref{cor:sparse-consistent}). 
In practice, $M=100$ already suffices for commonly used significance levels such as $\alpha = 0.05$, and increasing $M$ beyond that offers minimal further benefit.

The number of iterations $L$ in Algorithm~\ref{alg: exchangeable} controls the length of the generated Markov chains. 
As discussed in Section~\ref{sec: sample MC} in the main paper and Appendix~\ref{app: L}, 
$L=1$ is sufficient in most practical scenarios even when the graph is moderately dense. 

Moreover, the asymptotic theory on the power property in Section~\ref{sec: power theory} of the main paper holds for any $L$, including $L=1$. 
In our numerical studies, choosing $L=3$ achieves desirable power performance, with no noticeable improvement observed for larger values of $L$. 
Hence, we recommend choosing $L$ between 1 and 3. 

\subsection{Derivations of Test Statistics} \label{rem: stat}

Algorithm~\ref{alg:gof} is valid in controlling the Type-I error with any test statistic function, but to achieve high power, we aim to choose a function that is likely to be large when the null hypothesis is false. 
The general principle is to find a statistic that captures the ``distributional difference" between possible alternative hypotheses and the null hypothesis. 
Motivated by the classical inference methods such as the likelihood ratio test, we propose three types of test statistics in Appendix~\ref{sec: GoF test partial cor}, \ref{sec: gof test F}, and \ref{sec: GoF test glr}. 
We further discuss in Appendix~\ref{sec: GoF test general} how a practitioner can construct problem-specific test statistics when prior knowledge about plausible alternative hypotheses is available.

To address the most common scenarios in GoF testing for GGMs, we focus on alternatives where the population belongs to a GGM, but the underlying graph differs from $G_0=(\mc{V}, \mc{E}_0)$, the graph stipulated by the null hypothesis. 
In the following, we denote by $X=(X_{1}, X_{2}, \cdots, X_{p})\tp \in \mathbb{R}^p$ a multivariate normal random vector from $\mathbf{N}_{p}(\bs{\mu}, \bs{\Sigma})$ and the precision matrix is  $\bs{\Omega}=\left( \omega_{ij}\right)_{1\leq i,j\leq p}=\bs{\Sigma}^{-1}$. 
Suppose that the distribution of $X$ is \textit{pairwise faithful} to the graph $G=(\mc{V}, \mc{E})$ in the sense that for any distinct pairs of $i$ and $j$, $\omega_{ij}=0 \Leftrightarrow (i,j)\notin \mc{E}_{}$.
Under this assumption, the testing problem for the hypothesis in \eqref{eq: GoF null} is reduced to testing $H_0: \mc{E}\subset \mc{E}_{0} \mbox{~~versus~~} H_a: \mc{E}\not\subset \mc{E}_{0}$, 
which is equivalent to testing
\begin{eqnarray}\label{eq: hypotheses omega}
    H_0: \forall (i,j)\notin \mc{E}_{0} , ~~\omega_{i,j}=0 ~~ \mbox{~~versus~~} H_a: \exists (i,j) \notin \mc{E}_{0},~~ \omega_{i,j}\neq 0. 
\end{eqnarray}
We propose several choices of test statistics that aim to estimate some measures of ``distributional difference" between the null hypothesis and the alternative hypothesis in \eqref{eq: hypotheses omega}. 
Although these statistics are rooted in standard practices for GoF testing, they are particularly suited for high-dimensional problems. 
The simulation studies in Section~\ref{sec: Simulation} of the main paper demonstrate that our MC-GoF testing procedure using any of these test statistics has competitive finite-sample performance. 
It should be emphasized that the testing problem in \eqref{eq: hypotheses omega} merely serves as a motivation for deriving our test statistics, and the MC-GoF test continues to be valid for testing the null hypothesis stated in \eqref{eq: GoF null}.

In the simulations of the main paper, we examined the following four test statistic functions: $\text{PRC}$ \eqref{eq: partial cor ss}, $\text{ERC}$ \eqref{eq: efficient partial cor} in Appendix~\ref{sec: GoF test partial cor}, $\text{F}_{\Sigma}$ \eqref{eq: F sum} in Appendix~\ref{sec: gof test F}, GLR-$\ell_1$ \eqref{eq: GLR-l1} in Appendix~\ref{sec: GoF test glr}.

\subsubsection{Pairwise partial correlation}\label{sec: GoF test partial cor}
Our first test statistic is based on the estimation of partial correlations. 
The $ij$-th partial correlation $\rho_{ij}$ is the correlation between $X_i$ and $X_j$ in their conditional distribution given other variables. 
By the properties of the normal distribution, the partial correlation can be expressed as 
\begin{equation}\label{eq: pc}
\rho_{ij}=\frac{-\omega_{ij}}{\sqrt{\omega_{ii}\omega_{jj}}}, 
\end{equation}
and thus $\omega_{ij}=0$ if and only if $\rho_{ij}=0$. Therefore, the testing problem in \eqref{eq: hypotheses omega} can be converted to the testing of 
$$ 
H_{0}: \forall (i,j)\notin \mc{E}_0,  ~~\rho_{ij}=0 \mbox{~~versus~~} H_a: \exists (i,j) \notin \mc{E}_0,~~ \rho_{i,j}\neq 0. $$
This null hypothesis is a composition of $H_{0}^{ij}: \rho_{ij}=0$ for all $(i,j)$ not in $\mc{E}_0$. 
It is natural to consider using the $ij$-th sample partial correlation as the test statistic for each $H_0^{ij}$. 
When the sample size $n$ is larger than the dimension $p$, sample partial correlations can be computed by replacing the elements of $\mathbf{\Omega}$ in \eqref{eq: pc} with the corresponding elements of the inverse of the sample covariance matrix. 
See \citet{drton_multiple_2007} for an application of sample partial correlations for GGM selection.

Instead of using the sample partial correlations, we utilize the graph $G_0$ stipulated in the null hypothesis to formulate a statistic that can be computed even when $p\gg n$. 

By the properties of the normal distribution \citep{anderson_introduction_2007}, $X_i$ can be expressed as
\begin{equation}\label{eq: locl regression}
X_{i}= X_{-i}^{\top} \theta^{(i)} + \epsilon_{i}, 
\end{equation}
where $\epsilon_{i}\sim N(0,\omega_{ii}^{-1})$ 
 is independent of $X_{-i}$, and  
$\theta^{(i)}\in \mathbb{R}^{p-1}$ is a vector such that $\theta^{(i)}=-\bs{\Omega}_{-i,i}/\omega_{ii}$, where $\bs{\Omega}_{-i,i}$ is the $i$-th column of $\bs{\Omega}$ excluding the $i$-th entry.
For any pairs of distinct $i$ and $j$, $(\epsilon_i, \epsilon_j)$ is normally distributed and the $ij$-th partial correlation satisfies that  
$$\Cor(\epsilon_i, \epsilon_j)=-\rho_{ij}.$$
This suggests that we can first perform local linear regression to $X_{i}$ and $X_{j}$ respectively, and then consider the correlation between the residuals as a test statistic for the hypothesis $H_0^{ij}$. 

The detailed derivation is as follows. 
Suppose $(i,j)\notin \mc{E}_0$. Let $U$ be the union of the neighborhoods $N_i$ and $N_j$, and denote by $u$ the cardinality of $U$. 
If $u+2<n$, $\left[\mathbf{1}_n,\mathbf{X}_{U}\right]$ has full column rank almost surely. 
In this case, we can regress $\mathbf{X}_{i}$ on $\mathbf{X}_{U}$ (with the intercept term) using least squares estimation and obtain the residual vector $\mathbf{r}_{i}$. 
Similar estimation can be applied for $\mathbf{X}_{j}$ to obtain $\mathbf{r}_{j}$. 
Define the pairwise residual correlation $\hat{\gamma}_{ij}$ as 
\begin{equation}\label{residual correlation}
  \hat{\gamma}_{ij} :=  \mathbf{r}_{i}\tp \mathbf{r}_{j}/ \left(\|\mathbf{r}_{i}\| \|\mathbf{r}_{j}\|  \right). 
\end{equation} 
For simplicity, we define $\hat{\gamma}_{ij}:=1$ if $n\leq u+2$.
A simple test statistic can be defined as 
\begin{align}
    T_{\text{SRC}}(\mathbf{X})&=\sum_{ \left( i,j \right) \notin \mathcal{E}_0 } \hat{\gamma}_{ij}^2, \label{eq: sum residual correlation}
\end{align}
where SRC stands for the sum of squared pairwise residual correlations. Aggregation via summation effectively captures dense (but potentially weak) signals, which means the true precision matrix has many small nonzero entries not indicated by the null graph $G_0$. 
Another simple choice of test statistic is 
\begin{align}
    T_{\text{MRC}}(\mathbf{X})&=\max_{ \left( i,j \right) \notin \mathcal{E}_0 } \hat{\gamma}_{ij}^2,   \label{eq: maximum residual correlation}
\end{align}
where MRC stands for the maximum of pairwise residual correlations. 
Aggregation by taking the maximum effectively captures strong (but potentially sparse) signals, which means that the true precision matrix has at least one nonzero entry not indicated by the null graph $G_0$. 

Although simple, the power of the MC-GoF test with SRC or MRC can be theoretically studied. In Section~\ref{sec: power theory} of the main paper, we prove that when the alternative hypothesis has a dense signal (or a strong signal), the MC-GoF with $T_{\text{SRC}}$ (or with $T_{\text{SRC}}$) can achieve high power asymptotically with the optimal rate. 

By considering the distribution of $\hat{\gamma}_{ij}$ under the null hypothesis if $n>2+u$, we transform it as $\hat{t}_{ij}=\sqrt{(n-2-u)}\hat{\gamma}_{ij}/\sqrt{1-\hat{\gamma}_{ij}^2}$ and define 
\begin{equation}\label{eq:pcr pvalue exact}
p_{ij}:=2 F_{\text{t},n-2-u_{ij}}(- |\hat{t}_{ij}|),
\end{equation}
where $F_{\text{t},a}$ is the cumulative distribution function (CDF) of the t-distribution with $a$ degrees of freedom. 
If $n\leq u+2$, we define $p_{ij}:=1$
We define the z-score for $p_{ij}$ as
\begin{equation*}\label{eq:pcr z score}
z_{ij}=-\Phi^{-1}(p_{ij}/2), 
\end{equation*} 
where $\Phi^{-1}$ is the inverse CDF of the standard normal distribution. 
To aggregate $p_{ij}$'s in \eqref{eq:pcr pvalue exact} while preventing the signals from being overwhelmed by noises, we filter out large $p_{ij}$'s. 
Specifically, let $\delta$ be a pre-specified constant (0.05 by default). We define $\Xi =\left\{\left( i,j \right)\notin \mc{E}_0: p_{ij}\leqslant \delta\right\}$ and define 
\begin{align}
    T_{\text{PRC}}(\mathbf{X})&=\sum_{ \left( i,j \right) \in \Xi } z_{ij}^2. \label{eq: partial cor ss}
\end{align}
Since this statistic is obtained from the aggregation of \underline{p}airwise \underline{r}esidual \underline{c}orrelations, we call it the $\text{PRC}$ statistic.

The computation of the PRC statistic could become burdensome when $p$ is large, since it requires the computation of all pairwise residual correlations $\hat{\gamma}_{ij}$, each of which involves fitting two local linear regressions of $\mathbf{X}_{i}$ and $\mathbf{X}_{j}$ on variables in the union of their neighborhoods. 
We can improve the computational efficiency by fitting a local linear regression for each $\mathbf{X}_{i}$ on its own neighborhood $N_i$ to obtain the residual $\mathbf{r}_i$. 
We then define the residual correlation $\tilde{\gamma}_{ij}$ as in \eqref{residual correlation}. 
Unlike $\hat{\gamma}_{ij}$, the exact distribution of $\tilde{\gamma}_{ij}$ is not easy to obtain. 
Inspired by the Fisher's variance stabilizing transformation \citep[Proposition 2]{drton_multiple_2007}, we define 
\begin{equation}
\tilde{\xi}_{ij} :=\frac{\sqrt{n-2-\tilde{u}}}{2} \log \left(\frac{1+\tilde{\gamma}_{ij}}{1-\tilde{\gamma}_{ij}}\right),
\end{equation}
where $\tilde{u}=\min(|N_i|, |N_j|)$. 
If $n\leq \tilde{u}+2$, we define $\tilde{\xi}_{ij} :=0$. 
Define $\tilde{p}_{ij} =2 \Phi(- | \tilde{\xi}_{ij}|)$ and $\tilde{\Xi} =\left\{\left( i,j \right)\notin \mc{E}_0: \tilde{p}_{ij}\leqslant \delta\right\}$. 
A computationally efficient statistic is defined as
\begin{equation}\label{eq: efficient partial cor}
T_{\text{ERC}}=\sum_{(i,j)\in \Xi} \tilde{\xi}_{ij}^2.
\end{equation}
Since this statistic is obtained from the aggregation of an \underline{e}fficient computation of pairwise \underline{r}esidual \underline{c}orrelations, we call it the $\text{ERC}$ statistic.

\subsubsection{Conditional regression}\label{sec: gof test F}
We propose a powerful test statistic by taking advantage of a method originated in \citet{verzelen_tests_2009}. Based on the conditional regression of $X_i$ on the other variables, as expressed in \eqref{eq: locl regression}, the null hypothesis in \eqref{eq: hypotheses omega} can be decomposed into the following hypotheses for all $i\in [p]$:
$$
H_0^{(i)}: \forall k \notin N_{i}, ~~\theta^{(i)}_{k}=0, 
$$
where $N_i$ is the neighborhood of $i$ in the graph $G$. 
Following the Bonferroni adjustment, the overall significance level $\alpha$ is divided among the tests as $\alpha = \sum_{i \in [p]} \alpha_i$.
Given a node $i$ and any nonempty index set $A \subset \left( \{i\}\cup N_i \right)^{c}$, an alternative hypothesis can be $H_{A}^{(i)}: \theta^{(i)}_{A} \neq \bs{0}$. 
The Fisher test statistic for testing $H_0^{(i)}$ versus $H_{A}^{(i)}$ for the normal linear model in \eqref{eq: locl regression} is given by 
\begin{equation}\label{eq: phi}
\phi_{A}(\mathbf{X}_i, \mathbf{X}_{-i}) := \frac{ \| \Pi_{ N_i \cup A} \mathbf{X}_i - \Pi_{N_{i}}\mathbf{X}_i \|^2 / |A| }{  \| \mathbf{X}_i - \Pi_{N_i \cup A} \mathbf{X}_i \|^2/\left( n-h_i- |A|   \right)}, 
\end{equation}
where $|A|$ is the size of $A$, $h_i$ equals to $|N_i|+1$, $\Pi_{S}$ is the orthogonal projection from $\mathbb{R}^{n}$ onto the column spaces of $\left[\bs{1}_{n}, \mathbf{X}_{S}\right]$. 
If $n<h_i+ |A|$, we define $\phi_{A}(\mathbf{X}_i, \mathbf{X}_{-i}):=0$ for simplicity. 
In practice, the statistic in \eqref{eq: phi} can be computed by fitting linear regressions via standard statistical software, such as the \texttt{anova} function in the \textbf{R} language. 

For simplicity, we define the collection of all singleton subsets of $\left( \{i\}\cup N_i \right)^{c}$ as 
\begin{equation}\label{eq: vv collection Ai}
    \mathcal{A}_{i}=\left\{~ \{a\} ~: a\in\left( \{i\}\cup N_i \right)^{c} ~ \right\}, 
\end{equation}
and we consider $A\in  \mathcal{A}_{i}$. 
To aggregate the statistics $\phi_{A}(\mathbf{X}_i, \mathbf{X}_{-i})$ for all $A\in \mathcal{A}_{i}$, \citet{verzelen_tests_2009} proposed to use a Bonferroni correction with $\sum_{A\in \mathcal{A}_i } \alpha_A^{(i)} = \alpha_{i}$ and define the test as 
$$T_{\alpha_{i}}=\sup _{ A \in \mathcal{A}_i }\left\{\phi_A\left(\mathbf{X}_i, \mathbf{X}_{-i}\right)-\bar{F}_{ 1, n-h_i-1}^{-1}\left(\alpha_A^{(i)}\right)\right\}, 
$$
where $\bar{F}_{a,b}(u)$ is the probability for a Fisher variable with $a$ and $b$ degrees of freedom to be larger than $u$.   
$\alpha_A^{(i)}$ may be set equal to $\alpha_{i}/ \left( p-h_i \right)$ for all $A\in \mathcal{A}_{i}$. 
In the numerical experiments in \citet{verzelen_tests_2009}, this testing procedure with $\alpha_{i}=0.05/p$ proved nearly as powerful as the more complicated alternative methods in most cases. 
We refer to this testing procedure as $M^1 P_1$ and will use it as one of the benchmarks for comparison in our numerical studies in Section \ref{sec: Simulation} of the main paper.

Motivated by the $M^1 P_1$ procedure, we propose to compute either the maximal value or the summation of $\phi_{A}$. 
More concretely, we define the following test statistics 
\begin{equation}\label{eq: F max}
       T_{F_{\max}}(\mathbf{X}) =\max_{i\in [p]}\max_{A\in \mathcal{A}_{i}}\phi_A\left(\mathbf{X}_i, \mathbf{X}_{-i}\right)
\end{equation}
and
\begin{equation}\label{eq: F sum}
    T_{F_\Sigma}(\mathbf{X}) =\sum_{i\in [p]}\sum_{A\in \mathcal{A}_{i}} \phi_A\left(\mathbf{X}_i, \mathbf{X}_{-i}\right), 
\end{equation}
and call them the F$_{\max}$ statistic and the F$_{\Sigma}$ statistic, respectively. 
In the simulation studies in Section \ref{sec: Simulation} of the main paper, we demonstrate that our MC-GoF test with the F$_{\Sigma}$ statistic consistently outperforms most competitors in terms of power, and thus we recommend it as the default choice. 
Besides, we expect that the performance of F$_{\max}$ is similar to the $M^1 P_1$ procedure in \citet{verzelen_tests_2009} and confirm this in an additional simulation in Appendix~\ref{app: simulation VV}. 

\begin{remark}\label{rem: m1p2}
   \citet{verzelen_tests_2009} also proposed a procedure called $M^1 P_2$, which adaptively selects the quantity $\alpha_A^{(i)}$ using a Monte Carlo method. 
   Specifically, it computes the $\alpha_i$-quantile of the following random variable conditional on $\mathbf{X}_{-i}$:
\begin{equation}\label{eq:M1P2-quantile}
    \inf_{A\in \mathcal{A}_i} \bar{F}_{1, n-h_i-1}(\phi_A(\mathbf{Z}, \mathbf{X}_{-i})),
\end{equation}
where $\mathbf{Z}$ is an independent standard normal vector of dimension $n$. The procedure then sets $\alpha_A^{(i)}$ to this quantile for all $A \in \mathcal{A}_i$.

Compared to the $M^1P_1$ procedure, $M^1P_2$ avoids splitting $\alpha_i$ into smaller components for subsets $A$.
However, the improvement in power is marginal, as shown in the simulation studies in \cite{verzelen_tests_2009}. 
Both $M^1P_2$ and $M^1P_1$ are designed to detect strong signals rather than to aggregate weak signals. Consequently, they lose power in dense-but-weak signal scenarios compared to the MC-GoF test with the F$_\Sigma$ statistic.
We examine the power properties of the $M^1 P_2$ in Appendix~\ref{app: M1P2} and confirm that its performance is only slightly better than that of $M^1P_1$ and does not compete with the MC-GoF test.

Another critical limitation of the $M^1 P_2$ procedure lies in its computational scalability. 
Computing the quantile in \eqref{eq:M1P2-quantile} requires at least $\alpha_i^{-1}$ Monte Carlo replications for each variable $i$. 
For example, with $p = 120$ and $\alpha = 5\%$, at least 2400 replications are needed per variable. 
The total number of replications grows as $\sum_{i} \alpha_i^{-1} \geq p^2 / (\sum_{i} \alpha_i) = p^2 / \alpha$ (by the generalized mean inequality).
This computational burden makes the $M^1 P_2$ procedure impractical for high-dimensional inference. 
\end{remark}

\begin{remark}
The power of the MC-GoF test is maximized when the chosen test statistic aligns with the signal pattern of the alternative hypothesis. 
We offer practical guidance to help practitioners select an appropriate test statistic based on the expected structure of the alternative hypothesis.

Motivated by the theoretical results in Section~\ref{sec: power theory} and the experimental findings in Section~\ref{sec: Simulation} of the main paper, we identify two distinct and meaningful types of alternative hypotheses:
(a) the dense-but-weak alternative, where many small nonzero entries in the precision matrix violate the null hypothesis, and
(b) the strong-but-sparse alternative, where at least one large entry in the precision matrix violates the null hypothesis.
The test statistics introduced so far can be broadly categorized into two types: 
\begin{enumerate}
    \item[(1)] \textbf{sum-type statistics} including SRC, PRC, ERC, and F$_{\Sigma}$; 
    \item[(2)] \textbf{max-type statistics} including MRC and F$_{\max}$. 
\end{enumerate}
For alternatives expected to be strong but sparse, we recommend the use of max-type statistics. 
Conversely, for alternatives expected to be dense but weak, sum-type statistics are more effective. In the case of strong and dense alternatives---a comparatively easier scenario---any of these statistics is appropriate.

\end{remark}

\subsubsection{Generalized likelihood ratio test} \label{sec: GoF test glr}

Our third strategy to construct a GoF test statistic is based on the generalized likelihood ratio (GLR) test for comparing two nested models: a null model and a more complex model (full model). 
For the GoF testing problem, we consider the full model $\mc{M}_{full}$ that comprises all $p$-dimensional normal distributions and consider testing for 
$$
H_0: P\in \mc{M}_{G_0} \mbox{~~versus~~} H_a: P\in \mc{M}_{full}. 
$$
Suppose $\ell(\bOmega)$ is the profile likelihood function that replaced the unknown mean by the sample mean. 
The GLR test statistic depends on $\ell(\bOmega)$, the maximum likelihood estimator (MLE) $\widehat{\bOmega}_{ML}$ under the full model, and the restricted MLE $\widehat{\bOmega}_{ML}^{(G_0)}$ under $\mc{M}_{G_0}$. 
The test statistic is expressed as \begin{equation}\label{eq: GLR stat}
T_{\text{GLR}}= 2 \log \ell(\widehat{\bOmega}_{ML}) - 2\log \ell(\widehat{\bOmega}_{ML}^{(G_0)}). 
\end{equation}
When $T_{\text{GLR}}$ is used in Algorithm~\ref{alg:gof}, the second term in \eqref{eq: GLR stat} can be eliminated. This is because every $\widetilde{\mathbf{X}}^{(m)}$ and $\mathbf{X}$ share the same value of the sufficient statistic for the model $\mc{M}_{G_0}$, and consequently, the same value of $\ell(\widehat{\bOmega}_{ML}^{(G_0)})$.

The computation for $\ell(\widehat{\bOmega}_{ML})$ can be challenging, especially in high-dimensional problems where the sample sizes are not sufficiently large relative to the dimensions.
In this case, one should instead use a regularized estimator. 
For example, $\widehat{\bOmega}_{ML}$ in \eqref{eq: GLR stat} can be replaced by the following modified GLasso estimator \citep{friedman_sparse_2008}:
\begin{eqnarray}\label{modified glasso}
\widehat{\bOmega}_{GL}=\argmin_{\bOmega \succ 0} \left\{  \operatorname{tr}(\hat{\bs{\Sigma}} \bOmega )-\log \operatorname{det}(\bOmega) + \lambda \sum_{(i,j)\notin \mc{E}_0, i\neq j }\sqrt{ \hat{\bs{\Sigma}}_{ii} \hat{\bs{\Sigma}}_{jj} }\left| \Omega_{ij} \right| \right\}, 
\end{eqnarray}
where the minimization is taken over all positive definite matrices with dimension $p\times p$, and $\hat{\bs{\Sigma}}$ is the sample covariance matrix. The estimator in \eqref{modified glasso} takes into account the scaling of different coordinates, which is also considered by \citet{jankova_inference_2018}. In addition, the penalty is imposed only on edges not included in the null graph $G_0$. 
For implementation in the \textbf{R} language, one can compute $\widehat{\bOmega}_{GL}$ by properly setting the argument ``rho" as a matrix in the \texttt{glasso} function. 
The tuning parameter $\lambda$ is often selected by cross-validation,  but to speed up the computation, we may simply set $\lambda=2\sqrt{\log p /n}$ based on existing theory on GLasso; see \citet[Theorem 1]{rothman_sparse_2008} and \citet[Proposition 1]{jankova_inference_2018}. 
To this end, we propose the following GLR-${\ell_1}$ test statistic:
\begin{equation}\label{eq: GLR-l1}
T_{\text{GLR}-\ell_1}= 2 \log \ell(\widehat{\bOmega}_{ML}). 
\end{equation}

\subsubsection{Incorporating prior information}\label{sec: GoF test general}

In the presence of prior information about the underlying graph, we could construct test statistic functions that incorporate this information. Although a comprehensive discussion is beyond our scope, we provide three examples of such constructions.

We begin with a scenario where we have prior information that specific pairs of variables, unconnected in the null graph $G_0=(\mc{V},\mc{E}_0)$, have high (or low) chances of connection in the actual graph. 
To incorporate this prior information quantitatively, we introduce a non-negative weighting factor $w_{ij}$ for each pair $(i,j)\notin\mc{E}_0$. 
By default, $w_{ij}$ is set to 1 and can be adjusted according to the prior information about the chance that an edge exists between nodes $i$ and $j$. 
we can then modify the PRC statistic in \eqref{eq: partial cor ss} into the following statistic:
\begin{align}
    T_{\text{PRC-w}}(\mathbf{X})&=\sum_{ \left( i,j \right) \in \Xi} w_{ij} \cdot z_{ij}^2,  \label{eq: partial cor ss weight}
\end{align}
which is called the PRC-w statistic since it is a weighted version of the PRC statistic. 
Similarly, we can define the weighted version of the ERC statistic and call it the ERC-w statistic. 
In Appendix~\ref{app: PCR_W}, we present a numerical example where these weighted statistics achieve higher power compared to the other statistics that do not incorporate prior information. 

We next consider a general Bayesian approach to construct GoF test statistics for GGMs. 
Using the framework of Bayesian inference, beliefs about which edges are likely to exist can be directly integrated into the prior distribution and will be reflected in the posterior distribution. 
A test statistic can then be formulated based on the posterior distribution, such as the posterior predictive p-value, which measures how well the Bayesian model predicts new data. See \citet{williams2020comparing} for more details about model comparison methods using posterior predictive distributions and Bayes factors. 

To illustrate how the Bayesian inference can be used to construct test statistic functions for our MC-GoF test, we consider a situation where there is prior information about some possible alternative graphs (denoted by $G^{(i)}$) that are suspected to be the true graph.  
We can specify a hyper-prior distribution $\pi(\mc{M})$ on the null model $\mc{M}_{G_0}$ and the alternative models $\mc{M}_{G^{(i)}}$. 
For each of these models, we place a prior $\gamma_{i}(\bOmega)$ on the precision matrix $\bOmega$ so that the constraints of the model $\mc{M}_{G^{(i)}}$ are fulfilled. 
Subsequently, the reciprocal of the posterior probability of the null model $\pi\left( \mc{M}_{G^{(0)}} \mid \mathbf{X}  \right)$ can serve as a test statistic.

The Bayesian approach is flexible in the incorporation of prior knowledge. 
We point out that a promising choice for the prior on the precision matrix w.r.t. a graph is the G-Wishart prior \citep{roverato2002hyper}, which is a conjugate prior and has been studied extensively in Bayesian literature; see for example \citet{atay2005monte,dobra_bayesian_2010,Mohammadi2015}. 
However, the Bayesian approach requires the posterior computation to be conducted for not only $\mathbf{X}$ but also all the $M$ copies $\widetilde{\mathbf{X}}^{(m)}$, which can be intensive in computation. We left the exploration of this method for future investigations.

%% file: AppendixPower.tex
The power analysis is based on contrasting the distribution of the observed statistic $T(\mathbf{X})$ against the empirical distribution of $T(\widetilde{\mathbf{X}}^{(m)})$, $m\in [M]$, computed from the generated copies. 
This analysis involves understanding the following two aspects: \\
(A) the distribution of copies $\widetilde{\mathbf{X}}^{(m)}$ and statistics statistics $T(\widetilde{\mathbf{X}}^{(m)})$; \\
(B) the distribution of $T(\mathbf{X})$ under a meaningful alternative hypothesis. \\
Since Algorithm~\ref{alg: exchangeable} is a Monte Carlo procedure based on Markov chains, the distribution of its output is not always analytically tractable for a general graph. To make tangible progress, we focus on the following class of null graphs whose structure enables an explicit and tractable analysis.

\begin{definition}
A graph $G=([p], \mc{E})$ is said to be clique-star shaped if there exists a subset $\mc{H}\subsetneq [p]$ such that 
$\mathcal{E}=\{(i,j): \forall  i\in [p], j\in \mc{H}\}$. 
\end{definition}
If $P\in \mc{M}_{G_0}$ and $G_0$ is clique-star shaped, then every node in the clique set $\mc{H}$ is connected to all other variables, and the variables in $\mc{H}^c$ are conditionally independent given the variables in $\mc{H}$. 
Using these two properties, we can explicitly characterize the distribution of the output of Algorithm~\ref{alg: exchangeable} and thus address the challenge in studying Aspect (A). 


Under Condition~\ref{cond: G0 clique-star} in the main paper, we denote $\mc{C}=\mc{H}^{c}$, $q=|\mc{C}|$,  and denote by $\mathbf{\Omega}_{\mc{C}}$ the submatrix of $\mathbf{\Omega}$ consisting the rows and columns indexed by $\mc{C}$. 
For normal populations, we have
$$P\in \mc{M}_{G_0} \Leftrightarrow\mathbf{\Omega}_{\mc{C}}\text{ is a diagonal matrix}.$$ 
This equivalence motivates the following metric for quantifying the deviation of the population $P$ from the null hypothesis:
$$
D(\mathbf{\Omega}_{\mc{C}}, s )= \|\mathbf{R}(\mathbf{\Omega}_{\mc{C}}^{-1})-\mathbf I_q\|_s,
$$ 
where $\mathbf{R}(\mathbf M)=\text{diag}(\mathbf M)^{-1/2}\mathbf{\mathbf M}\text{ diag}(\mathbf M)^{-1/2}$ is the diagonal normalization and the symbol $s$ can be either $F$ or $\infty$, corresponding to the Frobenius norm or the maximum norm of a matrix, respectively.
The metric $D(\mathbf{\Omega}_{\mc{C}}, s)$ serves as a measure of signal strength under the alternative hypothesis. 
In particular, for both $s=F$ and $s=\infty$, we have $D(\mathbf{\Omega}_{\mc{C}}, s)=0$ if and only if $\mathbf{\Omega}_{\mc{C}}$ is a diagonal matrix, which indicates no deviation from the null hypothesis. 
Appendix~\ref{sec:power-metric} provides a detailed discussion on this metric. 
As we will see, the power of the MC-GoF test depends on the magnitude of $D(\mathbf{\Omega}_{\mc{C}}, s)$.

We now provide an overview of our theory as follows:
\begin{itemize}
\item  
Dense alternative:  
When many of the off-diagonal entries of $\mathbf{\Omega}_{\mc{C}}$ are nonzero, we use $D(\mathbf{\Omega}_{\mc{C}}, F)$ to quantify the deviation from the null hypothesis. 
In Appendix~\ref{sec:dense alternative}, we show that when both $n$ and $p$ diverge to infinity proportionally, if $D(\mathbf{\Omega}_{\mc{C}}, F)\gg \sqrt{q/n}$, then the asymptotic power of the MC-GoF test with some statistic $T_1$ converges to $1$. 
Additionally, we show that no $\alpha$-level test can distinguish between the null and alternative hypotheses with power tending to one when the separation rate in terms of $D(\mathbf{\Omega}_{\mc{C}}, F)$ is of order $\sqrt{q/n}$. 
This result implies that the MC-GoF test with $T_1$ is rate-optimal for dense alternatives.

\item Strong alternative: 
When some of the off-diagonal entries of $\mathbf{\Omega}_{\mc{C}}$ are separated from zero, we use $D(\mathbf{\Omega}_{\mc{C}}, \infty)$ to quantify the deviation from the null hypothesis.
In Appendix~\ref{sec:strong alternative},  we show that when $D(\mathbf{\Omega}_{\mc{C}}, \infty)=16\sqrt{\log q/n}$, the power of the MC-GoF test with some statistic $T_2$ converges to 1 as $n\rightarrow\infty$. Additionally, we show that any test will be powerless if the separation rate in terms of $D(\mathbf{\Omega}_{\mc{C}}, \infty)$ is of order smaller than $\sqrt{\log q /n}$. 
This result implies that the MC-GoF test with $T_2$ is rate-optimal for strong alternatives. 

\item 
Dense or strong alternative: 
The union of the aforementioned alternative hypotheses can also be tested using the MC-GoF test with a properly designed test statistic $T_3$. 
In Section~\ref{sec:union-alternative}, 
we show that when either $D(\mathbf{\Omega}_{\mc{C}}, F)\gg \sqrt{q/n}$ or $D(\mathbf{\Omega}_{\mc{C}}, \infty)\geq 16\sqrt{\log q/n}$, the power of the MC-GoF test with $T_3$ converges to 1. Again, this test is rate-optimal for the union alternative. 
\end{itemize}
We provide a proof sketch in Section~\ref{sec:power-proof-sketch} and place the complete proof in Appendix~\ref{app: power proof scheme}. 

\subsection{Metric of deviation}\label{sec:power-metric}

Before we study the power, we provide some intuitions about the metric $D(\mathbf{\Omega}_{\mc{C}}, s)$, which is the key quantity for characterizing the power of the MC-GoF test as illustrate in Figure~\ref{fig:phase_transition} in the main text. 
To connect this metric with some common population quantities, we provide the following result. 

\begin{proposition}\label{prop: Omega related to R}
    Suppose $X=(X_1,\ldots, X_p)^\top \sim \mathbf{N}_p(\mathbf{0}, \mathbf{\Omega}^{-1})$. 
    Let $\mc{C}$ be a nonempty subset of $[p]$ and let  $\mc{H}=[p]\setminus\mc{C}$. Then we have
    $$
    D(\mathbf{\Omega}_{\mc{C}}, \infty)=\max_{i,j\in \mc{C}} ~~ \left|\operatorname{cor}\left(X_i, X_j\mid X_{\mc{H}}\right)\right|. 
    $$
    Furthermore, $$
   D(\mathbf{\Omega}_{\mc{C}}, F)^2\geq  \frac{1}{\lambda} \sum_{i,j\in \mc{C}}\mathbf{\tilde{\Omega}}_{i,j}^2, 
    $$
    where $\lambda$ is the condition number of $\mathbf{\Omega}_{\mc{C}}$ and $\mathbf{\tilde{\Omega}}_{\mc{C}}=\mathbf{R}(\mathbf{\Omega}_C)$. 
\end{proposition}
Proposition~\ref{prop: Omega related to R} suggests that $D(\mathbf{\Omega}_{\mc{C}}, \infty)$ measures the largest magnitude of the partial correlations that violate the null hypothesis, and the square of $D(\mathbf{\Omega}_{\mc{C}}, F)$ serves as a proxy for the sum of squared entries of the normalized $\mathbf{\Omega}_{\mc{C}}$.



\subsection{Dense alternative}\label{sec:dense alternative}
We first consider the case where the true precision $\mathbf{\Omega}$ deviates from the model assumption of the null hypothesis with many small entries. 
In such a case, we consider the SRC statistic in \eqref{eq: sum residual correlation}. 
Under Condition~\ref{cond: G0 clique-star}, this statistic is expressed as 
$$T_1(\mathbf X)= \sum_{i,j \in \mc{C}, i\neq j}\hat{\gamma}_{ij}^2,$$
where $\hat{\gamma}_{ij}$ is the pairwise residual correlation statistic defined in \eqref{residual correlation}. 
To examine the power, we use the metric $D(\cdot, F)$ to quantify the deviation from the null hypothesis and define 
$$
\mathbf{\Theta}_{n1}(b):=\left\{\mathbf{A}\in \mathbb{S}_+^p: D(\mathbf{A}_C, F) \geq b\sqrt{\frac{q}{n}},\quad  \lambda_{\max}(\mathbf{A}_C)/\lambda_{\min}(\mathbf{A}_C)\leq b_0\right\},  \quad \forall b>0, 
$$
where $b_0\geq 1$ can be any fixed constant. 
We are interested in the dense alternative hypothesis defined as
\begin{equation}\label{eq:dense alternative}
    H_a: \mathbf{\Omega}\in \mathbf{\Theta}_{n1}(b). 
\end{equation}

\begin{theorem}\label{thm:power_dense}
Suppose that Condition~\ref{cond: G0 clique-star} holds and that $q/ n  \rightarrow \gamma \in(0, \infty)$. 
For any $\epsilon>0$, there exists a positive constant $b$ such that for any sequence of populations with $\mathbf{\Omega}\in \mathbf{\Theta}_{n1}(b)$, it holds 
   $$
  \liminf_{n} \beta_{n}(P;\alpha, T_1)\geq 1-\epsilon,
   $$ 
provided that the number of copies $M$ is at least $\max(2\alpha^{-1},\log(2\epsilon^{-1}))$. 
\end{theorem}

\begin{remark}
In the proof, we specifically choose  $b=\sqrt{2z_1 + 4 b_0^4 \gamma z_2}$, where $z_1$ and $z_2$ are the ($\alpha^2/16$)- and ($\epsilon/2$)-upper quantiles of the standard normal distribution respectively. 
In other words, $b$  depends on $\epsilon$ at the order of $\sqrt{\log(\epsilon^{-1})}$. 
\end{remark}

Theorem~\ref{thm:power_dense} shows that the MC-GoF test can distinguish between the null \eqref{eq: GoF null} and the dense alternative \eqref{eq:dense alternative} with power tending to 1 when $b=b_n\rightarrow\infty$. 
In particular, we have the following corollary on the consistency of the MC-GoF test against the dense alternative.

\begin{corollary}\label{cor:dense consistent}
Suppose that Condition~\ref{cond: G0 clique-star} holds and that $q/n\rightarrow  \gamma \in(0, \infty)$. 
For a sequence of populations such that the condition number of $\mathbf{\Omega}_{\mc{C}}$ is bounded and 
$D(\mathbf{\Omega}_{\mc{C}}, F)$ diverges to infinity, 
if the MC-GoF test is conducted with level $\alpha_n$ at least $1/\log(n)$ and with $M$ at least $2\max(\alpha_n^{-1},\log n)$, then $$\lim_{n\rightarrow\infty}\beta_{n}(P; \alpha_n , T_1) = 1.$$ 
\end{corollary}

Although Theorem~\ref{thm:power_dense}  and Corollary \ref{cor:dense consistent} require $q$ to grow proportionally in $n$, they allow the limit of $q/n$ to be much larger than $1$ so they are applicable in high-dimensional settings where the dimension is much larger than the sample size. When $q/n$ is bounded, the result in Corollary~\ref{cor:dense consistent} cannot be improved.  
Specifically, the following theorem proves that if $D(\mathbf{\Omega}_{\mc{C}}, F)$ is at the order of $\sqrt{q/n}$, then no $\alpha$-level test for $H_0: P\in \mc{M}_{G_0}$ can differentiate between the null hypothesis and the dense hypothesis with power tending to one as $n$ and $q$ grow. 

\begin{theorem}\label{thm: lower bound dense}
Let $0<\alpha<\beta<1$. 
Suppose that as $n \rightarrow \infty, q \rightarrow \infty$ and that Condition~\ref{cond: G0 clique-star} holds. Suppose $q / n \leq \kappa$ for some constant $\kappa<\infty$ and all $n$. 
Then there exists a constant $b=b(\kappa, \beta-\alpha)$ such that
$$
\limsup _{n \rightarrow \infty}~~\left\{ \sup_{\phi} \inf_{\mathbf{\Omega} \in \mathbf{\Theta}_{n1}(b)} \mathbb{E}(\phi) \right\}< \beta.
$$
where $\sup_{\phi}$ is taken over any $\alpha$-level test $\phi$ for $H_0: P\in \mc{M}_{G_0}$. 
\end{theorem}

Theorem~\ref{thm: lower bound dense} shows that no level $\alpha$ test for \eqref{eq: GoF null} can distinguish between the two hypotheses with power tending to 1 as $n$ and $p$ grow when the separation rate $\varepsilon_n$ is of order $\sqrt{q / n}$.
Thus, it provides a lower bound on the separation rate. 
Comparing with the power guarantee given in Theorem~\ref{thm:power_dense}, we see that the MC-GoF test with $T_1$ is rate-optimal for dense alternatives.

\subsection{Strong alternative}\label{sec:strong alternative}
We consider the case where the true precision $\mathbf{\Omega}$ deviates from the model assumption of the null hypothesis with at least one large entry. 
In such a case, we choose the following test statistic
$$T_2(\mathbf X)= \max_{i,j\in \mc{C}, i\neq j}\hat{\gamma}_{ij}^2,$$
where $\hat{\gamma}_{ij}$ is the pairwise residual correlation statistic defined in \eqref{residual correlation}. 

To examine the power, we consider the quantity $D(\cdot, \infty)$ and define 
$$
\mathbf{\Theta}_{n2}(b)=\left\{\mathbf{A}\in \mathbb{S}_+^p: D(\mathbf{A}_C, \infty) \geq b\sqrt{\frac{\log q}{n}}\right\},
$$
and consider the following alternative:
$$H_a: \mathbf{\Omega}\in \mathbf{\Theta}_{n2}(b). $$

\begin{theorem}\label{thm:power_sparse}
Suppose that Condition~\ref{cond: G0 clique-star} holds and that $\log q/n\rightarrow 0$.
For any $\epsilon>0$, 
if the sample size $n>\max(10/\alpha, 32\log(16/\eps),8/\eps)$ and 
$M>\max(2\alpha^{-1},\log(2\epsilon^{-1})$, then 
   $$
  \inf_{\mathbf{\Omega}\in \mathbf{\Theta}_{n2}(16)} \beta_{n}(P;\alpha, T_2)\geq 1-\epsilon.
   $$ 
\end{theorem}

\begin{corollary}\label{cor:sparse-consistent}
Suppose that Condition~\ref{cond: G0 clique-star} holds and that $\log q/n\rightarrow 0$.
If the MC-GoF test is conducted with level $\alpha_n$ at least $1/\log(n)$ and with $M$ at least $2\max(\alpha_n^{-1},\log n)$, then $$\lim_{n\rightarrow\infty}\inf_{\mathbf{\Omega}\in \mathbf{\Theta}_{n2}(16)} \beta_{n}(P; \alpha_n , T_2) = 1.$$

\end{corollary}

The following theorem proves that if $D(\mathbf{\Omega}_{\mc{C}}, \infty)$ is not sufficiently large compared to $\sqrt{\log (q)/n}$, then any test for $H_0: P\in \mc{M}_{G_0}$ cannot differentiate between the null hypothesis and the sparse hypothesis as $n$ and $q$ grow. 

\begin{theorem}\label{thm: lower bound sparse}
Let $0<\alpha<1$. 
Suppose that Condition~\ref{cond: G0 clique-star} holds and that as $n \rightarrow \infty, q \rightarrow \infty$ and that $\log (q) / n \leq \kappa$ for some constant $\kappa<\infty$ and all $n$. 
Then for any constant $b\in (0,  \min(1,\kappa^{-1}))$, the following holds:
$$
\limsup_{n \rightarrow \infty}~~ \left\{ \sup_{\phi} ~~ \inf_{\mathbf{\Omega} \in \mathbf{\Theta}_{n2}(b)} \mathbb{E}(\phi) \right\}\leq \alpha, 
$$
where $\sup_{\phi}$ is taken over all $\alpha$-level test for $H_0: P\in \mc{M}_{G_0}$. 
\end{theorem}

Theorem~\ref{thm: lower bound sparse}  provides a lower bound on the separation rate between the null hypothesis and the sparse alternative hypothesis required by any consistent test. 
Together with the power guarantee in Theorem~\ref{thm:power_sparse}, we see that the MC-GoF test with $T_2$ is rate-optimal for sparse alternatives.

\subsection{Union Alternative}\label{sec:union-alternative}

In this section, we consider the case where the true precision $\mathbf{\Omega}$ has either many small entries or has at least one large entry but not both. In other words, we are interested in the union alternative defined as  
$$
H_a: \mathbf{\Omega}\in \mathbf{\Theta}_{n3}(b)=\mathbf{\Theta}_{n1}(b)\cup \mathbf{\Theta}_{n2}(b).
$$
For this alternative, we propose the following test statistic $T_3$ that incorporates both $T_1$ and $T_2$: 
$$
T_3 (\mathbf{X})=\max\left( \frac{n}{2q} \left(T_1(\mathbf{X})-q(q-1)/n\right),  nT_2(\mathbf{X})-4\log q+\log\log q  \right).
$$

\begin{theorem}\label{thm:power_union}
Suppose that Condition~\ref{cond: G0 clique-star} holds and $q/n\rightarrow\gamma\in (0, \infty)$. 
For any $\epsilon>0$, there exists a positive constant $b$ such that for any sequence of populations with $\mathbf{\Omega}\in \mathbf{\Theta}_{n3}(b)$, it holds 
   $$
  \liminf_{n} \beta_{n}(P;\alpha, T_3)\geq 1-\epsilon,
   $$ 
provided that the number of copies $M$ is at least $\max(2\alpha^{-1},\log(2\epsilon^{-1})$. 
\end{theorem}

The following theorem is a corollary of Theorem~\ref{thm: lower bound dense} and Theorem~\ref{thm: lower bound sparse}, which shows that the MC-GoF test with $T_3$ is rate-optimal for the union alternative. 

\begin{theorem}\label{thm: lower bound union}
Let $0<\alpha<\beta<1$. 
Suppose that Condition~\ref{cond: G0 clique-star} holds and that as $n \rightarrow \infty, q \rightarrow \infty$ and  $q / n \leq \kappa$ for some constant $\kappa<\infty$ and all $n$. 
Then there exists a constant $b=b(\kappa, \beta-\alpha)$ such that
$$
\limsup _{n \rightarrow \infty}~~\left\{ \sup_{\phi} \inf_{\mathbf{\Omega} \in \mathbf{\Theta}_{n3}(b)} \mathbb{E}(\phi) \right\}< \beta.
$$
where $\sup_{\phi}$ is taken over any $\alpha$-level test $\phi$ for $H_0: P\in \mc{M}_{G_0}$. 
\end{theorem}

\subsection{Proof sketch for upper bounds}\label{sec:power-proof-sketch}

In this section, we provide an overview of our proof of the power guarantee of the MC-GoF test. 
For each of the tests being considered, the proof will follow the following three steps:
\begin{enumerate}
    \item[] \textbf{Step 1.} We first bound the power from below with some constant $t$ as 
    $$\mathbb P(\text{Reject }H_0)\geq \mathbb P(T(\mathbf{X}) >t) -\mathbb P \left( \frac{1}{M+1}\left[1+\sum_{m=1}^M \mathbf{1}\left\{T\left(\widetilde{\mathbf{X}}^{(m)}\right) \geq t\right\}\right]>\alpha \right),$$
    where the value of $t$ depends on the choice of the test statistic and the setting of the alternative hypothesis. 
    
    \item[] \textbf{Step 2.} Next, we show that  the distribution of the statistics $T(\mathbf{X})$ and $T(\mathbf{\tilde{X}}^{(m)})$ can be characterized using the  distribution of $T(\mathbf{U})$ and $T(\mathbf{\tilde{U}}^{(m)})$, where
    $\mathbf{U}$ is a $(n-1-|\mc{H}|)\times q$ matrix with i.i.d. rows sampled from $N_{q}(\mathbf{0}, \mathbf{R}(\mathbf{\Omega}_{\mc{C}}^{-1}))$, and $\mathbf{\tilde{U}}^{(m)}$ is a $(n-1-|\mc{H}|)\times q$ matrix with i.i.d. rows sampled from $N_{q}(\mathbf{0}, \mathbf{I}_{q})$.
    \item[] \textbf{Step 3.} 
    Based on the characterization in the last step, we derive an upper bound on the probability $\mathbb P \left( \frac{1}{M+1}\left[1+\sum_{m=1}^M \mathbf{1}\left\{T\left(\widetilde{\mathbf{U}}^{(m)}\right) \geq t\right\}\right]>\alpha \right)$ for each choice of statistic $T$. 
    Furthermore, we will analyze the asymptotic property of $T(\mathbf{U})$ to derive a tight lower bound on $\mathbb P(T(\mathbf{U}) >t)$.  
\end{enumerate}
The details of the proof are provided in Appendix~\ref{app: power proof scheme}. 
In the proof, Step 3 is the most important and it requires more technical analyses on the asymptotic properties of sample correlation statistics. 
Step 2 is proved by using the properties of Algorithm~\ref{alg: exchangeable}. 
This step is essential because it provides a tractable characterization that allows us to study the asymptotic properties of the MC-GoF test using the techniques developed for studying high-dimensional sample correlation tests.

Our theoretical results are related to high-dimensional hypothesis testing for correlation matrices \citep{zheng2019test,gao2017high,cai2011limiting}. 
Specifically, when the null graph $G_0$ does not have any edge, $\mc{H}=\emptyset$ and testing whether $P\in \mc{M}_{G_0} $ is equivalent to testing whether the correlation matrix equals to the identity matrix. 
In this scenario, our results imply that the MC-GoF test can achieve rate-optimal power, which matches the test for correlation matrices proposed in \cite{zheng2019test}.

%% file: AppendixSimulation.tex
We compare the power performance of our method against two baseline methods: the \nameVV procedure in \citet{verzelen_goodness--fit_2010} and the Bonferroni adjustment of multiple testing for the composite null hypothesis in \eqref{eq: hypotheses omega} as suggested in \citet{drton_multiple_2007}. These two methods are referred to as \nameVV and \nameDP. 
Specifically, the Bonferroni-adjusted p-value of the \nameDP method is defined as 
\begin{equation}\label{eqn:bonf}
 \min\left\{ 1,  K  p_{ij} : i < j,  (i,j)\notin \mc{E}  \right\},
 \end{equation}
where $K$ equals to the number of pairs of distinct $i$ and $j$ that are not connected in $G$, and $p_{ij}$ is the statistic defined in \eqref{eq:pcr pvalue exact}.

The MC-GoF test (Algorithm~\ref{alg:gof}) will be conducted using four test statistic functions discussed in Section~\ref{rem: stat}, including $\text{PRC}$ and $\text{ERC}$ in Appendix~\ref{sec: GoF test partial cor}, $\text{F}_{\Sigma}$ in Appendix~\ref{sec: gof test F}, GLR-$\ell_1$ in Appendix~\ref{sec: GoF test glr}. For simplicity, we fix the procedure parameters $M=100$ and $L=3$.

In Appendix~\ref{sec: GoF simulation setup}, we first explain the experimental settings. Appendix~\ref{app: Type-I error} demonstrates the theoretically valid Type-I error control for all methods. Then in Appendix~\ref{sec: GoF simulation power}, we compare power of each method for each setting with different configurations of population parameters. Appendix~\ref{app: M1P2} presents the numerical results of a baseline method, the $M^1P_2$ procedure, which performs similarly to the \nameVV procedure but incurs significantly higher computational costs as discussed in Remark~\ref{rem: m1p2}.

\subsection{Setup}\label{sec: GoF simulation setup}

For illustration, we consider three types of graphs, each with distinct topological properties. 
In each of the following items, we define the precision matrix $\bs{\Omega}=(\omega_{ij})_{i,j\in [p]}$ and the corresponding graph $G$ for the true distribution of the data, and define another graph $G_0$ for the null hypothesis.

\begin{itemize}
    \item \textbf{Band Graph}: The precision matrix $ \boldsymbol{\Omega} $ satisfies that for $i,j \in \left\{ 1, \cdots, p \right\}$:
$$ \omega_{i, j}=\left\{
\begin{array}{rcl}
1       &      & \text{if ~~} i=j, \\
s     &      & \text{if ~~} 1\leq |i-j|\leq K,\\
0       &      & \text{if ~~} |i-j|> K,
\end{array} \right. $$
where $K$ is the bandwidth and $s\in (0,0.2)$ is the signal magnitude.
We define the graph $G_0$ for a null hypothesis as a band graph with bandwidth $K_0$. 

    \item \textbf{Hub Graph}: The nodes are divided into separate groups of size $h=10$, and within each group, members are only connected to a hub node. Specifically,
    for each positive integer $k \leq p/h$, let $i=h (k-1) +1$, $\omega_{ij}=\omega_{ji}=1$ for $j \in \left\{ (i+1), \cdots, (i+h-1) \right\}$, and let all other off-diagonal entries equal to zero. 
    Then each of the diagonal entries is set to be the summation of absolute values along the row plus a parameter $\xi>0$: 
      $$
      \omega_{ii} = \sum_{i\neq j} |\omega_{ij}| + \xi, $$
    which guarantees that $\bOmega=(\omega_{ij})_{1\leq i,j\leq p}$ is positive definite. The parameter $\xi$ controls the signal magnitude: a larger value of $\xi$ corresponds to weaker signals.

    To define the graph for a null hypothesis, we consider a subgraph $G_0$ by randomly removing the existing edges independently with probability $0.7$.

    \item \textbf{\nameERG~Random Graph}: 
    There is an edge between a pair of distinct nodes independently and randomly with probability $q$. 
    We begin with constructing a random matrix $\bs{A}=(a_{ij})_{1\leq i,j\leq p}$. 
    Let $a_{ii}=1$, for $i \in \left\{ 1, \cdots, p \right\}$. 
    For any pair of $i<j$, let $a_{ij}=a_{ji}= u_{ij} \cdot \delta_{ij}$, where $u_{ij}$ is uniformly distributed on $(s/2, 3 s/2)$ and $\delta_{ij}$ is an independent Bernoulli random variable with success probability $q$. 
   Let $\lambda$ be the smallest eigenvalue of $\bs{A}$. The precision matrix is defined as 
   $$\bs{\Omega} = \bs{A}+\left( |\lambda| + 0.05 \right) \bs{I}.$$ 
   To define the graph for a null hypothesis, we consider a subgraph $G_0$ whose edges exist with probability $q_0$. This can be obtained by randomly removing the existing edges in $G$ independently with probability $1-q_0/q$. In our experiment, we fix $q_0=0.08$. 
\end{itemize}

In the construction of a band graph, the population parameter $s$ will be referred to as the \textit{signal magnitude}. This is because if the constructed normal distribution does not belong to $\mc{M}_{G_0}$, then the result given by \citet[Corollary 1]{jog_model_2015} ensures that the Kullback-Leibler (KL) divergence between this distribution and any distribution in $\mc{M}_{G_0}$ is at least $\eta_{-}(s) := - 2^{-1} \log \left(1-s^2\right)$, which is an increasing function of $s^2$.
For a hub graph model and an \nameERG random graph, similar lower bounds for the KL divergence between the true distribution and any distribution in the null model can be obtained, which are $\eta_{-}\left( 1/(h-1+\xi)  \right)$ and $\eta_{-}\left( s/(1+|\lambda|+0.05)  \right)$, respectively. Therefore, we will vary the value of $s$ (or $\xi$ for hub graph models) in the simulations to control the KL divergence between the true distribution and the null model. 

The simplicity of these graphs allows us to illustrate how the signal patterns affect the power of the tests. 
Appendix~\ref{app: more graphs} extends the experiments to more complex graph structures, including cases where the null graph is not a subgraph of the true graph. Additionally, Appendix~\ref{app: L} includes an example with moderately dense null graphs, where the graph degree is not small compared to the sample size. 
The results from these extended experiments agree with the findings presented in this section.

\subsection{Type I error control}\label{app: Type-I error}

We numerically illustrate that the proposed MC-GoF test controls the Type I error regardless of the dimension and the sample size. 
Specifically, we generate random samples from the GGM with mean 0 and precision matrix from the band graph $G$ with $K=6$ and $s=0.2$ as defined in Appendix~\ref{sec: GoF simulation setup} and conduct various GoF tests for the GGM w.r.t. $G$ --- the null hypothesis is true in this case. 
We vary the dimension $p$, the signal magnitude $s$, and the sample size $n$ for several configurations. 
We independently repeat each experiment 1200 times and record the p-values given by each method. 
We then estimate the size of each test at the significance level $\alpha=0.05$ by its rejection proportion across the replications.

We report the estimated sizes of the tests in Table~\ref{tab: GoF error control}. The first column (\nameVV) corresponds to the ${M^1}{P_1}$ procedure in \citet{verzelen_goodness--fit_2010}, the second (\nameDP) is a Bonferroni adjustment of multiple testing for partial correlation in \citet{drton_multiple_2007}.  
The other four columns correspond to our proposed MC-GoF testing procedure with different choices of test statistics in Section~\ref{sec: GoF test statistic}.

For each setting and each method, we also conduct the one-sample proportion test for the hypothesis that the size is no greater than $\alpha$ and report the p-value in a pair of parentheses. 
The smallest p-value appears in the result of the MC-GoF test using the GLR-$\ell_1$ statistic for the configuration with $(p,s,n)=(20, 0.2, 80)$. 
In this configuration, we further analyze the p-values from this GoF test and present the QQ plot to compare them with the uniform distribution in Figure~\ref{fig: qqplot-null-glr}. 
Note that Proposition~\ref{prop: GoF valid} only justifies the validity of any MC-GoF test, but the p-value could be stochastically larger than a uniform random variable. Nonetheless, the QQ plot in Figure~\ref{fig: qqplot-null-glr} shows that the distribution of the simulated p-values is very close to the uniform distribution. The exceedance of some estimated sizes in Table~\ref{tab: GoF error control} should be attributed to Monte Carlo errors and multiple comparisons.

\begin{table}[!btph]
    \caption{Estimated sizes of various GoF tests at the significance level $\alpha=0.05$ on band graphs ($K=6, K_0=6$). The p-value of the one-sample proportion test for each size is provided in parentheses. } \label{tab: GoF error control}
\input{table-0-p=20}
\end{table} 

\begin{figure}[hbtp]
    \centering
    \caption{QQ plot comparing the uniform distribution and the p-value of the MC-GoF test using the GLR-$\ell_1$ statistic in the setting $(p,s,n)=(20, 0.2, 80)$. }\label{fig: qqplot-null-glr}
\includegraphics[width=0.45\textwidth]{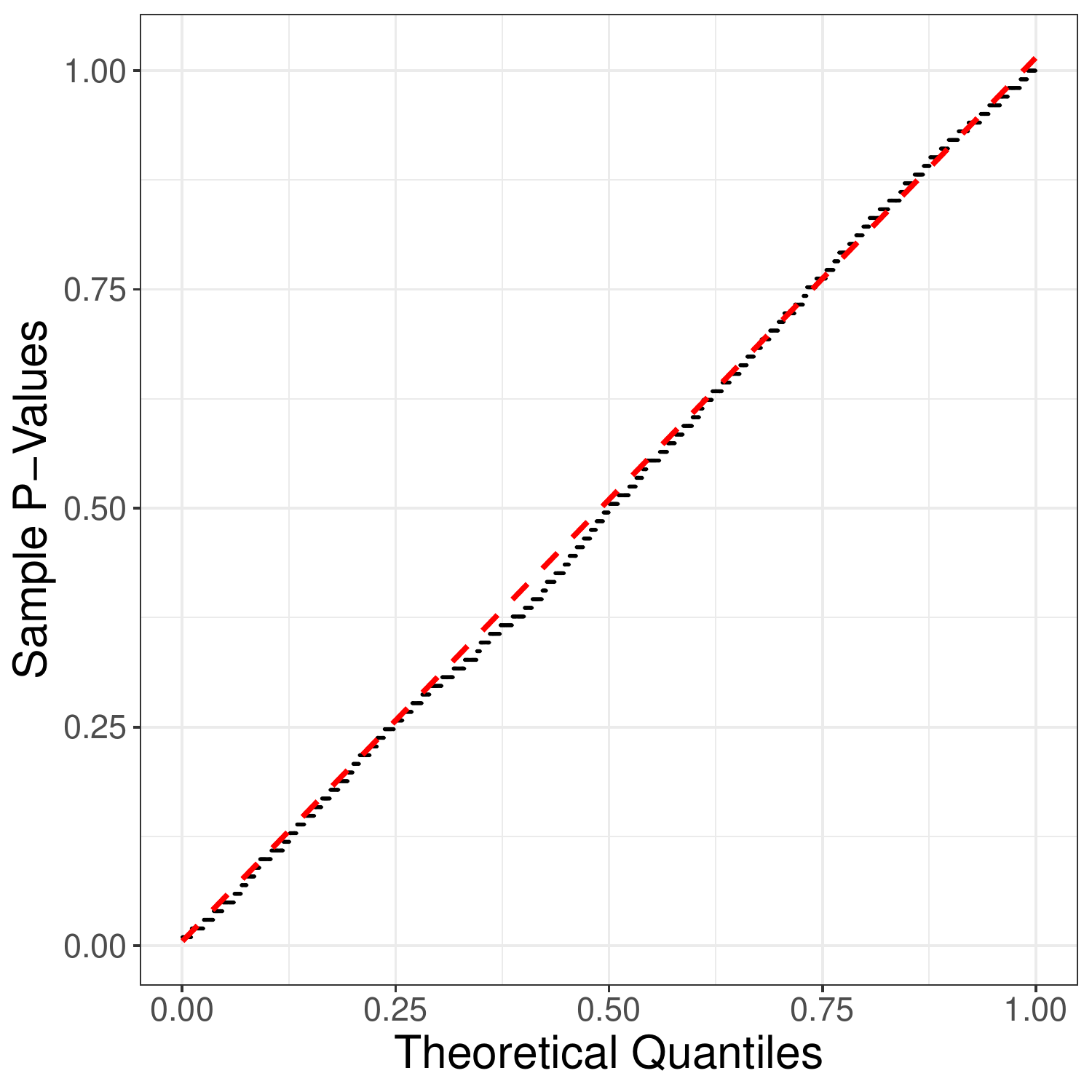}
\end{figure}

The MC-GoF test is guaranteed to control the Type I error when the true graph $G$ is a subgraph of $G_0$. As an example, we consider band graphs with $K = 4$ for $G$ and $K_0 = 6$ for $G_0$, setting $s = 0.2$ as defined in Appendix~\ref{sec: GoF simulation setup}. Table~\ref{tab: subset} presents the estimated sizes of the test. The MC-GoF test controls the sizes at the significance level $\alpha = 0.05$ in all cases, except for tests using PRC or ERC statistics in the configuration $(p, s, n) = (20, 0.2, 80)$---again, the deviations are likely caused by Monte Carlo errors and multiple comparisons.

\input{GoF_Subset_G}

\subsection{Power comparison}\label{sec: GoF simulation power}
We consider each of three graph types in Appendix~\ref{sec: GoF simulation setup} with different configurations of population parameters. 
In each configuration, we generate $n$ random samples from $\mathbf{N}_p( \bs{0}, \boldsymbol{\Omega}^{-1})$ and conduct various GoF tests for the null hypothesis $H_0: P\in \mc{M}_{G_0}$ at the significance level $\alpha=0.05$. 
We repeat each experiment 400 times and present the estimated power for different methods in tables. 
Standard errors are shown in parentheses, and the highest power in each case is indicated in bold type.

Table \ref{tab: band} presents the results for band graphs with $K=6$ for the true distribution and $K_0=1$ for the null hypothesis. 
The dimension $p$ is either 20 or 120, the sample size $n$ takes value among 20, 40, and 80, while the signal magnitude $s$ among 0.1, 0.15, and 0.2. 
As expected, the power of each test increases with the sample size $n$ and the signal magnitude $s$. 
We also observed that when $s$ is as large as $0.2$, F$_{\Sigma}$ consistently achieves the highest power in Table \ref{tab: band}. 
When $s$ is $0.1$ or $0.15$, PRC sometimes slightly outperforms F$_{\Sigma}$.  
ERC, designed as a computationally efficient alternative to PRC, has slightly lower power than PRC when $n$ is small, but the gap diminishes quickly as $n$ increases. 
The GLR-$\ell_1$ statistic does not perform as well as the other three statistics when $n$ is $40$. 
This may be due to the violation of the sparsity assumption for GLasso to work well since the degrees of nodes in this band graph are between $6$ and $12$, which are not negligible relative to $n=40$. 
Our GoF testing procedure with any of the four test statistics outperforms the baseline methods across all cases studied in this experiment.  
In some cases the advantages are significant; for example, when $p=120$, $s=0.15$, and $n=80$, all the four MC-GoF tests achieve power higher than $90\%$, while \nameVV and \nameDP have power not exceeding $15\%$. 
We also note that for any pair of fixed $s$ and $n$, increasing the dimension $p$ could boost the power of our MC-GoF test with F$_{\Sigma}$. In contrast, both \nameVV and \nameDP lose power when the dimension $p$ increases from $20$ to $120$. 
This indicates that our proposed method is more applicable to high-dimensional problems than existing methods.

\begin{table}[!btph]
    \caption{Power of various GoF tests on band graphs ($K=6, K_0=1$) at the significance level $\alpha=0.05$. The first two columns correspond to the baseline methods and the last four columns correspond to the MC-GoF test with four test statistic functions. Standard errors are in parentheses and the highest power in each case is in bold type.}\label{tab: band}
\input{table-1-K=6}
\end{table} 

Table \ref{tab: hub} shows the results for hub graphs. 
The dimension $p$ is either 20 or 120, the sample size $n$ takes value among 20, 40, and 80, and the noise level $\xi$ among 0.9, 1.4, and 2. 
A larger value of $\xi$ corresponds to a weaker deviation from the null hypothesis. 
We observe that F$_{\Sigma}$ and GLR-$\ell_1$ achieve the highest power across all settings. 
In general, the performance of GLR-$\ell_1$ tends to be better than F$_{\Sigma}$ with a large value of $n$ or a small value of $\xi$. 
Again, all our methods outperform the existing methods in every case, especially when the dimension $p=120$.  
In particular, when $p=120$, $n=80$ and $\xi=2$, the power of either \nameVV or \nameDP is around $20\%$ whereas our MC-GoF test using either F$_{\Sigma}$ or GLR-$\ell_1$ could achieve power as high as $70\%$.

\begin{table}[!btph]
\caption{Power of various GoF tests on hub graphs ($\alpha=0.05$).}\label{tab: hub}
\input{table-2}
\end{table} 

Table \ref{tab: ER} presents the results for \nameERG~random graphs. 
We vary the connection probability $q$ of the true graph $G$ between $0.2$ and $0.4$, and $s$ between $0.01$ and $0.02$. 
The dimension $p$ is either $40$ or $120$ and the sample size $n$ is either $50$ or $100$. 
We observe that F$_{\Sigma}$ performs consistently the best across all settings. 
In the high-dimensional cases, GLR-$\ell_1$ has the second best power, with a significant advantage over the rest four tests. 
Both tests become more powerful when the dimension $p$ increases from $40$ to $120$. In contrast, the baseline methods \nameVV and \nameDP suffer from a loss of power in the high-dimensional cases. Notably, when $p=120$,  $q=0.4$, $s=0.01$, and $n=50$, the power of F$_{\Sigma}$ is as high as $94\%$ while \nameVV and \nameDP have power lower than $10\%$.

\begin{table}[!btph]
\caption{Power of various GoF tests on  \nameERG random graphs ($\alpha=0.05$).}\label{tab: ER}
\input{table-3}
\end{table} 

For each graph type, we identified specific situations where the baseline methods, \nameVV and \nameDP, perform poorly while our method could achieve high power. 
These situations occur when the true precision matrix has many nonzero entries that are absent from the null model and the magnitudes of these entries are not large enough for \nameVV or \nameDP to reject the null hypothesis. 
In our experiments, we control these magnitudes using the parameter $s$ in band graphs and \nameERG random graphs,  and the inverse of $\xi$ in hub graphs. 
We refer to this pattern of deviation from the null model as \textit{dense-but-weak} pattern. 
Conversely, the \textit{strong-but-sparse} pattern occurs when the null model only misses a few entries of the precision matrix, but these entries are large enough for the baseline methods to reject the null hypothesis. Additional simulations in Appendix~\ref{app: simulation VV} confirm that our tests remain competitive and powerful in situations where the strong-but-sparse pattern occurs. 

By comparing these two empirical patterns to the theoretical regimes identified in Section~\ref{sec: power theory} of the main paper, we see a clear alignment between practice and theory. 
The theoretical results on the rate-optimality predict that the MC-GoF test is powerful when the signal pattern is either \textit{dense but weak} or \textit{strong but sparse}, which is confirmed by the empirical studies. 
Although our theoretical framework in Section~\ref{sec: power theory} focuses on clique-star-shaped graphs, the insights gained from the theory extend to a broader range of graph structures.

In summary, the MC-GoF test achieves higher power than the baseline methods when the signal pattern is dense and weak, and it remains competitive with the \nameVV procedure when the signal pattern is strong but sparse. 
Across all the studied settings, the MC-GoF test with the F$_{\Sigma}$ statistic consistently has superior power performance and is thus a promising method for GoF testing in high-dimensional GGMs.

\subsection{Performance of $M^1P_2$ in \cite{verzelen_tests_2009}} \label{app: M1P2}

In this section, we evaluate the power properties of the $M^1P_2$ procedure introduced by \citet{verzelen_tests_2009}. 
This analysis complements the discussion in Remark~\ref{rem: m1p2}, where we highlighted the marginal power improvement of $M^1P_2$ over $M^1P_1$ and its computational limitations.

We conducted simulations using the same data-generating process and null hypotheses for the band graphs introduced in Section~\ref{sec: Simulation} of the main text.  
In addition to the methods considered in the main text, we include $M^1P_2$ in this analysis.
For the $M^1P_2$ procedure, we computed the $\alpha_i$-quantile as described in Equation \eqref{eq:M1P2-quantile} for each variable $i$ using 2400 Monte Carlo replications  (compared to 1000 in the numerical studies of \cite{verzelen_tests_2009}).

\input{GoF_M1P2}

Table~\ref{tab:M1P2_size} extends Table~\ref{tab: GoF error control} by including the estimated sizes for the $M^1P_2$ procedure, and Table~\ref{tab:M1P2_power} extends Table~\ref{tab: band} to present the power of $M^1P_2$. 
As noted in Remark~\ref{rem: m1p2}, the power of $M^1P_2$ is close to that of $M^1P_1$ across all scenarios, and they are both less competitive compared to the MC-GoF tests. 
We conclude that the MC-GoF test remains the preferred method in scenarios where the signal pattern is the dense but weak.

Furthermore, a critical limitation of the $M^1P_2$ procedure is its substantial computational burden. The method relies on Monte Carlo simulations to compute the $\alpha_i$-quantile, which is computationally intensive, particularly in high-dimensional settings. This results in a total computation time that is significantly higher than that of other methods.
As the method does not provide sufficient practical benefits to justify its significant computational demands, we have opted not to include the $M^1P_2$ procedure in further experimental comparisons. 

%% file: table-0-p=20.tex
\begin{tabular}{lll|cccccc}
\hline
& & & \multicolumn{6}{c}{Method} \\ 
$p$ & $s$ & $n$ & \nameVV & \nameDP & PRC & ERC & F$_{\Sigma}$ & \multicolumn{1}{c}{GLR-${\ell_1}$} \\ 
\hline
\hline 20 & 0.2 & 20  & 0.037 (0.977) & 0.000 (1.000) & 0.053 (0.298) & 0.045 (0.787) & 0.049 (0.553) & 0.049 (0.553) \\
 &  & 40  & 0.053 (0.298) & 0.054 (0.254) & 0.060 (0.056) & 0.051 (0.447) & 0.055 (0.213) & 0.052 (0.396) \\
 &  & 80  & 0.044 (0.823) & 0.046 (0.746) & 0.052 (0.396) & 0.052 (0.396) & 0.050 (0.500) & 0.061 (0.043) \\
\hline 120 & 0.2 & 20  & 0.051 (0.447) & 0.000 (1.000) & 0.043 (0.883) & 0.059 (0.073) & 0.060 (0.056) & 0.047 (0.702) \\
 &  & 40  & 0.047 (0.702) & 0.040 (0.944) & 0.041 (0.927) & 0.046 (0.746) & 0.049 (0.553) & 0.043 (0.855) \\
 &  & 80  & 0.048 (0.604) & 0.043 (0.855) & 0.044 (0.823) & 0.032 (0.998) & 0.042 (0.907) & 0.052 (0.346) \\
\hline 
\end{tabular}

%% file: GoF_Subset_G.tex
\begin{table}[!btph]
\centering
\begin{tabular}{lll|cccccc}
\hline
& & & \multicolumn{6}{c}{Method} \\ 
$p$ & $s$ & $n$ & $M^1$$P_1$ & Bonf & PRC & ERC & F$_{\Sigma}$ & \multicolumn{1}{c}{GLR-${\ell_1}$} \\ 
\hline
\hline 20 & 0.2 & 20  & 0.047 (0.702) & 0.000 (1.000) & 0.053 (0.298) & 0.057 (0.145) & 0.053 (0.346) & 0.058 (0.117) \\
 &  & 40  & 0.047 (0.702) & 0.039 (0.957) & 0.053 (0.346) & 0.054 (0.254) & 0.047 (0.702) & 0.050 (0.500) \\
 &  & 80  & 0.052 (0.396) & 0.047 (0.702) & 0.061 (0.043) & 0.063 (0.023) & 0.060 (0.056) & 0.050 (0.500) \\
\hline 120 & 0.2 & 20  & 0.049 (0.553) & 0.000 (1.000) & 0.043 (0.883) & 0.058 (0.117) & 0.053 (0.298) & 0.046 (0.746) \\
 &  & 40  & 0.043 (0.855) & 0.041 (0.927) & 0.054 (0.254) & 0.053 (0.298) & 0.046 (0.746) & 0.043 (0.883) \\
 &  & 80  & 0.048 (0.654) & 0.028 (1.000) & 0.040 (0.944) & 0.042 (0.907) & 0.048 (0.604) & 0.058 (0.117) \\
\hline 
\end{tabular}
\caption{Estimated sizes of various GoF tests at the significance level $\alpha$ = 0.05 on band graphs ($K = 4, K_0 = 6$). The p-value of the one-sample proportion test for each size is provided in parentheses.}\label{tab: subset}
\end{table} 

%% file: table-1-K=6.tex
\begin{tabular}{lll|cccccc}
\hline
& & & \multicolumn{6}{c}{Method} \\ 
$p$ & $s$ & $n$ & \nameVV & \nameDP & PRC & ERC & F$_{\Sigma}$ & \multicolumn{1}{c}{GLR-${\ell_1}$} \\ 
\hline
\hline 20 & 0.1 & 20  & 0.050 (.011) & 0.033 (.009) & 0.062 (.012) & 0.060 (.012) & \textbf{0.098} (.015) & 0.085 (.014) \\
 &  & 40  & 0.058 (.012) & 0.083 (.014) & \textbf{0.217} (.021) & 0.190 (.020) & 0.212 (.020) & 0.165 (.019) \\
 &  & 80  & 0.098 (.015) & 0.170 (.019) & \textbf{0.542} (.025) & 0.517 (.025) & 0.510 (.025) & 0.352 (.024) \\
\hline  & 0.15 & 20  & 0.052 (.011) & 0.065 (.012) & 0.125 (.017) & 0.115 (.016) & 0.133 (.017) & \textbf{0.135} (.017) \\
 &  & 40  & 0.083 (.014) & 0.138 (.017) & 0.490 (.025) & 0.465 (.025) & \textbf{0.532} (.025) & 0.365 (.024) \\
 &  & 80  & 0.182 (.019) & 0.280 (.022) & \textbf{0.960} (.010) & 0.948 (.011) & 0.958 (.010) & 0.810 (.020) \\
\hline  & 0.2 & 20  & 0.075 (.013) & 0.080 (.014) & 0.265 (.022) & 0.207 (.020) & \textbf{0.395} (.024) & 0.273 (.022) \\
 &  & 40  & 0.195 (.020) & 0.268 (.022) & 0.875 (.017) & 0.815 (.019) & \textbf{0.910} (.014) & 0.752 (.022) \\
 &  & 80  & 0.497 (.025) & 0.642 (.024) & \textbf{1.000} (.000) & \textbf{1.000} (.000) & \textbf{1.000} (.000) & 0.995 (.004) \\
\hline 120 & 0.1 & 20  & 0.055 (.011) & 0.055 (.011) & 0.125 (.017) & \textbf{0.128} (.017) & 0.115 (.016) & 0.070 (.013) \\
 &  & 40  & 0.062 (.012) & 0.055 (.011) & 0.225 (.021) & 0.215 (.021) & \textbf{0.260} (.022) & 0.160 (.018) \\
 &  & 80  & 0.055 (.011) & 0.085 (.014) & \textbf{0.723} (.022) & 0.685 (.023) & 0.705 (.023) & 0.370 (.024) \\
\hline  & 0.15 & 20  & 0.062 (.012) & 0.060 (.012) & 0.185 (.019) & 0.172 (.019) & \textbf{0.247} (.022) & 0.160 (.018) \\
 &  & 40  & 0.060 (.012) & 0.098 (.015) & 0.647 (.024) & 0.593 (.025) & \textbf{0.725} (.022) & 0.427 (.025) \\
 &  & 80  & 0.085 (.014) & 0.147 (.018) & \textbf{0.995} (.004) & 0.988 (.006) & 0.993 (.004) & 0.932 (.013) \\
\hline  & 0.2 & 20  & 0.075 (.013) & 0.083 (.014) & 0.487 (.025) & 0.345 (.024) & \textbf{0.695} (.023) & 0.492 (.025) \\
 &  & 40  & 0.113 (.016) & 0.125 (.017) & 0.995 (.004) & 0.955 (.010) & \textbf{1.000} (.000) & 0.975 (.008) \\
 &  & 80  & 0.435 (.025) & 0.525 (.025) & \textbf{1.000} (.000) & \textbf{1.000} (.000) & \textbf{1.000} (.000) & \textbf{1.000} (.000) \\
\hline 
\end{tabular}

%% file: table-2.tex
\begin{tabular}{lll|cccccc}
\hline
& & & \multicolumn{6}{c}{Method} \\ 
$p$ & $\xi$ & $n$ & \nameVV & \nameDP & PRC & ERC & F$_{\Sigma}$ & \multicolumn{1}{c}{GLR-${\ell_1}$} \\ 
\hline
\hline 20 & 2 & 20  & 0.043 (.010) & 0.070 (.013) & 0.120 (.016) & 0.117 (.016) & \textbf{0.138} (.017) & 0.128 (.017) \\
 &  & 40  & 0.102 (.015) & 0.130 (.017) & 0.240 (.021) & 0.225 (.021) & \textbf{0.305} (.023) & 0.295 (.023) \\
 &  & 80  & 0.265 (.022) & 0.280 (.022) & 0.542 (.025) & 0.542 (.025) & \textbf{0.665} (.024) & 0.660 (.024) \\
\hline  & 1.4 & 20  & 0.077 (.013) & 0.092 (.015) & 0.163 (.018) & 0.150 (.018) & \textbf{0.193} (.020) & 0.182 (.019) \\
 &  & 40  & 0.142 (.017) & 0.163 (.018) & 0.335 (.024) & 0.318 (.023) & 0.455 (.025) & \textbf{0.463} (.025) \\
 &  & 80  & 0.490 (.025) & 0.472 (.025) & 0.787 (.020) & 0.792 (.020) & 0.873 (.017) & \textbf{0.897} (.015) \\
\hline  & 0.9 & 20  & 0.085 (.014) & 0.095 (.015) & 0.205 (.020) & 0.185 (.019) & 0.302 (.023) & \textbf{0.315} (.023) \\
 &  & 40  & 0.333 (.024) & 0.305 (.023) & 0.603 (.024) & 0.578 (.025) & 0.770 (.021) & \textbf{0.772} (.021) \\
 &  & 80  & 0.855 (.018) & 0.765 (.021) & 0.943 (.012) & 0.938 (.012) & \textbf{0.995} (.004) & \textbf{0.995} (.004) \\
\hline 120 & 2 & 20  & 0.037 (.010) & 0.048 (.011) & 0.090 (.014) & 0.095 (.015) & \textbf{0.115} (.016) & 0.105 (.015) \\
 &  & 40  & 0.065 (.012) & 0.062 (.012) & 0.215 (.021) & 0.212 (.020) & \textbf{0.292} (.023) & 0.255 (.022) \\
 &  & 80  & 0.203 (.020) & 0.212 (.020) & 0.603 (.024) & 0.580 (.025) & 0.700 (.023) & \textbf{0.782} (.021) \\
\hline  & 1.4 & 20  & 0.045 (.010) & 0.048 (.011) & 0.128 (.017) & 0.120 (.016) & \textbf{0.190} (.020) & 0.170 (.019) \\
 &  & 40  & 0.102 (.015) & 0.085 (.014) & 0.365 (.024) & 0.357 (.024) & 0.472 (.025) & \textbf{0.500} (.025) \\
 &  & 80  & 0.427 (.025) & 0.407 (.025) & 0.853 (.018) & 0.840 (.018) & 0.930 (.013) & \textbf{0.993} (.004) \\
\hline  & 0.9 & 20  & 0.060 (.012) & 0.065 (.012) & 0.200 (.020) & 0.198 (.020) & 0.335 (.024) & \textbf{0.352} (.024) \\
 &  & 40  & 0.177 (.019) & 0.135 (.017) & 0.672 (.023) & 0.652 (.024) & 0.860 (.017) & \textbf{0.900} (.015) \\
 &  & 80  & 0.897 (.015) & 0.748 (.022) & 0.988 (.006) & 0.985 (.006) & \textbf{1.000} (.000) & \textbf{1.000} (.000) \\
\hline 
\end{tabular}

%% file: table-3.tex
\begin{tabular}{llll|cccccc}
\hline
& & & & \multicolumn{6}{c}{Method} \\ 
$p$ & $q$ & $s$ & $n$ & \nameVV & \nameDP & PRC & ERC & F$_{\Sigma}$ & \multicolumn{1}{c}{GLR-${\ell_1}$} \\ 
\hline
\hline 40 & 0.2 & 0.01 & 50  & 0.075 (.013) & 0.095 (.015) & 0.245 (.022) & 0.215 (.021) & \textbf{0.323} (.023) & 0.312 (.023) \\
 &  &  & 100  & 0.165 (.019) & 0.190 (.020) & 0.677 (.023) & 0.670 (.024) & \textbf{0.797} (.020) & 0.713 (.023) \\
\hline  &  & 0.02 & 50  & 0.180 (.019) & 0.155 (.018) & 0.550 (.025) & 0.500 (.025) & \textbf{0.800} (.020) & 0.740 (.022) \\
 &  &  & 100  & 0.485 (.025) & 0.453 (.025) & 0.985 (.006) & 0.983 (.007) & \textbf{1.000} (.000) & 0.998 (.002) \\
\hline  & 0.4 & 0.01 & 50  & 0.107 (.016) & 0.113 (.016) & 0.492 (.025) & 0.497 (.025) & \textbf{0.715} (.023) & 0.532 (.025) \\
 &  &  & 100  & 0.220 (.021) & 0.245 (.022) & 0.973 (.008) & 0.965 (.009) & \textbf{0.995} (.004) & 0.970 (.009) \\
\hline  &  & 0.02 & 50  & 0.180 (.019) & 0.205 (.020) & 0.940 (.012) & 0.912 (.014) & \textbf{0.988} (.006) & 0.963 (.010) \\
 &  &  & 100  & 0.578 (.025) & 0.560 (.025) & \textbf{1.000} (.000) & \textbf{1.000} (.000) & \textbf{1.000} (.000) & \textbf{1.000} (.000) \\
\hline 120 & 0.2 & 0.01 & 50  & 0.070 (.013) & 0.055 (.011) & 0.230 (.021) & 0.182 (.019) & \textbf{0.613} (.024) & 0.445 (.025) \\
 &  &  & 100  & 0.107 (.016) & 0.092 (.015) & 0.863 (.017) & 0.818 (.019) & \textbf{0.998} (.002) & 0.953 (.011) \\
\hline  &  & 0.02 & 50  & 0.090 (.014) & 0.065 (.012) & 0.492 (.025) & 0.310 (.023) & \textbf{0.965} (.009) & 0.907 (.015) \\
 &  &  & 100  & 0.265 (.022) & 0.190 (.020) & 0.995 (.004) & 0.985 (.006) & \textbf{1.000} (.000) & \textbf{1.000} (.000) \\
\hline  & 0.4 & 0.01 & 50  & 0.080 (.014) & 0.075 (.013) & 0.530 (.025) & 0.355 (.024) & \textbf{0.938} (.012) & 0.755 (.022) \\
 &  &  & 100  & 0.138 (.017) & 0.163 (.018) & \textbf{1.000} (.000) & \textbf{1.000} (.000) & \textbf{1.000} (.000) & \textbf{1.000} (.000) \\
\hline  &  & 0.02 & 50  & 0.133 (.017) & 0.100 (.015) & 0.885 (.016) & 0.677 (.023) & \textbf{1.000} (.000) & 0.998 (.002) \\
 &  &  & 100  & 0.282 (.023) & 0.260 (.022) & \textbf{1.000} (.000) & \textbf{1.000} (.000) & \textbf{1.000} (.000) & \textbf{1.000} (.000) \\
\hline 
\end{tabular}

%% file: GoF_M1P2.tex
\begin{sidewaystable}[!btph]
\centering
\begin{tabular}{lll|ccccccc}
\hline
& & & \multicolumn{7}{c}{Method} \\ 
$p$ & $s$ & $n$ & $M^1$$P_1$ & $M^1$$P_2$ & Bonf & PRC & ERC & F$_{\Sigma}$ & \multicolumn{1}{c}{GLR-${\ell_1}$} \\ 
\hline
\hline 20 & 0.2 & 20  & 0.037 (.977) & 0.040 (.944) & 0.000 (1.000) & 0.053 (.298) & 0.045 (.787) & 0.049 (.553) & 0.049 (.553) \\
 &  & 40  & 0.053 (.298) & 0.031 (.999) & 0.054 (.254) & 0.060 (.056) & 0.051 (.447) & 0.055 (.213) & 0.052 (.396) \\
 &  & 80  & 0.044 (.823) & 0.047 (.702) & 0.046 (.746) & 0.052 (.396) & 0.052 (.396) & 0.050 (.500) & 0.061 (.043) \\
 
\hline 120 & 0.2 & 20  & 0.051 (.447) & 0.044 (.823) & 0.000 (1.000) & 0.043 (.883) & 0.059 (.073) & 0.060 (.056) & 0.047 (.702)   \\
 &  & 40  & 0.047 (.702) & 0.036 (.988) & 0.040 (.944) & 0.041 (.927) & 0.046 (.746) & 0.049 (.553) & 0.043 (.855) \\
 &  & 80  & 0.048 (.604) & 0.043 (.883) & 0.043 (.855) & 0.044 (.823) & 0.032 (.998) & 0.042 (.907) & 0.052 (.346) \\
\hline 
\end{tabular}

\caption{Estimated sizes of various GoF tests, including the $M^1P_2$ procedure, at the significance level $\alpha= 0.05$ on band graphs ($K = 6$, $K_0 = 6$). The p-value of the one-sample proportion test for each size is provided in parentheses. 
}\label{tab:M1P2_size}
\end{sidewaystable}

\begin{sidewaystable}[!btph]
\centering
\begin{tabular}{lll|ccccccc}
\hline
& & & \multicolumn{7}{c}{Method} \\ 
$p$ & $s$ & $n$ & $M^1$$P_1$ & $M^1$$P_2$ & Bonf & PRC & ERC & F$_{\Sigma}$ & \multicolumn{1}{c}{GLR-${\ell_1}$} \\ 
\hline
\hline 20 & 0.1 & 20  & 0.050 (.011) & 0.058 (.012) & 0.033 (.009) & 0.062 (.012) & 0.060 (.012) & \textbf{0.098} (.015) & 0.085 (.014) \\
 &  & 40  & 0.058 (.012) & 0.055 (.011) & 0.083 (.014) & \textbf{0.217} (.021) & 0.190 (.020) & 0.212 (.020) & 0.165 (.019) \\
 &  & 80  & 0.098 (.015) & 0.093 (.014) & 0.170 (.019) & \textbf{0.542} (.025) & 0.517 (.025) & 0.510 (.025) & 0.352 (.024) \\
\hline  & 0.15 & 20  & 0.052 (.011) & 0.030 (.009) & 0.065 (.012) & 0.125 (.017) & 0.115 (.016) & 0.133 (.017) & \textbf{0.135} (.017) \\
 &  & 40  & 0.083 (.014) & 0.078 (.013) & 0.138 (.017) & 0.490 (.025) & 0.465 (.025) & \textbf{0.532} (.025) & 0.365 (.024) \\
 &  & 80  & 0.182 (.019) & 0.240 (.021) & 0.280 (.022) & \textbf{0.960} (.010) & 0.948 (.011) & 0.958 (.010) & 0.810 (.020) \\
\hline  & 0.2 & 20  & 0.075 (.013) & 0.063 (.012) & 0.080 (.014) & 0.265 (.022) & 0.207 (.020) & \textbf{0.395} (.024) & 0.273 (.022) \\
 &  & 40  & 0.195 (.020) & 0.195 (.020) & 0.268 (.022) & 0.875 (.017) & 0.815 (.019) & \textbf{0.910} (.014) & 0.752 (.022) \\
 &  & 80  & 0.497 (.025) & 0.475 (.025) & 0.642 (.024) & \textbf{1.000} (.000) & \textbf{1.000} (.000) & \textbf{1.000} (.000) & 0.995 (.004) \\
 
\hline 120 & 0.1 & 20  & 0.055 (.011) & 0.053 (.011) & 0.055 (.011) & 0.125 (.017) & \textbf{0.128} (.017) & 0.115 (.016) & 0.070 (.013) \\
 &  & 40  & 0.062 (.012) & 0.045 (.010) & 0.055 (.011) & 0.225 (.021) & 0.215 (.021) & \textbf{0.260} (.022) & 0.160 (.018) \\
 &  & 80  & 0.055 (.011) & 0.050 (.011) & 0.085 (.014) & \textbf{0.723} (.022) & 0.685 (.023) & 0.705 (.023) & 0.370 (.024) \\
\hline  & 0.15 & 20  & 0.062 (.012) & 0.048 (.011) & 0.060 (.012) & 0.185 (.019) & 0.172 (.019) & \textbf{0.247} (.022) & 0.160 (.018) \\
 &  & 40  & 0.060 (.012) & 0.035 (.009) & 0.098 (.015) & 0.647 (.024) & 0.593 (.025) & \textbf{0.725} (.022) & 0.427 (.025) \\
 &  & 80  & 0.085 (.014) & 0.073 (.013) & 0.147 (.018) & \textbf{0.995} (.004) & 0.988 (.006) & 0.993 (.004) & 0.932 (.013) \\
\hline  & 0.2 & 20  & 0.075 (.013) & 0.060 (.012) & 0.083 (.014) & 0.487 (.025) & 0.345 (.024) & \textbf{0.695} (.023) & 0.492 (.025) \\
 &  & 40  & 0.113 (.016) & 0.095 (.015) & 0.125 (.017) & 0.995 (.004) & 0.955 (.010) & \textbf{1.000} (.000) & 0.975 (.008) \\
 &  & 80  & 0.435 (.025) & 0.358 (.024) & 0.525 (.025) & \textbf{1.000} (.000) & \textbf{1.000} (.000) & \textbf{1.000} (.000) & \textbf{1.000} (.000) \\
\hline 
\end{tabular}

\caption{Power of various GoF tests, including the $M^1P_2$ procedure, on band graphs ($K = 6$, $K_0 = 1$) at the significance level $\alpha = 0.05$. 
See the caption of Table~\ref{tab: band} for details on methods. }\label{tab:M1P2_power}
\end{sidewaystable}

%% file: 8Proof.tex
This section provides the omitted proofs for the theoretical results in the main paper.

\subsection{Proofs in Section~\ref{sec: residual rotation}}\label{app: pf css}

Proposition~\ref{prop: residual rotation} is based on the following two lemmas, whose proof is given latter

\begin{lemma}\label{lem: conditional exchangeable}
    Suppose three random elements $X$, $Y$, and $Z$ satisfy that conditional on $Z$, $X$, and $Y$ are independent and identically distributed. Then for any bounded Borel function $g$ defined on the range of $X$, it holds that $\E\left( g(X)\mid Y,Z\right)=\E\left( g(Y)\mid X,Z\right)$. 
\end{lemma} 
\begin{lemma}\label{lem: residual rotation}
For any given $i\in [p]$, the distribution of the output $\widetilde{\mathbf{X}}_{i}$ from Algorithm~\ref{alg: residual rotation} is the uniform distribution on 
\begin{equation}\label{eq: lem residual rotation}
\{ 
\mx_{i}\in \mathbb{R}^{n}: \mx_{i}\tp \bs{1}_{n} = \mX_{i}\tp \bs{1}_{n}, \mx_{i}\tp\mx_{i} =\mX_{i}\tp \mX_{i},  \mx_{i}\tp \mX_{j} = \mX_{i}\tp \mX_{j}, \forall j \in N_{i}
\}.
\end{equation}
\end{lemma}
It is straightforward to see that the set in \eqref{eq: lem residual rotation} is the same as $\mc{X}_{\Psi}^{i}$ in \eqref{eq: conditional support}. 
We can now prove Proposition~\ref{prop: residual rotation} as follows. 
\begin{proof}[Proof of Proposition~\ref{prop: residual rotation}]
Lemma~\ref{lem: residual rotation} shows that both  $\mathbf{X}_{i}$ and $\widetilde{\mathbf{X}}_{i}$ lie on $\mc{X}_{\Psi}^{i}$, and thus $\psi_{G}(\mathbf{X})=\psi_{G}(\widetilde{\mathbf{X}})$. 

We have argued that the conditional distribution of $\mathbf{X}_{i}$ given $(\mathbf{X}_{-i}, \Psi)$ is the uniform distribution on $\mc{X}_{\Psi}^{i}$ defined in \eqref{eq: conditional support}. 
In addition, when $\widetilde{\mathbf{X}}_{i}$ is generated, it only uses the information about $(\mathbf{X}_{-i}, \Psi)$. Therefore, Lemma~\ref{lem: residual rotation} implies that conditioning on $(\mathbf{X}_{-i}, \Psi)$, 
the output column $\widetilde{\mathbf{X}}_{i}$ and the input column $\mathbf{X}_{i}$ are independent and they follow the same distribution. 
We can then apply Lemma~\ref{lem: conditional exchangeable} to conclude that the conditional distribution of $\widetilde{\mathbf{X}}_{i}\mid (\mathbf{X}_{i}, \mathbf{X}_{-i}, \Psi)$ is almost surely the same as that of $\mathbf{X}_{i}\mid (\widetilde{\mathbf{X}}_{i},  \mathbf{X}_{-i}, \Psi)$. 
By definition of $\widetilde{\mathbf{X}}$, we proved the desired statement. 
\end{proof}


\begin{proof}[Proof of Lemma~\ref{lem: residual rotation}]
    
We only need to show that if the rows of $\mathbf{X}$ are i.i.d. samples from a distribution in $\mc{M}_{G}$, then conditional on $(\mathbf{X}_{-i}, \Psi)$, the observed data $\mathbf{X}_{i}$ and the output $\widetilde{\mathbf{X}}_{i}$ share the same conditional distribution.

By properties of multivariate normal distribution, the conditional distribution of $\mathbf{X}_{i}$ given $\mathbf{X}_{-i}$ is $\mathbf{N}_{n}(\mu_{i} \bs{1} + \mathbf{X}_{N_i}\bs{\alpha}, ~~\bOmega_{i,i}^{-1} \bs{I}_{n} )$ for some $|N_i|$-dimensional parameter $\bs{\alpha}$ that depends on $\bOmega$. 

Without loss of generality, we would assume in the following that $\left[\bs{1},\mathbf{X}_{N_{i}}\right]$ has full column rank and $n>1+|N_i|$. The first statement here is guaranteed to hold with probability 1 because of the continuity of the Gaussian measure. 
If $n\leq 1+|N_i|$, it is clear that the output $\widetilde{\mathbf{X}}_{i}=\mathbf{X}_{i}$ because the residual of the linear regression is zero. 

In the following, we condition on $\mathbf{X}_{-i}$. 
Let $\bs{U}\in \mathbb{R}^{n\times (1+|N_i|)}$ be a matrix whose columns are an orthonormal basis of the column space of $\left[\bs{1},\mathbf{X}_{N_{i}}\right]$. 
Let $l_i=n-1-|N_i|$ and $\bs{V}\in \mathbb{R}^{n\times l_i}$ be a matrix such that $\bs{V}\tp \bs{U}=0$ and $\bs{V}\tp \bs{V}=\bs{I}_{l_i}$. 
By properties of least squares estimation, the fitted vector in Step 2 of Algorithm \ref{alg: residual rotation} satisfies that $\bs{F}=\bs{U}\bs{U}\tp \mathbf{X}_{i}$ and the residual vector satisfies that $\bs{R}=\bs{V}\bs{V}\tp \mathbf{X}_{i}$. 
We also have $\mathbf{X}_{i}=\bs{F}+\bs{R}$. 
By properties of Gaussian linear regression models, the following statements hold: 
\begin{enumerate}
    \item $\bs{F}$ can be expressed by $\mathbf{X}_{-i}$ and $\Psi$;
    \item the distribution of $\bs{V}\tp \mathbf{X}_{i}$ is $\mathbf{N}_{l_i}(\bs{0},\bOmega_{i,i}^{-1} \bs{I}_{l_i} )$. 
\end{enumerate}

By the statement 1 and the fact that $\Psi$ contains $\mathbf{X}_{i}\tp \mathbf{X}_{i}$, we can write
$\|\bs{V}\tp \mathbf{X}_{i}\|^2 = \|\bs{R}\|^2=\|\mathbf{X}_{i}\|^2-\|\bs{F}\|^2$ as a function of $\mathbf{X}_{-i}$ and $\Psi$. 

By the statement 2, $\bs{V}\tp \mathbf{X}_{i} / \|\bs{R}\|$ follows the uniform distribution on $\mathbb{S}^{l_i-1}$, the unit sphere in $\mathbb{R}^{l_i}$. 
By Basu's theorem and the properties of exponential families of full rank, we can conclude that $\bs{V}\tp \mathbf{X}_{i} / \|\bs{R}\|$ and $\Psi$ are conditionally independent given $\mathbf{X}_{-i}$.

Let $\bs{E}$ be the standard normal $n$-vector draw in Step 3 of Algorithm~\ref{alg: residual rotation}. We can use a similar argument to show that $\bs{V}\tp \bs{E} / \|\widetilde{\bs{R}}\|$ also follows the uniform distribution on $\mathbb{S}^{l_i-1}$.  Therefore, $\bs{V}\tp \mathbf{X}_{i} / \|\bs{R}\|$ and $\bs{V}\tp \bs{E} / \|\widetilde{\bs{R}}\|$ are conditionally i.i.d. given $(\mathbf{X}_{-i}, \Psi)$. 
Furthermore, if $l_i=n-1-|N_i|\geq 2$, then w.p. 1, $\bs{V}\tp \mathbf{X}_{i} / \|\bs{R}\|$ and $\bs{V}\tp \bs{E} / \|\widetilde{\bs{R}}\|$ are different. 

Conditional on $(\mathbf{X}_{-i},\Psi)$, we have
\begin{align*}
    \mathbf{X}_{i} & =\bs{F}+\bs{R}\\
    & = \bs{F}+\|\bs{R}\| \bs{V} \left(  \bs{V}\tp \mathbf{X}_{i}/\|\bs{R}\| \right) \\
    & \eqd \bs{F}+\|\bs{R}\| \bs{V} \left(\bs{V}\tp \bs{E} / \|\widetilde{\bs{R}}\| \right) \\ 
    & = \bs{F}+ \frac{\|\bs{R}\| }{\|\widetilde{\bs{R}}\|} \widetilde{\bs{R}} \\ 
    & = \widetilde{\mathbf{X}}_{i}, 
\end{align*}
and $\mathbf{X}_{i}\neq \widetilde{\mathbf{X}}_{i}$ w.p. 1 if $l_i\geq 2$. 
This completes the proof. 
\end{proof}

\begin{proof}[Proof of Lemma~\ref{lem: conditional exchangeable}]
    It is straightforward to see that 
    \begin{align*}
        \E\left( g(X)\mid Y,Z\right)=& \E\left( g(X)\mid Z\right)\\
        =& \E\left( g(Y)\mid Z\right)\\
        =& \E\left( g(Y)\mid X,Z\right), \text{a.s.}
    \end{align*}
    where the first and last equations are due to the conditional independence between $X$ and $Y$, and the second one is due to the identical conditional distribution. 
\end{proof}

\subsection{Proofs in Section~\ref{sec: sample MC}}

\begin{proof}[Proof of Proposition~\ref{prop: exchangeable}]
Denote by $P$ the distribution of the input $\mathbf{X}$,  by $\mc{C}_0$ the (time-inhomogeneous) Markov chain from $\mathbf{X}$ to the hub $\mathbf{X}^{(hub)}$, by $K_{i}(\mathbf{x}, \mathbf{y})$ the transition kernel of $\mc{C}_0$ at the $i$-th step, 
and by $\overleftarrow{\mc{C}}_0$ the time-reversed chain of $\mc{C}_0$.

By Proposition~\ref{prop: residual rotation}, $K_{i}(\mathbf{x}, \mathbf{y})$ is either 0 or $K_{i}(\mathbf{x}, \mathbf{y})=K_{i}(\mathbf{y}, \mathbf{x})$ with $\psi_{G}(\mathbf{x})= \psi_{G}(\mathbf{y})$. 
If $P\in \mc{M}_{G}$, we can conclude that 
$$
f_{\bmu, \bOmega}(\mathbf{x})K_{i}(\mathbf{x}, \mathbf{y}) =  f_{\bmu, \bOmega}(\mathbf{y})K_{i}(\mathbf{y}, \mathbf{x}), \quad \forall \mathbf{x},~ \mathbf{y} \in \R^{n\times p}. 
$$
This means that $P$ satisfies the detailed balance for Markov chain $\mc{C}_0$. 
As a result, the end point of $\mc{C}_0$ (i.e. $\mathbf{X}^{(hub)}$) follows the distribution $P$ and $K_{i}$ is the transition kernel of the reversed chain $\overleftarrow{\mc{C}}_0$ at the $(L|\mc{I}|-i+1)$-th step. 

Consider the Markov chain $\overleftarrow{\mc{C}}_{m}$ that starts from the hub $\mathbf{X}^{(hub)}$ and ends at a generated $\widetilde{\mathbf{X}}^{(m)}$ in Step 2 of Algorithm~\ref{alg: exchangeable}. 
By the procedure of the algorithm and the property of the transition kernel $K_i$ we proved above, $\overleftarrow{\mc{C}}_{m}$ follows the same distribution as $\overleftarrow{\mc{C}}_0$. 

Furthermore, when conditioning on the initial point $\mathbf{X}^{(hub)}$, the Markov chains $\overleftarrow{\mc{C}}_{m}$ $(m=0,1,\ldots,M)$ are mutually independent. Therefore, $\overleftarrow{\mc{C}}_{m}$ $(m=0, 1,\ldots,M)$ are i.i.d. conditional on $\mathbf{X}^{(hub)}$. As a result, the endpoints of these chains are conditionally exchangeable given $\mathbf{X}^{(hub)}$, which can be expressed as follows: for any permutation $\pi$ of $(m+1)$ elements and any $(m+1)$-variate integral function $g$, it holds that 
\begin{equation}\label{eq: expectation permuate}
\E \left[ g(\mathbf{X},\widetilde{\mathbf{X}}_{1}, \ldots, \widetilde{\mathbf{X}}^{(M)})   \mid \mathbf{X}^{(hub)}\right]=\E \left[ g\circ \pi (\mathbf{X},\widetilde{\mathbf{X}}_{1}, \ldots, \widetilde{\mathbf{X}}^{(M)})   \mid \mathbf{X}^{(hub)}\right], \mbox{a.s.} 
\end{equation}
Taking the expectation on both sides of Equation~\eqref{eq: expectation permuate}, we conclude that $\mathbf{X},\widetilde{\mathbf{X}}_{1}, \ldots, \widetilde{\mathbf{X}}^{(M)}$ are exchangeable. 

To prove the second statement, we point out that the coordinates in $\mc{T}^{c}$ remain still during the generation, so $\mathbf{X}_{-\mc{T}}=\mathbf{X}_{-\mc{T}}^{(hub)}=\widetilde{\mathbf{X}}_{-\mc{T}}^{(m)}$. 
We then take the conditional expectation given $\mathbf{X}_{-\mc{T}}^{(hub)}$ on both sides of Equation~\eqref{eq: expectation permuate}, which implies the conditional exchangeability as desired. 
\end{proof}

\subsection{Validity of the MC-GoF Test}\label{app:Proof of Lemma 3}

Throughout this section, we simplify the notation $T(\mathbf{X}), T(\widetilde{\mathbf{X}}^{(1)}), T(\widetilde{\mathbf{X}}^{(2)}), \cdots, T(\widetilde{\mathbf{X}}^{(M)}) $  as $T_0, T_1, T_2, \cdots, T_M$.

By Proposition~\ref{prop: exchangeable}, $\mathbf{X}, \widetilde{\mathbf{X}}^{(1)}, \widetilde{\mathbf{X}}^{(2)}, \cdots, \widetilde{\mathbf{X}}^{(M)}$ are exchangeable. Therefore, $T_0, T_1, T_2, \cdots, T_M$ are exchangeable. 
This already guarantees the validity of the p-value of Algorithm~\ref{alg:gof} as stated in Proposition~\ref{prop: GoF valid} using classic results in statistical inference; see for example \cite[Chapter 15.2.1]{lehmannTesting2022}. For completeness, we prove the validity of the p-values mentioned in Appendix~\ref{remark:p-value}. 

In the proof,  $ S $ is a random integer between $ 1 $ and $ 1 + \sum_{m=1}^{M} \mathbb{I}(T_{m} = T_0) $.

We use $ V_1, V_2, \ldots, V_K $ to denote the unique values of $ T_0, T_1, \ldots, T_M $ sorted in descending order. Let $ n_1, n_2, \ldots, n_K $ be the respective frequencies of these values. The following two conditions hold:
\[
\begin{aligned}
    &1. \quad n_i \geq 1 \quad \text{for all } i \\
    &2. \quad \sum_{i=1}^{K} n_i = 1 + M
\end{aligned}
\]

Denote by $\tilde{\mathbf{P}}(\cdot)=\mathbf{P}(\cdot \mid  V_1, \ldots, V_K, n_1, \ldots, n_K)$ the conditional probability of $ T_0 $ being equal to $ V_j $ given all distinct values and their frequencies. 
By exchangeability of $T_{m}$'s, it holds that 
\[
\tilde{\mathbf{P}}(T_0 = V_j ) = \frac{n_j}{1+M}, \quad j = 1,2,\ldots,K.
\]

\subsubsection{One-sided randomized test}

Consider the p-value defined as
\[
\text{pVal} = \frac{1}{1+M} \left( S + \sum_{m=1}^{M} \mathbb{I}(T_m > T_0) \right). 
\]

For $ \alpha \in (0,1) $, suppose $ E $ is an integer between $1$ and $K$ such that
\[
\sum_{i=1}^{E-1} n_i < \alpha (1+M) \leq \sum_{i=1}^{E} n_i,
\]
where the $\sum_{i=1}^0$ is regarded as $0$ when $E$ is 1. 
We have
\begin{align*}
\tilde{\mathbf{P}}(\text{reject } H_0) & = \tilde{\mathbf{P}}\left( S + \sum_{m=1}^{M} \mathbb{I}(T_{m} > T_0) \leq \alpha (1+M) \right) \\
& = \tilde{\mathbf{P}}\left( S \leq \alpha (1+M) - \sum_{m=1}^{M} \mathbb{I}(T_{m} > T_0)\right). 
\end{align*}
We further condition on  the event that $T_0=V_{k}$ for any $k\in [K]$, so that 
\begin{align*}
&\quad \tilde{\mathbf{P}}\left( S \leq \alpha (1+M) - \sum_{m=1}^{M} \mathbb{I}(T_{m} > T_0) \mid  T_0=V_{k} \right)\\
& = \begin{cases}
  1, & \text{when } k=1, \ldots, E-1, \\
  \frac{\lfloor \alpha (1+M) \rfloor - \sum_{j=1}^{E-1} n_j}{n_{E}}, & \text{when } k=E, \\
  0, & \text{when } k=E+1, \ldots, K.
\end{cases}
\end{align*}
Therefore, we can compute the rejection probability as follows:
\begin{align*}
\tilde{\mathbf{P}}(\text{reject } H_0) & = \sum_{k=1}^{K} \tilde{\mathbf{P}}\left( S \leq \alpha (1+M) - \sum_{m=1}^{M} \mathbb{I}(T_{m} > T_0), T_0=V_k\right) \\
&= \sum_{j=1}^{K} \frac{n_k}{1+M}\tilde{\mathbf{P}}\left( S \leq \alpha (1+M) - \sum_{m=1}^{M} \mathbb{I}(T_{m} > T_0) \mid  T_0=V_{k} \right) \\
&=(1+M)^{-1}\left( \sum_{k=1}^{E-1}n_k + \lfloor \alpha (1+M) \rfloor - \sum_{j=1}^{E-1} n_j \right)\\
&= (1+M)^{-1}\lfloor \alpha (1+M) \rfloor \\
&\leq \alpha. 
\end{align*}

This proves the validity of the randomized p-value defined in Remark~\ref{remark:p-value}.

\subsubsection{Two-sided randomized p-value}
 The two-sided p-value is defined as
\[
\frac{2}{1+M}\left[ S+\min\left(  \sum^M_{m=1}  I\left\{T_{m} >T_0\right\},\sum^M_{m=1}I\left\{T_{m} <T_0\right\}  \right)   \right].
\]

For $ \alpha \in (0,1) $, suppose $ E $ and $F$ are two integers between $1$ and $K$ such that
\[
\begin{aligned}
\sum_{i=1}^{E-1} n_i & < \frac{\alpha}{2} (1+M) \leq \sum_{i=1}^{E} n_i, \\
\sum_{i=F+1}^{K} n_i & < \frac{\alpha}{2} (1+M) \leq \sum_{i=F}^{K} n_i,
\end{aligned}
\]
where the $\sum_{i=1}^0$ is regarded as $0$ when $E$ is 1 and the $\sum_{i= K+1}^{K}$ is regarded as 0 when $F=K$. 
Summing up the left-hand sides of the two inequalities, we have $\sum_{i=1}^{E-1} n_i + \sum_{i=F+1}^{K} n_i< \alpha(1+M) < (1+M) =\sum_{i=1}^{K} n_i$. 
So $E\leq F$. 

We have
\begin{align*}
\tilde{\mathbf{P}}(\text{reject } H_0) & = \tilde{\mathbf{P}}\left( S+\min\left(  \sum^M_{m=1}  I\left\{T_{m} >T_0\right\},\sum^M_{m=1}I\left\{T_{m} <T_0\right\}  \right)  \leq \frac{\alpha}{2} (1+M) \right). 
\end{align*}

We first consider the case where $E<F$. 
If we further condition on $T_0=V_{k}$ for any $k\in [K]$, we have 
\begin{align*}
&\tilde{\mathbf{P}}\left( S+\min\left(  \sum^M_{m=1}  I\left\{T_{m} >T_0\right\},\sum^M_{m=1}I\left\{T_{m} <T_0\right\}  \right)  \leq \frac{\alpha}{2} (1+M) \right)\\
= & \begin{cases}
  1, & \text{when } k=1, \ldots, E-1, \\
  \frac{\lfloor \alpha (1+M) /2 \rfloor - \sum_{j=1}^{E-1} n_j}{n_{E}}, & \text{when } k=E, \\
  0, & \text{when } k=E+1, \ldots, F-1, \\
  \frac{\lfloor \alpha (1+M) /2 \rfloor - \sum_{j=F+1}^{K} n_j}{n_{F}}, & \text{when } k=F, \\
1,  & \text{when } k=F+1, \ldots, K.
\end{cases}
\end{align*}
Therefore, we can compute the rejection probability as follows:

\begin{align*}
&\quad \tilde{\mathbf{P}}(\text{reject } H_0)\\
& = \sum_{k=1}^{K} \tilde{\mathbf{P}}\left( S \leq \frac{\alpha}{2} (1+M) - \min\left(  \sum^M_{m=1}  I\left\{T_{m} >T_0\right\},\sum^M_{m=1}I\left\{T_{m} <T_0\right\}  \right), T_0=V_k\right) \\
&= \sum_{j=1}^{K} \frac{n_k}{1+M}\tilde{\mathbf{P}}\left( S \leq \frac{\alpha}{2} (1+M) - \min\left(  \sum^M_{m=1}  I\left\{T_{m} >T_0\right\},\sum^M_{m=1}I\left\{T_{m} <T_0\right\}  \right) \mid  T_0=V_{k} \right) \\
&=(1+M)^{-1}\left( \sum_{k=1}^{E-1}n_k + \lfloor \alpha (1+M) /2 \rfloor - \sum_{j=1}^{E-1} n_j  + \lfloor \alpha (1+M) /2 \rfloor - \sum_{j=F+1}^{K} n_j + \sum_{k=F+1}^{K} n_k \right)\\
&= (1+M)^{-1}2 \lfloor \alpha (1+M)/2 \rfloor \\
&\leq \alpha. 
\end{align*}

Next, we consider the case where $E=F$. We have 
\begin{align*}
&\tilde{\mathbf{P}}\left( S+\min\left(  \sum^M_{m=1}  I\left\{T_{m} >T_0\right\},\sum^M_{m=1}I\left\{T_{m} <T_0\right\}  \right)  \leq \frac{\alpha}{2} (1+M) \right)\\
= & \begin{cases}
  1, & \text{when } k=1, \ldots, E-1, \\
  \frac{\lfloor \alpha (1+M) /2 \rfloor - \min\left( \sum_{j=1}^{E-1} n_j , \sum_{j=F+1}^{K} n_j \right) }{n_{E}}, & \text{when } k=E=F, \\
1,  & \text{when } k=F+1, \ldots, K.
\end{cases}
\end{align*}
Similar to the previous case, the rejection probability can be computed as follows: 
\begin{align*}
\tilde{\mathbf{P}}(\text{reject } H_0) 
&= \sum_{j=1}^{K} \frac{n_k}{1+M}\tilde{\mathbf{P}}\left( S \leq \frac{\alpha}{2} (1+M) - \min\left(  \sum^M_{m=1}  I\left\{T_{m} >T_0\right\},\sum^M_{m=1}I\left\{T_{m} <T_0\right\}  \right) \mid  T_0=V_{k} \right) \\
&=(1+M)^{-1}\left( \sum_{k=1}^{E-1}n_k + \lfloor \alpha (1+M) /2 \rfloor -  \min\left( \sum_{j=1}^{E-1} n_j , \sum_{j=F+1}^{K} n_j \right) + \sum_{k=F+1}^{K} n_k \right)\\
&=(1+M)^{-1}\left( \lfloor \alpha (1+M) /2 \rfloor + \max\left( \sum_{j=1}^{E-1} n_j , \sum_{j=F+1}^{K} n_j \right)  \right)\\
&< (1+M)^{-1}2 \lfloor \alpha (1+M)/2 \rfloor \\
&\leq \alpha, 
\end{align*}
where the second to last inequality is due to the definition of $E$ and $F$. 

Combining the two cases, we complete the proof of the validity of the randomized two-sided p-value defined in Remark~\ref{remark:p-value}.

\subsection{Validity of local MC-GoF tests in Section~\ref{sec: GoF local}} \label{app: pf GoF local}

In the proof of Proposition~\ref{prop: residual rotation}, we only make use of the neighborhood $N_i$ defined by $G$ (or equivalently, the edges connecting $i$ and any other nodes). 
Therefore, the result of Proposition~\ref{prop: residual rotation} still holds for any $i\in \mc{T}$ under the hypothesis in \eqref{eq: local null GoF}. 
Since only $\mathbf{X}_i$ with $i\in \mc{T}$ will be updated when running Algorithm~\ref{alg: exchangeable}, the same argument in Proposition~\ref{prop: exchangeable} shows that the generated copies $\widetilde{\mathbf{X}}^{(1)}, \ldots, \widetilde{\mathbf{X}}^{(M)}$ are jointly exchangeable with the input $\mathbf{X}$ under the hypothesis in \eqref{eq: local null GoF}. As a result, the p-value computed as in \eqref{eqn:pval-1} remains valid for testing the local graph structure.

\input{8aProofPower}

%% file: 8aProofPower.tex
\subsection{Metric of deviation}\label{app: pf metric}

\begin{proof}[Proof of Proposition~\ref{prop: Omega related to R}]

The equation for $D(\mathbf{\Omega}_{\mc{C}}, \infty)$ is obvious once we recognize that $\mathbf{R}(\mathbf{A})$ is the correlation matrix of $X_C$ given $X_H$. 
We focus on the lower bound on $ D(\mathbf{\Omega}_{\mc{C}}, F)$. 

Without loss of generality, we rearrange the indices so that $\mc{C}=\{1,2,\ldots, q\}$ and $\mc{H}=\{q+1, \ldots, p\}$. 
Furthermore, the desired inequality is implied by the following result:  Suppose $0<c_0\leq \operatorname{var}(X_i\mid X_{\mc{H}}) \leq c_1$ and  $\|\mathbf{\Omega}_{\mc{C}}\|\leq c_2$, then it holds that 
    $$
   D(\mathbf{\Omega}_{\mc{C}}, F)^2\geq  \frac{c_0^2}{c_1 c_2} \sum_{i,j\in \mc{C}}\mathbf{\Omega}_{i,j}^2. 
    $$

Let $\mathbf{A} = \mathbf{\Omega}_{\mc{C}}^{-1}$ is the conditional covariance matrix of $X_C$ given $X_H$.
Let $\mathbf{D} = \operatorname{diag}(\mathbf{A})$ is the diagonal matrix of $\mathbf{A}$, whose diagonal elements are the conditional variance of $X_i$ ($i\in \mc{C}$) given $X_{\mc{H}}$. 
By assumption, we have $c_0\leq \mathbf{D}_{ii}\leq c_1$.

Define
$\mathbf{K} = \mathbf{R}(\mathbf{A}) - \mathbf{I}_q=\mathbf{D}^{-1/2} \mathbf{A} \mathbf{D}^{-1/2} - \mathbf{I}_q$. 
The diagonal elements of $\mathbf{K}$ are
     $$
     \mathbf{K}_{ii} = \mathbf{D}_{ii}^{-1} \mathbf{A}_{ii} - 1 = 1 - 1 = 0,
     $$
while the off-diagonal elements are 
$$
     \mathbf{K}_{ij} = \mathbf{D}_{ii}^{-1/2} \mathbf{A}_{ij} \mathbf{D}_{jj}^{-1/2}, \quad \text{for } i \ne j.
$$
Since $
     \mathbf{A} = \mathbf{D}^{1/2} (\mathbf{K} + \mathbf{I}_q) \mathbf{D}^{1/2}$, we have
     $$
     \mathbf{\Omega}_{\mc{C}} = \mathbf{A}^{-1} = \left( \mathbf{D}^{1/2} (\mathbf{K} + \mathbf{I}_q) \mathbf{D}^{1/2} \right)^{-1} = \mathbf{D}^{-1/2} (\mathbf{K} + \mathbf{I}_q)^{-1} \mathbf{D}^{-1/2}.
     $$
Therefore, for $i \ne j$, we have 
     $$
     \mathbf{\Omega}_{i,j} =  \mathbf{D}_{ii}^{-1/2} \left[ (\mathbf{K} + \mathbf{I}_q)^{-1} \right]_{ij} \mathbf{D}_{jj}^{-1/2}.
     $$

     Using the identities $(\mathbf{K} + \mathbf{I}_q)^{-1} - \mathbf{I}_q = - \mathbf{K} (\mathbf{K} + \mathbf{I}_q)^{-1}$ and $(\mathbf{I}_q)_{ij}=0$, we have:
     $$
     \mathbf{\Omega}_{i,j} = - \mathbf{D}_{ii}^{-1/2} \left[ \mathbf{K} (\mathbf{K} + \mathbf{I}_q)^{-1} \right]_{ij} \mathbf{D}_{jj}^{-1/2},  ~~ i\neq j.
     $$

The sum of squares of $\mathbf{\Omega}_{i,j}$ over all pairs of distinct $i$ and $j$ in $\mc{C}$ is
$$
     \begin{aligned}
     \sum_{i \ne j} \left( \mathbf{\Omega}_{i,j} \right)^2 &= \sum_{i \ne j} \left( \mathbf{D}_{ii}^{-1/2} \left[ - \mathbf{K} (\mathbf{K} + \mathbf{I}_q)^{-1} \right]_{ij} \mathbf{D}_{jj}^{-1/2} \right)^2 \\
     &\leq \|\mathbf{D}^{-1 / 2} \mathbf{K}\left(\mathbf{K}+\mathbf{I}_q\right)^{-1} \mathbf{D}^{-1 / 2}\|_{F}^2\\
    & \leq \| \mathbf{D}^{-1/2} \|_{}^4 \| (\mathbf{K} + \mathbf{I}_q)^{-1} \|_{}^2 \| \mathbf{K} \|_F^2.
     \end{aligned}
     $$

Since $D_{ii} \geq c_0$, we have $\| \mathbf{D}^{-1/2} \|_{}^4\leq c_0^{-2}$. 
By definition, we have $(\mathbf{K} + \mathbf{I}_q)^{-1}=\mathbf{D}^{1/2}\mathbf{\mathbf{\Omega}_{\mc{C}}}\mathbf{D}^{1/2}$. 
Since $D_{ii}\leq c_1$ and , we have 
$$
\| (\mathbf{K} + \mathbf{I}_q)^{-1} \|_{}^2\leq \|\mathbf{D}^{1/2}\|^2\|\mathbf{\mathbf{\Omega}_{\mc{C}}}\|\leq c_1 c_2, 
$$
where we have used the assumption that $\|\mathbf{\mathbf{\Omega}_{\mc{C}}}\|\leq c_2$. We thus conclude that 
$$
\sum_{i \ne j} \left( \mathbf{\Omega}_{i,j} \right)^2\leq \frac{c_1 c_2}{c_0^2}  \| \mathbf{K} \|_F^2. 
$$
By definition, it holds that $\| \mathbf{K} \|_F=D(\mathbf{\Omega}_{\mc{C}}, F)$ and thus we complete the proof. 
\end{proof}

\subsection{Power Guarantee}\label{app: power proof scheme}

Recall the assumption that $p>n$, for any $j\in \mc{H}$, the size of its neighborhood $N_j$ is larger than or equal to $n$. Therefore, when running Algorithm~\ref{alg: exchangeable}, the $j$th column of every generated data matrix remains the same as the observed one. In other words, $\mathbf{\tilde{X}}_{,\mc{H}}^{(m)}=\mathbf{X}_{,\mc{H}}$.

Denote by $\Pi_1$ the projection matrix onto the column space of the matrix $[\mathbf{1}_n, \mathbf{X}_{,\mc{H}}]$. 
Recall that $n\geq 2+|\mc{H}|$, we have the decomposition
$\Pi_2:=\mathbf{I}_n-\Pi_1=\Gamma\Gamma^\top$ where $\Gamma$ is an orthonormal matrix of dimension $n\times (n-1-|\mc{H}|)$. 
Let $\mathbf{W}=\Gamma^\top\mathbf{X}_{,\mc{C}}$, where $\mc{C}=\mc{H}^c$.
Then we can express the data matrix as
$$
\mathbf{X}_{,\mc{C}}=\Pi_1 \mathbf{X}_{,\mc{C}} + \Gamma \mathbf{W}, 
$$
and express the partial residual correlation statistics as 
$$
\hat{\gamma}_{ij}=\frac{(\mathbf{W}_{,i}^\top \mathbf{W}_{,j})}{\|\mathbf{W}_{,i}\|_2\|\mathbf{W}_{,j}\|_2} = \left\langle \frac{\mathbf{W}_{,i}}{\|\mathbf{W}_{,i}\|_2}, \frac{\mathbf{W}_{,j}}{\|\mathbf{W}_{,j}\|_2} \right\rangle, \qquad i,j\in \mc{C}. 
$$
For the generated copies,  we similarly define $\widetilde{\mathbf{W}}^{(m)}=\Gamma^\top\widetilde{\mathbf{X}}^{(m)}_{,\mc{C}}$ and express the corresponding partial residual correlations as  
$$
\tilde{\gamma}_{ij}^{(m)}= \left\langle \frac{\widetilde{\mathbf{W}}^{(m)}_{,i}}{\|\widetilde{\mathbf{W}}^{(m)}_{,i}\|_2}, \frac{\widetilde{\mathbf{W}}^{(m)}_{,j}}{\|\widetilde{\mathbf{W}}^{(m)}_{,j}\|_2} \right\rangle, \qquad i,j\in \mc{C}. 
$$

Our first lemma characterizes the joint distribution of $\{\hat{\gamma}_{ij},\tilde{\gamma}_{ij}^{(m)}: i,j \in \mc{C}, m=1,\ldots, M\}$.

\begin{lemma}\label{lemma: joint distribution W ratio}
Consider the notations and assumptions in Section~\ref{sec: power theory}.
Consider independent random matrices $\mathbf{U}$ and $\mathbf{\tilde{U}}^{(m)}$ ($m=1,2,\ldots, M)$ constructed as follows. 
$\mathbf{U}$ is a $(n-1-|\mc{H}|)\times q$ matrix with i.i.d. rows sampled from $N_{q}(\mathbf{0}, \mathbf{R}(\mathbf{\Omega}_{\mc{C}}^{-1}))$, and $\mathbf{\tilde{U}}^{(m)}$ is a $(n-1-|\mc{H}|)\times q$ matrix with i.i.d. rows sampled from $N_{q}(\mathbf{0}, \mathbf{I}_{q})$. 
Then the joint distribution of 
$$
\left(\frac{\mathbf{W}_{,i}}{\|\mathbf{W}_{,i}\|_2}\right)_{i\in \mc{C}}, 
\left(\frac{\widetilde{\mathbf{W}}^{(m)}_{,i}}{\|\widetilde{\mathbf{W}}^{(m)}_{,i}\|_2}\right)_{i\in \mc{C}}, m=1, \ldots, M, 
$$
is the same as the joint distribution of  
$$
\left(\frac{\mathbf{U}_{i}}{\|\mathbf{U}_{i}\|}\right)_{i\leq q},
\left(\frac{\mathbf{\tilde{U}}^{(m)}_{i}}{\|\mathbf{\tilde{U}}^{(m)}_{i}\|}\right)_{i\leq q}, m=1,\ldots, M.
$$
Furthermore, if $\lambda_{\max}(\mathbf{\Omega}_{\mc{C}})/\lambda_{\min}(\mathbf{\Omega}_{\mc{C}})\leq b_0$, then the spectral norm of $\mathbf{R}(\mathbf{\Omega}_{\mc{C}}^{-1}))$ is bounded by $b_0$. 
\end{lemma}

We defer the proof of Lemma~\ref{lemma: joint distribution W ratio} in Appendix~\ref{sec: pf W ratio} but discuss the implication. 
Lemma~\ref{lemma: joint distribution W ratio} allows us to analyze the distribution of the statistics $T(\mathbf{X})$ and $T(\mathbf{\tilde{X}}^{(m)})$ using an equivalent version of joint distribution defined using $\mathbf{U}$ and $\mathbf{\tilde{U}}^{(m)}$, provided that if the function $T(\mathbf{x})$ only depends on $\Gamma^\top \mathbf{x}$. 

In particular, the distributional equivalence guaranteed by Lemma~\ref{lemma: joint distribution W ratio} implies we can represent $\hat{\gamma}_{ij}$ as 
$$
\hat{\gamma}_{ij}
\overset{d}{=}
\left\langle\frac{\mathbf{U}_{i}}{\|\mathbf{U}_{i}\|},\frac{\mathbf{U}_{j}}{\|\mathbf{U}_{j}\|}\right\rangle
$$
and similarly for $\tilde{\gamma}_{ij}^{(m)}$ using $\mathbf{\tilde{U}}^{(m)}$. 
In the following proofs, we abuse the notation $T(\mathbf{U})$ to denote the statistic computed by substituting  $\hat{\gamma}_{ij}$ by $\left\langle\frac{\mathbf{U}_{i}}{\|\mathbf{U}_{i}\|},\frac{\mathbf{U}_{j}}{\|\mathbf{U}_{j}\|}\right\rangle$ and $T(\mathbf{\tilde{U}}^{(m)})$ to denote the statistic computed by substituting $\hat{\gamma}_{ij}$ by $\left\langle\frac{\mathbf{\tilde{U}}^{(m)}_{i}}{\|\mathbf{\tilde{U}}^{(m)}_{i}\|},\frac{\mathbf{\tilde{U}}^{(m)}_{j}}{\|\mathbf{\tilde{U}}^{(m)}_{j}\|}\right\rangle$.

Before we get into the proof, we provide an overview of our strategy. 

We will find some constant $t$ depending on the choice of test statistic and we will write 
\begin{align}
\label{eq: rejection probability general}   & \quad \mathbb P(\text{Reject }H_0)\\
   &=\mathbb P \left( \frac{1}{M+1}\left[1+\sum_{m=1}^M \mathbf{1}\left\{T\left(\widetilde{\mathbf{U}}^{(m)}\right) \geq T(\mathbf{U})\right\}\right]\leq \alpha   \right)\nonumber \\
    &\geq \mathbb P \left( T(\mathbf{U}) > t  ,  \frac{1}{M+1}\left[1+\sum_{m=1}^M \mathbf{1}\left\{T\left(\widetilde{\mathbf{U}}^{(m)}\right) \geq T(\mathbf{U})\right\}\right]\leq \alpha  \right)\nonumber \\
    &=\mathbb P(T(\mathbf{U}) >t )  -\mathbb P\left( T(\mathbf{U}) > t,\frac{1}{M+1}\left[1+\sum_{m=1}^M \mathbf{1}\left\{T\left(\widetilde{\mathbf{U}}^{(m)}\right) \geq T(\mathbf{U})\right\}\right]>\alpha\right)\nonumber \\
    &\geq \mathbb P(T(\mathbf{U}) >t) -\mathbb P \left( \frac{1}{M+1}\left[1+\sum_{m=1}^M \mathbf{1}\left\{T\left(\widetilde{\mathbf{U}}^{(m)}\right) \geq t\right\}\right]>\alpha \right), \nonumber 
\end{align}
where the last inequality is because if $b>c$ then $\mathbf{1}\left\{a\geq b\right\}\leq \mathbf{1}\left\{a\geq c\right\}$. 

Our proof will be dedicated to (1) derive an upper bound on 
\begin{equation}\label{eq: null tail}
    \mathbb P \left( \frac{1}{M+1}\left[1+\sum_{m=1}^M \mathbf{1}\left\{T\left(\widetilde{\mathbf{U}}^{(m)}\right) \geq t\right\}\right]>\alpha \right),
\end{equation}
and (2) derive a lower bound for $\mathbb P(T(\mathbf{U}) >t) $. 

The following lemma is useful in the first task. 

\begin{lemma}\label{lem: null tail bound}
If $M>2/\alpha$ and 
if the probability $\theta=P\left(T\left(\widetilde{\mathbf{U}}^{(m)}\right) \geq t_n\right)$ satisfies $2(\theta+\sqrt{\theta})<\alpha$ and $\theta<0.5$, then 
$$
P \left( \frac{1}{M+1}\left[1+\sum_{m=1}^M \mathbf{1}\left\{T\left(\widetilde{\mathbf{U}}^{(m)}\right) \geq t_n\right\}\right]>\alpha \right) \leq e^{-M}. 
$$
\end{lemma}

For the second task, 
it is sufficient to find a value $t_n$ and derive a lower bound on the probability $P(T(\mathbf{U}) >t_n)$ 
in the proof of either the dense or sparse alternative.

\subsubsection{Proof for dense alternative}

\begin{proof}[Proof of Theorem \ref{thm:power_dense}]

We first consider the test statistic $T_1$. 
Following the argument in the last subsection, we consider
$$
T_1(\mathbf U)= \sum_{i,j\in \mc{C}, i\neq j} \left\langle\frac{\mathbf{U}_{i}}{\|\mathbf{U}_{i}\|},\frac{\mathbf{U}_{j}}{\|\mathbf{U}_{j}\|}\right\rangle^2. 
$$

For this statistic, we first state a result regarding its asymptotic distribution and its first two moments.

\begin{theorem}\label{thm: dense asymptotics}
Suppose $\mathbf{U}$ is a $n \times q$ random matrix with i.i.d. rows sampled from $N_q(\mathbf{0}, \mathbf{R}_n)$ where $\mathbf{R}_n$ is a correlation matrix with spectral norm bounded by a constant $b_0$. 
Suppose $q$ increases along with $n$ and $\lim \frac{q}{n}\in (0, \infty)$. 
Let $\mu_{n}=\mathbb{E} T_1(\mathbf{U})$. 
We can choose a sequence of positive numbers $\sigma_{n}^2$ such that the followings hold:
\begin{enumerate}
    \item As $n,q\rightarrow \infty$, it holds that
$$
\frac{T_1(\mathbf{U})-\mu_{n}}{\sigma_{n}} \xrightarrow{D} N(0,1);
$$
\item When $\mathbf{R}_n =\mathbf I_q$, we have
$$
\mu_n = \frac{q(q-1)}{n}, \qquad \sigma_{n}^2= (2q/n)^2. 
$$

\item For a general $\mathbf{R}_n$, 
$$
\mu_n\geq \frac{q(q-1)}{n} + \left(1-\frac{4}{n}\right)\|\mathbf R_n -\mathbf I_q\|^2_{F}
$$
and 
$$
\sigma_{n}^2\leq b_0^4\left[ (2 q/n)^2 + 8 n^{-1} \|\mathbf R_n -\mathbf I_q\|^2_{F} \right]. 
$$
\end{enumerate}
\end{theorem}
We defer the proof of Theorem~\ref{thm: dense asymptotics} at the end of this subsection.

Denote by $z$ the $ (\alpha^2/16)$-upper quantile of the standard normal distribution. 
For any $n$, we take $$t_n= z \frac{2 q}{n} + \frac{q(q-1)}{n}.$$
Theorem~\ref{thm: dense asymptotics} implies that 
$$
\mathbb{P}(T_1(\mathbf{\tilde{U}})\geq t_n)\rightarrow \alpha^2/16, \quad (q, n\rightarrow \infty). 
$$

Since $2(\alpha^2/16 + \sqrt{\alpha^2/16})<\alpha$, when $n$ is sufficiently large, the probability $\theta=\mathbb P\left(T_1\left(\widetilde{\mathbf{U}}^{(m)}\right) \geq t_n\right)$ satisfies $2(\theta+\sqrt{\theta})<\alpha$ and $\theta<0.5$.

Therefore, we can apply Lemma~\ref{lem: null tail bound} to get 
\begin{align}
  \mathbb P \left( \frac{1}{M+1}\left[1+\sum_{m=1}^M \mathbf{1}\left\{T_1\left(\widetilde{\mathbf{U}}^{(m)}\right) \geq t_n\right\}\right]>\alpha \right)
    &\leq e^{-M}\leq \epsilon/2,  \label{eq: null bound}
\end{align} 
since $M>\log(2\epsilon^{-1})$.

Recall that the rows of $\mathbf{U}$ are i.i.d. sampled from $N_{q}(\mathbf{0}, \mathbf{R}(\mathbf{\Omega}_{\mc{C}}^{-1}))$. 
Let $\mu_n$ and $\sigma_n$ be the quantities define in Theorem~\ref{thm: dense asymptotics} with $\mathbf{R}_n$ being $\mathbf{R}(\mathbf{\Omega}_{\mc{C}}^{-1})$.
If $(1-4/n) b^2 > 2 z $, then Theorem~\ref{thm: dense asymptotics} implies that $\mu_n\geq t_n$. 
Furthermore, if we denote by $\delta=\|\mathbf R_n -\mathbf I_q\|^2_{F}/(q/n)$, we have $\delta\geq b^2$ and 
\begin{align*}
      \quad & \frac{ \mu_n - t_n  }{\sigma_n} \\
    \geq \; & \frac{ \left(1-\frac{4}{n}\right)\delta- 2 z   }{\sigma_n} \times \frac{q}{n}\\
    \geq \; & \frac{ \left(1-\frac{4}{n}\right)\delta- 2 z  }{b_0^4 \left[ (2 q/n)^2 + 8 n^{-1} \delta \right]}\times \frac{q}{n}.
\end{align*}
Note that the right hand side of the above display is an increasing function in $\delta$, so it is lower bounded by 
$$
\frac{ \left(1-\frac{4}{n}\right) b^2 - 2 z  }{b_0^4 \left[ (2 q/n)^2 + 8 n^{-1} b^2 \right]}\times \frac{q}{n}, 
$$
whose limit is 
$$
\frac{ b^2 - 2 z  }{4 b_0^4\gamma}. 
$$
It then follows from the asymptotic normality in Theorem~\ref{thm: dense asymptotics} that
\begin{align*}
   & \quad \liminf_{n} \mathbb P(T_1(\mathbf{U}) >t_n)\\
   &=\liminf_{n}  \mathbb P\left( 
 \frac{T_1(\mathbf{U})-\mu_n}{\sigma_n} > \frac{ t_n - \mu_n }{\sigma_n}\right) \\
 &\geq 1-\Phi\left(-\frac{ b^2 - 2 z  }{4 b_0^4\gamma}\right),  
\end{align*}
where $\Phi$ denotes the standard normal CDF. 
For any given $\epsilon$, if we choose $b^2 > 2z + 4b_0^4 \gamma \tilde{z}$ where $\tilde{z}$ is the $(\epsilon/2)$-upper quantile of the standard normal distribution, we can conclude that 
\begin{align*} 
  & \quad \liminf_{n}  \mathbb P(\text{Reject }H_0)\\
&\geq \liminf_{n} \left\{ \mathbb P(T_1(\mathbf{U}) >t_n) - \mathbb P \left( \frac{1}{M+1}\left[1+\sum_{m=1}^M \mathbf{1}\left\{T_1\left(\widetilde{\mathbf{U}}^{(m)}\right) \geq t_n\right\}\right]>\alpha \right) \right\}\\
&\geq 1-\epsilon/2 - \epsilon/2, 
\end{align*}
where we have used \eqref{eq: null bound}. 

\end{proof}

\begin{proof}[Proof of Theorem~\ref{thm: dense asymptotics}]
The asymptotic normality and the expression for $\sigma_n^2$ are adapted from Theorem 2.2 in \cite{zheng2019test} with $\mathbf{R}_{*}$ there being $\mathbf{I}_q$ and $\mathbf{R}$ being $\mathbf{R}_n$. 
Furthermore, using the elementary identities that $\operatorname{tr}(\mathbf{A} \operatorname{diag}(\mathbf{B}))=\operatorname{tr}(\operatorname{diag}(\mathbf{A}) \operatorname{diag}(\mathbf{B}))=\operatorname{tr}(\operatorname{diag}(\mathbf{A}) \mathbf{B})$ and $\operatorname{diag}(\mathbf{R})=\mathbf{I}_q$,
we can reduce the expression of the asymptotic variance given in \cite[Section~A.1]{zheng2019test}  to 
$$
\begin{aligned}  \sigma_{n}^2= & 4 n^{-1}\left[ 2 \operatorname{tr}\left(\mathbf{R}^4\right)+n^{-1}\left[\operatorname{tr}\left(\mathbf{R}^2\right)\right]^2 \right.\\ 
& \qquad\quad +2 \operatorname{tr}\left(\operatorname{diag}\left(\mathbf{R}^2\right) \mathbf{R}^2 \operatorname{diag}\left(\mathbf{R}^2\right)\right) \\ 
&\left.  \qquad\quad -4 \operatorname{tr}\left(\mathbf{R}^2 \operatorname{diag}\left(\mathbf{R}^3\right)\right)
\right].
\end{aligned}
$$
The last expression can be further reduce as follows:
\begin{align}\label{eq: dense variance expression 2}
&  8^{-1} n \sigma_{n}^2-(2 n)^{-1}\left[\operatorname{tr}\left(\mathbf{R}^2\right)\right]^2
\nonumber \\
 = & \operatorname{tr}\left(\mathbf{R}^4\right)+\operatorname{tr}\left(\operatorname{diag}\left(\mathbf{R}^2\right)^2\mathbf{R}^2\right)-2 \operatorname{tr}\left(\operatorname{diag}\left(\mathbf{R}^2\right) \operatorname{diag}\left(\mathbf{R}^3\right)\right) \nonumber\\ 
  = & \operatorname{tr}\left(\left[\mathbf{R}^2-\operatorname{diag}\left(\mathbf{R}^2\right)\right]^2\right) + 2\operatorname{tr}\left(\mathbf{R}^2\operatorname{diag}\left(\mathbf{R}^2\right)\right) -\operatorname{tr}\left(\operatorname{diag}\left(\mathbf{R}^2\right)^2\right) \nonumber\\
  & 
  -2 \operatorname{tr}\left(\operatorname{diag}\left(\mathbf{R}^2\right) \operatorname{diag}\left(\mathbf{R}^3\right)\right) +\operatorname{tr}\left(\operatorname{diag}\left(\mathbf{R}^2\right)^2\mathbf{R}^2\right) \nonumber\\
 = & \operatorname{tr}\left(\left[\mathbf{R}^2-\operatorname{diag}\left(\mathbf{R}^2\right)\right]^2\right)\nonumber\\
& + \operatorname{tr}\left(\mathbf{R}^2\operatorname{diag}\left(\mathbf{R}^2\right)\right) -2 \operatorname{tr}\left(\operatorname{diag}\left(\mathbf{R}^2\right) \mathbf{R}^3\right)+ \operatorname{tr}\left(\mathbf{R}^4\operatorname{diag}\left(\mathbf{R}^2\right)\right)\nonumber \\
  &  -\operatorname{tr}\left(\mathbf{R}^4\operatorname{diag}\left(\mathbf{R}^2\right)\right) +\operatorname{tr}\left(\operatorname{diag}\left(\mathbf{R}^2\right)^2\mathbf{R}^2\right)
 \nonumber \\
=& \operatorname{tr}\left(\left[\mathbf{R}^2-\operatorname{diag}\left(\mathbf{R}^2\right)\right]^2\right)+\operatorname{tr}\left(\mathbf{R}^2[\mathbf{R}-\mathbf{I}]^2 \operatorname{diag}\left(\mathbf{R}^2\right)\right)\nonumber\\
& -\operatorname{tr}\left(\mathbf{R}^2\left[ \mathbf{R}^2-\operatorname{diag}\left(\mathbf{R}^2\right)\right]\operatorname{diag}\left(\mathbf{R}^2\right)\right)\nonumber \nonumber\\
=& \operatorname{tr}\left(\left[\mathbf{R}^2-\operatorname{diag}\left(\mathbf{R}^2\right)\right]^2\right)+\operatorname{tr}\left(\mathbf{R}^2[\mathbf{R}-\mathbf{I}]^2 \operatorname{diag}\left(\mathbf{R}^2\right)\right)\nonumber\\
& -\operatorname{tr}\left(\left[ \mathbf{R}^2-\operatorname{diag}\left(\mathbf{R}^2\right)\right]^2\operatorname{diag}\left(\mathbf{R}^2\right)\right)\nonumber\\
&+ \operatorname{tr}\left(\left[ \mathbf{R}^2-\operatorname{diag}\left(\mathbf{R}^2\right)\right]\operatorname{diag}\left(\mathbf{R}^2\right)^2\right). 
\end{align}

Since the diagonal entries of $\mathbf{R}^2-\operatorname{diag}\left(\mathbf{R}^2\right)$ are zero while $\operatorname{diag}\left(\mathbf{R}^2\right)^2$ is a diagonal matrix, we have $\operatorname{tr}\left(\left[ \mathbf{R}^2-\operatorname{diag}\left(\mathbf{R}^2\right)\right]\operatorname{diag}\left(\mathbf{R}^2\right)^2\right)=0$. 
Furthermore, the diagonal entries of $\operatorname{diag}\left(\mathbf{R}^2\right)$ are all greater than or equal to $1$. 
Therefore, the right hand side of \eqref{eq: dense variance expression 2} can be bounded as
\begin{align*}
8^{-1} n \sigma_{n}^2-(2 n)^{-1}\left[\operatorname{tr}\left(\mathbf{R}^2\right)\right]^2
  = &~ \operatorname{tr}\left(\mathbf{R}^2[\mathbf{R}-\mathbf{I}]^2 \operatorname{diag}\left(\mathbf{R}^2\right)\right)\nonumber\\
&~ -\operatorname{tr}\left(\left[ \mathbf{R}^2-\operatorname{diag}\left(\mathbf{R}^2\right)\right]^2\left[\operatorname{diag}\left(\mathbf{R}^2\right)-\mathbf{I}\right]\right) \\
\leq &~ \operatorname{tr}\left(\mathbf{R}^2[\mathbf{R}-\mathbf{I}]^2 \operatorname{diag}\left(\mathbf{R}^2\right)\right) \\
\leq &~ \|\mathbf{R}^2\|_{2}\|\operatorname{diag}\left(\mathbf{R}^2\right)\|_2 \operatorname{tr}\left([\mathbf{R}-\mathbf{I}]^2\right) \\
\leq &~ b_0^4\|\mathbf{R}-\mathbf{I}\|_{F}^2,
\end{align*}
since $\|\mathbf{R}\|_2\leq b_0$ by assumption. Furthermore, since $\operatorname{tr}\left(\mathbf{R}^2\right)\leq q \|\mathbf{R}^2\|_2\leq q b_0^2$, we have 
$$
\sigma_n^2\leq 4  b_0^4 (q/n)^2+ 8 n^{-1}b_0^4\|\mathbf{R}-\mathbf{I}\|_{F}^2 . 
$$

It remains to derive the expression for $\mu_{n}=\mathbb{E} T_1(\mathbf{U})$ using properties of normal distributions. By linearity of expectation, we only need to compute the expectation for each $\left\langle\frac{\mathbf{U}_{i}}{\|\mathbf{U}_{i}\|},\frac{\mathbf{U}_{j}}{\|\mathbf{U}_{j}\|}\right\rangle^2$. We denote this by $r^2$, where $r$ is the sample correlation coefficient for a pair of bivariate normal variables with correlation $\rho$. 

When $\mathbf{R}_n =\mathbf I_q$, it is well-known that  $r^2$ follows a Beta($1,n-1$) distribution and the expectation is $1/n$. 
For a general $\mathbf{R}_n$, one can compute the second moment of  using Theorem 4.4.5 in \cite{anderson_introduction_2007} to lower bound $\mathbb{E} r^2$. 
More concretely, $$\operatorname E(r^2)=1-\frac{n-1}{n}(1-\rho^2)\left[1+\sum_{k=1}^{\infty}\frac{\rho^{2k}\cdot \prod_{i=1}^k (2i) }{\prod_{i=1}^k (2i-1+n) } \right].$$ 
For $n>10$ and $|\rho|<1$, we have
\begin{align*}
	\sum_{k=1}^{\infty}\frac{\rho^{2k}\cdot \prod_{i=1}^k (2i) }{\prod_{i=1}^k (2i-1+n) }&=\frac{2\rho^2}{n+1} +\frac{1}{n+1}\sum_{k=2}^{\infty}\frac{\rho^{2k}\cdot \prod_{i=1}^k (2i) }{\prod_{i=2}^k (2i-1+n) }\\
	&\leq \frac{2\rho^2}{n+1} +\frac{1}{n+1}\sum_{k=2}^{\infty}\rho^{2k}\\
	&= \frac{2\rho^2}{n+1} + \frac{\rho^4}{(n+1)(1-\rho^2)}.
\end{align*}
Then we have 
\begin{align*}
	\operatorname E(r^2)&=\frac{1}{n}+\frac{n-1}{n}\rho^2-\frac{n-1}{n}(1-\rho^2)\sum_{k=1}^{\infty}\frac{\rho^{2k}\cdot \prod_{i=1}^k (2i) }{\prod_{i=1}^k (2i-1+n) }\\
	&\geq \frac{1}{n}+\frac{n-1}{n}\rho^2-\frac{n-1}{n}(1-\rho^2)\left[\frac{2\rho^2}{n+1} + \frac{\rho^4}{(n+1)(1-\rho^2)}\right]\\
	&\geq \frac{1}{n} + \rho^2- \frac{4\rho^2}{n}
\end{align*}

Summing up the expectation of $\left\langle\frac{\mathbf{U}_{i}}{\|\mathbf{U}_{i}\|},\frac{\mathbf{U}_{j}}{\|\mathbf{U}_{j}\|}\right\rangle^2$ for all distinct pairs of $(i,j)$, we obtain the desired equation and inequality for $\mu_n$. 
\end{proof}

\subsubsection{Proof for strong alternative}

We consider the test statistic $T_2$. 
Following the argument in the Appendix \ref{app: power proof scheme}, we consider
$$T_2(\mathbf{U})= \max_{i,j\leq q, i\neq j}\left\langle\frac{\mathbf{U}_{i}}{\|\mathbf{U}_{i}\|},\frac{\mathbf{U}_{j}}{\|\mathbf{U}_{j}\|}\right\rangle^2.$$

For this statistic, we first state a result regarding the tail behavior and the concentration of $T_2(\mathbf U)$. In the following, $(q\vee n)$ is a shorthand for $\max(q,n)$.

\begin{theorem}\label{thm: strong asymptotics}
Suppose $\mathbf{U}$ is a $n \times q$ random matrix with i.i.d. rows sampled from $N_q(\mathbf{0}, \mathbf{R}_n)$ where $\mathbf{R}_n$ is a correlation matrix. 
Suppose $q$ increases along with $n$, then the followings hold:

\begin{enumerate}
\item When $\mathbf{R}_n =\mathbf I_q$, for $t\in (0,1)$, we have
$$
\mathbb P\left(T_2(\mathbf{U})>t\right) \leq 6\exp\left(-\frac{nt}{16}+2\log q \right);
$$
 \item When $\mathbf{R}_n \neq \mathbf I_q$, if $|\max_{i<j} |\mathbf{R}_n(i,j)| \geq \sqrt{\frac{64}{n} \log \left(\frac{96 (q\vee n)^2}{\alpha^2}\right)}$, then we have  
$$
 \mathbb P\left(T_2(\mathbf{U})\leq  \frac{16}{n} \log \left(\frac{96 q^2}{\alpha^2}\right) \right) \leq \frac{2}{(q\vee n)}+4\exp\left(-\frac{n}{32}\right) .
$$

\end{enumerate}
\end{theorem}
We defer the proof of Theorem~\ref{thm: strong asymptotics} at the end of this subsection. 

For any $n$, we take $$t_n= \frac{16}{n} \log \left(\frac{96 q^2}{\alpha^2}\right).$$
Recall that the rows of $\mathbf{\tilde{U}}$ are i.i.d. sampled from $N_{q}(\mathbf{0}, \mathbf{I}_q)$. 
Statement 1 in Theorem~\ref{thm: strong asymptotics} implies that 
$$
\mathbb{P}(T_2(\mathbf{\tilde{U}})\geq t_n)\leq \alpha^2/16. 
$$
Since $2(\alpha^2/16 + \sqrt{\alpha^2/16})<\alpha$, the probability $\theta=\mathbb P\left(T_2\left(\widetilde{\mathbf{U}}^{(m)}\right) \geq t_n\right)$ satisfies $2(\theta+\sqrt{\theta})<\alpha$ and $\theta<0.5$.

Therefore, we can apply Lemma~\ref{lem: null tail bound} to get
\begin{align}
  \mathbb P \left( \frac{1}{M+1}\left[1+\sum_{m=1}^M \mathbf{1}\left\{T_2\left(\widetilde{\mathbf{U}}^{(m)}\right) \geq t_n\right\}\right]>\alpha \right)
    &\leq e^{-M}\leq \epsilon/2, \label{eq: null bound sparse}
\end{align} 
since $M>\log(2\epsilon^{-1})$.

Recall that the rows of $\mathbf{U}$ are i.i.d. sampled from $N_{q}(\mathbf{0}, \mathbf{R}(\mathbf{\Omega}_{\mc{C}}^{-1}))$ and we are about to apply Statement 2 in Theorem~\ref{thm: strong asymptotics}. 
Consider $n>\max(10/\alpha, 32\log(16/\eps),8/\eps)$. 
We have $\log \left(\frac{96 (q\vee n)^2}{\alpha^2}\right)\leq 4\log(q\vee n)$. 
As a result, for $\mathbf{R}_n$ being $\mathbf{R}(\mathbf{\Omega}_{C,C}^{-1})$, we have $\max_{i<j} |\mathbf{R}_n(i,j)|\geq 16\sqrt{\frac{\log (q\vee n)}{n}}\geq \sqrt{\frac{64}{n} \log \left(\frac{96 (q\vee n)^2}{\alpha^2}\right)}$. 
Theorem~\ref{thm: strong asymptotics} implies that  
\begin{equation}\label{eq: sparse tail bound}
P(T_2(\mathbf{U}) < t_n)\leq  \frac{2}{(q\vee n)}+4\exp\left(-\frac{n}{32}\right) \leq \eps/4. 
\end{equation}
Combining the inequalities in (\ref{eq: rejection probability general}, \ref{eq: null bound sparse},\ref{eq: sparse tail bound}), we have 
\begin{align*} 
    \mathbb P(\text{Reject }H_0) &\geq 1-\epsilon/2 - \epsilon/2=1-\epsilon.
\end{align*}

\begin{proof}[Proof of Theorem~\ref{thm: strong asymptotics}]

When $\mathbf{R}_n =\mathbf I_q$,  we first 
introduce an event which will happen with high probability. 
For any two distinct column indices $i$ and $j$, define
\begin{align*}
    &\mathcal{E}_1(i,j)=\{|\|\mathbf U_i\|^2-n|\leq \sqrt{2}/2n,\|\mathbf U_j\|^2-n|\leq \sqrt{2}/2n\}.
\end{align*}
According to Statement 1 in Lemma \ref{lemma:concentration},  we have $\mathbb P(\mathcal{E}_1(i,j))\geq 1-4\exp(-n/16)$.
For any $t\in (0,1)$, we have
\begin{align*}
    \mathbb P\left( T_2(\mathbf U)>t \right)
    &\leq \sum_{1\leq i<j\leq q}\mathbb P\left( \frac{(\mathbf{U}_{i}^\top \mathbf{U}_{j})^2}{\|\mathbf{U}_{i}\|^2_2\|\mathbf{U}_{j}\|^2_2}  >t\right)\\
    &\leq  \sum_{1\leq i<j\leq q}\left[\mathbb P\left( \frac{(\mathbf{U}_{i}^\top \mathbf{U}_{j})^2}{\|\mathbf{U}_{i}\|^2_2\|\mathbf{U}_{j}\|^2_2}  >t,\mathcal{E}_1(i,j)\right)+ \mathbb P \left(\mathcal{E}_1(i,j)^c\right)\right]\\
    &\stackrel{(1)}{\leq} 2q^2 \mathbb P\left(\frac{1}{n^2}(\mathbf{U}_{1}^\top \mathbf{U}_{2})^2>\frac{1}{2}t \right)+2q^2\exp(-n/16)\\
    &= 2q^2\mathbb P\left(\left|\frac{1}{n}\sum_{k=1}^n U_{k,1}U_{k,2}  \right| >\sqrt{\frac{1}{2}t}  \right)+2q^2\exp(-n/16)\\
    &\stackrel{(2)}{\leq}  4\exp\left(-\frac{nt}{16}+2\log q \right)+2\exp\left(-\frac{n}{16}+2\log q \right) \\
    &\leq 6\exp\left(-\frac{nt}{16}+2\log q \right),
\end{align*}
where (1) is due to the symmetry across the pairs of $(i,j)$ and (2) follows from Statement 2 in Lemma \ref{lemma:concentration} with $\rho=0$.

For a general $\mathbf{R}_n$, suppose $i_*\in [q], j_*\in [q], i_*\neq j_*$ such that  $|\mathbf{R}_n(i_*,j_*)|=\max_{i<j} |\mathbf{R}_n(i,j)| $. 
We introduce the following events:
\begin{align*}
    &\mathcal{E}_1 =\{|\|\mathbf U_{i_*}\|^2-n|\leq 1/2n,\|\mathbf U_{j_*}\|^2-n|\leq 1/2n\},\\
    &\mathcal{E}_2 =\left\{\left| \frac{(\mathbf{U}_{i_*}^\top \mathbf{U}_{j_*})}{n}-\mathbf{R}_n(i_*,j_*)\right|\leq \sqrt{\frac{8\log (q\vee n) }{n}} \right\}.
\end{align*}
Based on Lemma \ref{lemma:concentration}, we have $\mathbb P (\mathcal{E}_1)\geq 1-4\exp(-n/32)$ and $\mathbb P (\mathcal{E}_2)\geq 1-2/(q\vee n)$. 
Note that the event $\mathcal{E}_1$ implies $\left\|\mathbf{U}_{i_*}\right\|_2\left\|\mathbf{U}_{j_*}\right\|_2\leq 3/2$. 
When both $\mathcal{E}_1$ and $\mathcal{E}_2$ occurs, 
we have 
\begin{align*}
\left|\frac{\mathbf{U}_{i_*}^{\top} \mathbf{U}_{j_*}}{\left\|\mathbf{U}_{i_*}\right\|_2\left\|\mathbf{U}_{j_*}\right\|_2}\right| 
&\geq \frac{2}{3}\left|\mathbf{U}_{i_*}^{\top} \mathbf{U}_{j_*}\right|\\
&\geq \frac{2}{3} \left(\mathbf{R}_n(i_*,j_*) - \sqrt{\frac{8\log (q\vee n)}{n}}\right)\\
&> \sqrt{\frac{16}{n} \log \left(\frac{96 q^2}{\alpha^2}\right)},
\end{align*}
where the last inequality is because by assumption  $\mathbf{R}_n(i_*,j_*)\geq 8\sqrt{\frac{1}{n}\log \left(\frac{96 (q\vee n)^2}{\alpha^2}\right)} > \sqrt{\frac{8\log (q\vee n)}{n}}+ \frac{3}{2}\sqrt{\frac{16}{n} \log \left(\frac{96 q^2}{\alpha^2}\right)}$. 
Therefore, we have
\begin{align*}
    &  \mathbb P \left(T_2(\mathbf{U}) \leq \frac{16}{n} \log \left(\frac{96 q^2}{\alpha^2}\right)  \right)\\
    &\stackrel{(1)}{\leq}\mathbb P \left(\left| \frac{(\mathbf{U}_{i_*}^\top \mathbf{U}_{j_*})}{\|\mathbf{U}_{i_*}\|_2\|\mathbf{U}_{j_*}\|_2}\right|  \leq  \sqrt{\frac{16}{n} \log \left(\frac{96 q^2}{\alpha^2}\right)} \right)\\
     &\leq \mathbb P( \mathcal{E}^c_1\cup \mathcal{E}^c_2)\\ 
     &\leq \mathbb P(\mathcal{E}^c_1)+\mathbb P(\mathcal{E}^c_2) \leq 4e^{-n/32}+2/(q\vee n), 
\end{align*}
where the inequality (1) follows from $T_2(\mathbf{U})\geq \left( \frac{(\mathbf{U}_{i_*}^\top \mathbf{U}_{j_*})}{\|\mathbf{U}_{i_*}\|_2\|\mathbf{U}_{j_*}\|_2}\right)^2$.

\end{proof}

\begin{lemma}   
\label{lemma:concentration}
Let $\bs{X}=\left(X_1, X_2, \ldots, X_n\right)$ and $\bs{Y}=\left(Y_1, Y_2, \ldots, Y_n\right)$ be $n$-dimensional random vectors. Assume that each pair $\left(X_i, Y_i\right)$ for $i=1,2, \ldots, n$, is independently sampled from a zero mean bivariate normal distribution. Specifically, $\operatorname{Var}\left(X_i\right)=\operatorname{Var}\left(Y_i\right)=1, \mathbb E(X_iY_i)=\rho$. Then we have:  
 \begin{enumerate} 
     \item For $\bs{W}$ equals to either $\bs{X}$ or $\bs{Y}$, \begin{equation}\label{lemma:col_norm_concentrate}
         \mathbb P\left(\left|\frac{\|\bs{W}\|^2_2}{n}-1 \right|>t \right)\leq 2\exp\left(-\frac{nt^2}{8} \right)\quad \text{ for all }  t\in(0,1) 
     \end{equation}
     \item  
     \begin{equation}\label{lemma:bound_inner_product}
         \mathbb P\left(\left|\frac{1}{n} \sum_{i=1} X_iY_i-\rho\right| \geq t \right)\leq 2\exp\left(-\frac{nt^2}{8} \right)\quad \text{ for all }  t\in(0,1) 
     \end{equation}
 \end{enumerate}
    
\end{lemma}

\begin{proof}[Proof of Lemma \ref{lemma:concentration}]
    Direct consequence of $\mathbb E(e^{\lambda (W_i^2-1)})\leq e^{2\lambda^2}$  and $\mathbb E(e^{\lambda (X_iY_i-\rho)})\leq e^{2\lambda^2}$, for all $|\lambda|<1/4$.
\end{proof}

\subsubsection{Proof for union alternative}

\begin{proof}[Proof of Theorem~\ref{thm:power_union}]
    Denote by $z $ the $( \alpha^2/16)$-upper quantile of the standard normal distribution.  Denote by $u_0$ the solution of  $\exp \left[-(8 \pi)^{-1 / 2} \exp \left(-u_0  / 2\right)\right]=1-\alpha^2 / 16$. Note that both $z$ and $u_0$ are constants.
 
For any $n$, we take $$t_n=\max\{z,u_0\}.$$

Theorem~\ref{thm: dense asymptotics} and Theorem 3 in \cite{cai2011limiting}  implies that 
\begin{align*}
   & \mathbb P\left(\frac{n}{2q} \left(T_1(\widetilde{\mathbf{U}}^{(m)})-q(q-1)/n\right) \geq z\right) \rightarrow \frac{\alpha^2}{16}, \quad (q, n\rightarrow \infty),\\
    &\mathbb P\left( nT_2(\widetilde{\mathbf{U}}^{(m)})-4\log q+\log\log q   \geq u_0\right)\rightarrow \frac{\alpha^2}{16}, \quad (q, n\rightarrow \infty).\\
\end{align*}
Since $2(\alpha^2/8 + \sqrt{\alpha^2/8})<\alpha$, when $n$ is sufficiently large together, we can use the union bound to see that the probability $\theta=\mathbb P\left( T_3(\widetilde{\mathbf{U}}^{(m)})  \geq t_n\right)$ satisfies $2(\theta+\sqrt{\theta})<\alpha$ and $\theta<0.5$. 
Therefore, we can apply Lemma~\ref{lem: null tail bound} to get 
\begin{align*}
  \mathbb P \left( \frac{1}{M+1}\left[1+\sum_{m=1}^M \mathbf{1}\left\{T_3\left(\widetilde{\mathbf{U}}^{(m)}\right) \geq t_n\right\}\right]>\alpha \right)
    &\leq e^{-M}\leq \epsilon/2,  \label{eq: null bound}
\end{align*} 
since $M>\log(2\epsilon^{-1})$.  

Next, we will show 
\begin{equation}\label{eq:T3U_power}
    \liminf_{n} \mathbb P(T_3(\mathbf{U}) >t_n)\geq 1-\epsilon/2.
\end{equation}

Since it is assumed that $\mathbf{\Omega}\in \mathbf{\Theta}_{n3}(b)=\mathbf{\Theta}_{n1}(b)\cup \mathbf{\Theta}_{n2}(16)$, we can divide the sequence of populations into two subsequences (possibly finite) so that either of the followings \begin{itemize}
    \item \textbf{Case 1:} $\mathbf{\Omega}\in \mathbf{\Theta}_{n1}(b)$,
    \item \textbf{Case 2:} $\mathbf{\Omega}\in \mathbf{\Theta}_{n2}(16)$
\end{itemize}
holds for the whole subsequence.
Therefore, it suffices to prove \eqref{eq:T3U_power} for both cases.

For Case 1, following a similar argument as in the proof of Theorem~\ref{thm:power_dense}, and setting $b=\sqrt{2\max\{z,u_0\} + 4 b_0^4 \gamma z_2}$, where $z_2$ is the $(\epsilon/2)$-upper quantile of the standard normal distribution, we can conclude that 
$$  \liminf_{n} \mathbb P(T_3(\mathbf{U}) >t_n)\geq \liminf_{n} \mathbb P\left(\frac{n}{2q}\left(T_1( \mathbf{U} )-q(q-1)/n\right) >\max\{z,u_0\}\right)  \geq  1-\epsilon/2. $$
For Case 2, 
\begin{align*}
     \liminf_{n} \mathbb P(T_3(\mathbf{U}) >t_n)&\geq \liminf_{n} \mathbb P\left(nT_2(\mathbf{U} )-4\log q+\log\log q >\max\{z,u_0\}\right) \\
     &=\liminf_{n}\mathbb P\left(T_2(\mathbf{U} )\geq \frac{\max\{z,u_0\}+4\log q-\log\log q }{n}  \right).
\end{align*}
Recall from Theorem~\ref{thm:power_sparse}, we have
$$\liminf_{n} \mathbb P \left(T_2(\mathbf{U}) \geq \frac{16}{n} \log \left(\frac{96 q^2}{\alpha^2}\right)  \right)\geq  1-\epsilon/2.$$
Then, it suffices to show that 
$$ \frac{16}{n} \log \left(\frac{96 q^2}{\alpha^2}\right)\geq \frac{\max\{z,u_0\}+4\log q-\log\log q}{n},$$
which is a direct consequence of Lemma~\ref{lemma: order of z and u_0}.
\end{proof}

\begin{lemma}\label{lemma: order of z and u_0}

Let $0<\alpha<1$. Denote by $z $ the $(\alpha^2/16)$-upper quantile of the standard normal distribution. and let $u_0$ be the solution of  $\exp \left[-(8 \pi)^{-1 / 2} \exp \left(-u_0  / 2\right)\right]=1-\alpha^2 / 16$. Then we have
     \begin{align}
         & u_0\leq -2\log\left(\frac{\alpha^2}{16}\right),\label{eq:u0 order}\\
         & z\leq \sqrt{2\log\left( \frac{16}{\alpha^2\sqrt{2\pi}}\right)}
     \end{align}
\end{lemma}

\begin{proof}[Proof of Lemma \ref{lemma: order of z and u_0}]

First, we solve the equation $\exp \left[-(8 \pi)^{-1 / 2} \exp \left(-u_0 / 2\right)\right]=1-\alpha^2 / 16$. This yields

$$
u_0=-2 \log \left(-\log \left(1-\frac{\alpha^2}{16}\right) \sqrt{8 \pi}\right) .
$$

Using the inequality $\log (1-x) \leq-x$ for $0<x<1$ and the monotonicity of $-\log (x)$, we find that
\begin{align*}
    u_0&=-2\log \left(-\log\left(1-\frac{\alpha^2}{16} \right)  \right)-2\log \left( \sqrt{8\pi}  \right)\\
    &\leq -2\log \left(-\log\left(1-\frac{\alpha^2}{16} \right)  \right) \\
    &\leq -2\log \left( \frac{\alpha^2}{16}   \right). 
\end{align*}
Next, we establish an inequality for the tail probability of the standard normal distribution. Let $X\sim N(0,1)$ and for   $t\geq x>0$, we have that $1\leq t/x$ and 
$$
\mathbb{P}(X>x)=\frac{1}{\sqrt{2 \pi}} \int_x^{\infty} 1 \cdot e^{-t^2 / 2} \mathrm{~d} t \leq \frac{1}{\sqrt{2 \pi}} \int_x^{\infty} \frac{t}{x} e^{-t^2 / 2} \mathrm{~d} t=\frac{e^{-x^2 / 2}}{x \sqrt{2 \pi}} .
$$
Since $\frac{\alpha^2}{16}<1/16$, we have $z>1$, which implies that $\exp(z^2/2)\leq z\exp(z^2/2)\leq \frac{16}{\alpha^2\sqrt{2\pi}} $.

\end{proof}

\subsubsection{Proof of Lemma~\ref{lemma: joint distribution W ratio}}\label{sec: pf W ratio}

The proof is divided into three steps. 

\textbf{Step 1.} We characterize the randomness of any sampled $\mathbf{\tilde{X}}$ when running Algorithm~\ref{alg: residual rotation}. Let $\mathbf{\tilde{W}}:=\Gamma^\top \mathbf{\tilde{X}}_{,\mc{C}}$. Then we have 
$$
\mathbf{\tilde{X}}_{,\mc{C}}=\Pi_1 \mathbf{\tilde{X}}_{,\mc{C}} + \Gamma \mathbf{\tilde{W}}. 
$$
Consider running Algorithm~\ref{alg: residual rotation} with index $i$. (1) if $i\in \mc{H}$, then $|N_i|>n$ and we have $\mathbf{\tilde{X}}_{i}=\mathbf{X}_{i}$. (2) If $i\in \mc{C}$, then $N_i=\mc{H}$, so we have
$$
\bs{F}+\widetilde{\bs{R}}\frac{\|\bs{R}\|}{\|\widetilde{\bs{R}}\| },
$$
where $\bs{F}=\Pi_1 \mathbf{X}_{i}$ and $\|\bs{R}\|=\|\Gamma\mathbf{X}_{i}\|$ both remain unchanged over iterations, and $\|\widetilde{\bs{R}}\|^{-1}\widetilde{\bs{R}}$ is an uniformly distributed unit vector that is orthogonal to the column space of $[\mathbf{1}_n, \mathbf{X}_{,\mc{H}}]$. Therefore, we have 
$$
\mathbf{\tilde{W}}_{i}=\|\bs{R}\| \Gamma^{\tp}\frac{\widetilde{\bs{R}}}{\|\widetilde{\bs{R}}\| }
$$
is uniformly distributed on the sphere in $\mathbb{R}^{n-1-|\mc{H}|}$ with radius $\|\Gamma\mathbf{X}_{i}\|$ and is independent with $\{\mathbf{\tilde{X}}_{j}: j\in \mc{C}\setminus\{i\} \}$. 
It follows that 
$$
\frac{\mathbf{\tilde{W}}_{i}}{\|\mathbf{\tilde{W}}_{i}\|}\sim \text{Unif}(\mathbb{S}^{n-1-|\mc{H}|-1}),
$$
that is, the uniform distribution w.r.t. the Haar measure on the unit the sphere in $\mathbb{R}^{n-1-|\mc{H}|}$ and is independent with all other random variables. 

\textbf{Step 2.}
Now we characterize the sampled copies in Algorithm~\ref{alg:gof}. Since $\mc{I}=(1,2,\ldots, p)$, every column of any copies $\mathbf{\tilde{X}}$ in Step 1 has been updated by Algorithm~\ref{alg: residual rotation} at least one. By the analysis above for Algorithm~\ref{alg: residual rotation}, we know that for every $i\in \mc{C}$, $\frac{\mathbf{\tilde{W}}_{i}}{\|\mathbf{\tilde{W}}_{i}\|}\sim \text{Unif}(\mathbb{S}^{n-1-|\mc{H}|-1})$ and they are mutually independent and independent with all other random variables. 
In fact, if $\mathbf{\tilde{U}}$ is a $(n-1-|\mc{H}|)\times q$ matrix with i.i.d. rows sampled from $N_{q}(\mathbf{0}, \mathbf{I}_{q})$, then 
$$
\left(\frac{\mathbf{\tilde{W}}_{i}}{\|\mathbf{\tilde{W}}_{i}\|}\right)_{i\in \mc{C}} 
\overset{d}{=}
\left(\frac{\mathbf{\tilde{U}}_{i}}{\|\mathbf{\tilde{U}}_{i}\|}\right)_{i\leq q}.
$$

\textbf{Step 3.} 

Since $\Gamma^\top [\mathbf{1}_n, \mathbf{X}_{,\mc{H}}]=0$, it is easy to show that $\mathbf{W}:= \Gamma^\top \mathbf{X}_{,\mc{C}}\sim  N(0,\mathbf{I}_{n-1-|\mc{H}|}\otimes \mathbf{\Omega}_{\mc{C}}^{-1})$. 
Recall the diagonal normalization $\mathbf{R}(\mathbf M)=\text{diag}(\mathbf M)^{-1/2}\mathbf{\mathbf M}\text{ diag}(\mathbf M)^{-1/2}$. 
We claim that if $\mathbf{U}$ is a $(n-1-|\mc{H}|)\times q$ matrix with i.i.d. rows sampled from $N_{q}(\mathbf{0}, \mathbf{R}(\mathbf{\Omega}_{\mc{C}}^{-1}))$, then 
\begin{equation}\label{eq: observed W direction}
\left(\frac{\mathbf{W}_{i}}{\|\mathbf{W}_{i}\|}\right)_{i\in \mc{C}} 
\overset{d}{=}
\left(\frac{\mathbf{U}_{i}}{\|\mathbf{U}_{i}\|}\right)_{i\leq q}.
\end{equation}

This is because for any $q\times q$ diagonal matrix $\mathbf{D}$ with $\mathrm{D}_{i, i}>0$, the value of $\frac{\mathbf{U}_{i}}{\|\mathbf{U}_{i}\|}$ remains invariant under the transformation of $\mathbf{U}$ to $\mathbf{U D}$. In particular, we set $\mathbf{D}=\text{diag}(\mathbf{\Omega}_{\mc{C}}^{-1})^{1/2}$ and  $\mathbf{U}=\mathbf{W}\mathbf{D}^{-1}$ to get \eqref{eq: observed W direction}. 

Lastly, the spectral norm of $\mathbf{R}(\mathbf{\Omega}_{\mc{C}}^{-1}))$ can be  bounded as 
\begin{align*}
& \lambda_{\max}\left(\mathbf{R}(\mathbf{\Omega}_{\mc{C}}^{-1}))\right)\\
\leq & 
\lambda_{\max}\left(\mathbf{\Omega}_{\mc{C}}^{-1}\right)
\lambda_{\max}\left(\mathbf{D}^{-1}\right)^2\\
= & 
\frac{1}{\lambda_{\min}\left(\mathbf{\Omega}_{\mc{C}}\right)}
\max_{i\in \mc{C}}\frac{1}{\left(\mathbf{D}_{ii}\right)^2}\\
= & 
\frac{1}{\lambda_{\min}\left(\mathbf{\Omega}_{\mc{C}}\right)}
\max_{i\in \mc{C}}\left([\mathbf{\Omega}_{\mc{C}}^{-1}]_{ii}^{-1}  \right).
\end{align*}
Using Schur complement, it is easy to see that $[\mathbf{\Omega}_{\mc{C}}^{-1}]_{ii}^{-1}\leq \mathbf{\Omega}_{ii}\leq \lambda_{\max}(\mathbf{\Omega}_{\mc{C}})$ for any $i\in \mc{C}$. Therefore, if $$\lambda_{\max}(\mathbf{\Omega}_{\mc{C}})/\lambda_{\min}(\mathbf{\Omega}_{\mc{C}})\leq b_0,$$ then $\lambda_{\max}\left(\mathbf{R}(\mathbf{\Omega}_{\mc{C}}^{-1}))\right)$ is bounded by $b_0$.

\subsubsection{Proof of Lemma~\ref{lem: null tail bound}}

We first state the following concentration inequality resulting from Kearns--Saul inequality \cite{kearns1998large}. 

\begin{lemma}\label{lem: concentration Bernoulli}
    Suppose $\{\xi_i\}_{i=1}^M$ are i.i.d. Bernoulli random variables with probability $\theta$. 
    If $\theta<0.5$, then it holds that 
    $$
    \mathbb{P}\left( \frac{1}{M}\sum_{i=1}^M \xi_i - \theta \geq x \right)\leq \exp \left( - \theta^{-1} M  x^2 \right), \forall x>0. 
    $$
\end{lemma}

We apply Lemma~\ref{lem: concentration Bernoulli} with $\xi_{m}=\mathbf{1}\left\{T\left(\widetilde{\mathbf{U}}^{(m)}\right) \geq t\right\}$ and $x=[(1+M)\alpha-1]/M - \theta$. 
Since $M>2/\alpha$ and $2(\theta+\sqrt{\theta})<\alpha$, we can show $x^2>\theta$. 
Therefore, we have 
\begin{align*}
  &\quad   \mathbb P \left( \frac{1}{M+1}\left[1+\sum_{m=1}^M \mathbf{1}\left\{T\left(\widetilde{\mathbf{U}}^{(m)}\right) \geq t_n\right\}\right]>\alpha \right) \\
    &=\mathbb{P}\left( \frac{1}{M}\sum_{i=1}^M \xi_i - \theta \geq x \right) \\
    &\leq \exp(-Mx^2/\theta)\\
    &\leq e^{-M}.
\end{align*}
The proof of Lemma~\ref{lem: null tail bound} is completed.

\begin{proof}[Proof of Lemma~\ref{lem: concentration Bernoulli}]
By Lemma 1 in \citet{kearns1998large}, the moment generating function of $(\xi_i-\theta)$ satisfies

$$
\ln \left[(1-\theta) e^{-t \theta}+ \theta e^{\theta(1-\theta)}\right] \leq \frac{t^2}{4} g(\theta), \quad t\in \mathbb{R}, 
$$
where $g(\theta)=\frac{(1-2 \theta)}{  \log ((1-\theta) / \theta)}$. 
By Markov inequality, for any $t > 0$ and $x>0$, we have
$$
    \mathbb{P}\left( \frac{1}{M}\sum_{i=1}^M \xi_i - \theta \geq x \right)\leq  \exp \left( - t x M + M \frac{t^2}{4} g(\theta)\right). 
$$
In particular, taking $t=2x/g(\theta)$, we have 
$$
    \mathbb{P}\left( \frac{1}{M}\sum_{i=1}^M \xi_i - \theta \geq x \right)\leq  \exp \left( - M x^2/g(\theta) \right), \qquad x>0. 
$$
To prove the lemma, it suffices to show $g(\theta)\leq \theta^{-1}$. 
This can be proved using elementary calculus. 
Consider $h(x)=\log((1-x)/x)-x(1-2x)$ for any $x\in (0,1/2]$. We can show that the derivative of $h$ is $-(1-x)^{-1}-x^{-1}-1+4x<0$, since $x(1-x)<1/4<1/3$. So $h(x)$ attains the minimal value $0$ at $x=1/2$. From $\log((1-\theta)/\theta)-\theta(1-2\theta)\geq 0$, we have $\theta g(\theta)\leq 1$ for $\theta\in (0,1/2]$.

\end{proof}

\subsection{Lower bound of separation rate}\label{app: power lower bound}

\begin{proof}[Proof of Theorem~\ref{thm: lower bound dense}]
Denote by $\rho=b/\sqrt{n(q-1)}$. 
Consider the following subset of $\mathbf{\Theta}_{n1}(b)$:
\begin{align}\label{eq:least favorable dense}
\mathbf{\Theta}_{n1}^{*}(b)=&\left\{\mathbf{A}^{(\bs{v})}: 
\mathbf{A}^{(\bs{v})}_{\mc{H},\mc{H}}=\mathbf{I}_{p-q},\mathbf{A}^{(\bs{v})}_{\mc{H},\mc{C}}=\left(\mathbf{A}^{(\bs{v})}_{\mc{C},\mc{H}}\right)^\top=\mathbf{0}_{(p-q)\times q}, \right. \\
& 
\left. \qquad 
\mathbf{A}^{(\bs{v})}_{\mc{C},\mc{C}}=\frac{1}{1-\rho} \mathbf{I}_{q}-[\frac{1}{1-\rho}-\frac{1}{1+(q-1)\rho}]\frac{1}{q} \bs{v} \bs{v}'
,  \bs{v}\in \{-1, 1\}^{q} \right\}. \nonumber
\end{align}
For any $\mathbf{A}^{(\bs{v})}\in \mathbf{\Theta}_{n1}^{*}(b)$, 
it is straightforward to see that $\left[\mathbf{A}^{(\bs{v})}_{\mc{C},\mc{C}}\right]^{-1}=(1-\rho) I_{q}+ \rho \bs{v} \bs{v}^{\prime}$ and thus $D(\mathbf{A}^{(\bs{v})}_{\mc{C},\mc{C}},F)=\sqrt{\rho^2 q(q-1)}=b\sqrt{q/n}$. 
Since the eigenvalues of $\mathbf{A}^{(\bs{v})}_{\mc{C},\mc{C}}$ are either $1/(1-\rho)$ or $1/(1+(q-1)\rho)$, the condition number is $1+q\rho/(1-\rho)$ and is bounded by $1+2b\sqrt{\kappa}$ (for $q$ large, we have $1-\rho>1/2$). 
Therefore, $\mathbf{\Theta}_{n1}^{*}(b)\subset \mathbf{\Theta}_{n1}(b)$.

Let $P_0$ be the probability measure when the rows of $\mathbf{X}$ are i.i.d. $\mathbf{N}_p(\mathbf{0}, \mathbf{I}_p)$ and let $P_{\bs{v}}$ be the probability measure when the rows of $\mathbf{X}$ are i.i.d. $\mathbf{N}_p(\mathbf{0}, \left(\mathbf{A}^{(\bs{v})}\right)^{-1} )$ for any $\bs{v}\in \{-1, 1\}^{q}$. 
Let $f_0$ and $f_{\bs{v}}$ be the respective Lebesgue density functions. 
A direct computation yields
\begin{align*}
\forall \mathbf{x}\in \mathbb{R}^{n\times p}, \qquad    \frac{f_{\bs{v}}(\mathbf{x})}{f_{0}(\mathbf{x})} 
   & = \left(\det \mathbf{A}^{(\bs{v})}_{\mc{C},\mc{C}} \right)^{n/2}\exp \left( -\frac{1}{2}\mathbf{x}_{,\mc{C}}^\top\left( \mathbf{A}^{(\bs{v})}_{\mc{C},\mc{C}} - \mathbf{I}_{q}\right)\mathbf{x}_{,\mc{C}} \right)\\
  & =: g_{\bs{v}}(\mathbf{x}_{,\mc{C}}).
\end{align*}

Let $P_1$ be the mixture measure $P_1=\frac{1}{2^q}\sum_{\bs{v}}P_{\bs{v}}$ and let $f_1$ be the Lebesgue density function. 
Since $f_1=\frac{1}{2^q}\sum_{\bs{v}}f_{\bs{v}}$, we have 
\begin{equation}\label{eq: lower bound dense density ratio}
\forall \mathbf{x}\in \mathbb{R}^{n\times p}, \qquad 
\left(\frac{f_{1}(\mathbf{x})}{f_{0}(\mathbf{x})}\right)^2 = \frac{1}{2^{2q}} \sum_{\bs{u},\bs{v}}  g_{\bs{u}}(\mathbf{x}_{,\mc{C}})g_{\bs{v}}(\mathbf{x}_{,\mc{C}}). 
\end{equation}
The rest of the proof is an adaption of the proof of Theorem~1 in \cite{cai2013optimal} in our setting. 
We highlight the key steps below to be self-contained.

Note that the marginal distribution of $\mathbf{X}_{,\mc{C}}$ under $P_0$ is $N_{n\times q}\left( \mathbf{0}_{q}, \mathbf{I}_{q}\right)$. 
For this marginal distribution, we can follow the proof of Theorem~1 in \cite{cai2013optimal} and use the expression in \eqref{eq: lower bound dense density ratio} to compute
\begin{equation}\label{eq: lower bound dense chisq divergence}
\mathbb{E}_{P_0}\left(\frac{f_{1}(\mathbf{X})}{f_{0}(\mathbf{X})}\right)^2
=
\frac{\left(1-\rho^2\right)^{n-n q / 2}}{\left[1+(q-1) \rho^2\right]^n} \mathbb{E}_{V}\left[1-\left(\frac{q \rho}{1+(q-1) \rho^2}\right)^2\left(\frac{\mathbf{1}^{\prime} V}{q}\right)^2\right]^{-n / 2},
\end{equation}
where the random vector $V\in \{-1,1\}^q$ has i.i.d. symmetric Rademacher entries. 
The condition needed is 
$$
b\leq 1/\sqrt{2\kappa}, \quad b<1. 
$$
This condition implies $2(q-1)>b^2 q \log 2$ when $q$ is large enough. 
Then the proof of Theorem~1 in \cite{cai2013optimal} also suggests that 
\begin{equation}\label{eq: lower bound dense average over V}
\mathbb{E}_{V}\left[1-\left(\frac{q \rho}{1+(q-1) \rho^2}\right)^2\left(\frac{\mathbf{1}^{\prime} V}{q}\right)^2\right]^{-n / 2} 
\leq 
1+\frac{2 b^2 q \log 2}{2(q-1)-b^2 q \log 2}.
\end{equation}
 Since $b^2\leq 1$, we can see by calculus that there is some constant $x_0$ such that for any $x,y\geq x_0$, it holds that 
$$
(1 - b^2 / x)^{-x} < e^{b^2}\sqrt{1+b^2\log 2} ,   \quad (1+ b^2 /y)^y > e^{b^2}/\sqrt{1+b^2\log 2}. 
$$
Therefore, $n\geq 2x_0$, we have 
\begin{equation}
    \label{eq: lower bound dense chisq divergence scalar}
    \frac{\left(1-\rho^2\right)^{n-n q / 2}}{\left[1+(q-1) \rho^2\right]^n}\leq (1+b^2\log 2).
\end{equation}
Combining the inequalities \eqref{eq: lower bound dense chisq divergence scalar} and \eqref{eq: lower bound dense average over V}, we obtain from \eqref{eq: lower bound dense chisq divergence} that 
$$
\mathbb{E}_{P_0}\left(\frac{f_{1}(\mathbf{X})}{f_{0}(\mathbf{X})}\right)^2 -1 
< \frac{8 b^2 q \log 2}{2(q-1)-b^2 q \log 2}. 
$$
Given any $0<\alpha<\beta<1$, there is some $b_0$ such that for all $b\leq b_0$, we have 
\begin{equation}\label{eq: lower bound dense chisq divergence upper bound}
    \mathbb{E}_{P_0}\left(\frac{f_{1}(\mathbf{X})}{f_{0}(\mathbf{X})}\right)^2 -1 < 4(\beta-\alpha)^2.
\end{equation}

Finally, for any test $\phi$, the sum of probabilities of its two types of errors satisfies
$$
\begin{aligned}
\sup_{\bs{v}} \left(\mathbb{E}_{P_0} \phi+\mathbb{E}_{P_{\bs{v}}}(1-\phi)  \right)
& \geq \frac{1}{2^p} \sum_{\bs{v}} \left( \mathbb{E}_{P_0} \phi+\mathbb{E}_{P_{\bs{v}}}(1-\phi) \right)\\
& = \mathbb{E}_{P_0} \phi+ \mathbb{E}_{P_1}(1-\phi)\\
& \geq \inf_{\psi }\left( \mathbb{E}_{P_0} \psi+ \mathbb{E}_{P_1}(1-\psi) \right)\\
& =  1 -  d_{TV}(P_0, P_1),
\end{aligned}
$$
where $d_{TV}(P_0, P_1)$ is the total variation distance, which is in turn bounded by the chi-square divergence as 
$$
d_{TV}^2(P_0, P_1)\leq \frac{1}{4} \left[\mathbb{E}_{P_0}\left(\frac{f_{1}(\mathbf{X})}{f_{0}(\mathbf{X})}\right)^2 -1 \right]< (\beta-\alpha)^2, 
$$
where the last inequality is due to \eqref{eq: lower bound dense chisq divergence upper bound}. 
Hence, for any test $\phi$, we have 
$$\sup_{\bs{v}} \left(\mathbb{E}_{P_0} \phi+\mathbb{E}_{P_{\bs{v}}}(1-\phi)  \right)> 1-\beta+\alpha,$$ which implies 
$$
\inf_{\mathbf{\Omega} \in \mathbf{\Theta}_{n1}(b)} \mathbb{E}(\phi)\leq \inf_{\mathbf{\Omega} \in \mathbf{\Theta}_{n1}^*(b)} \mathbb{E}(\phi) \leq \beta.
$$

\end{proof}


\begin{proof}[Proof of Theorem~\ref{thm: lower bound sparse}]
Our strategy of this proof is the same as the proof of Theorem~\ref{thm: lower bound dense} except that we need to construct another least favorable subset for the sparse alternative and compute the chi-square divergence.

Denote by $\rho=b\sqrt{\frac{\log q}{n}}$. 
Let $\mc{C}^2=\{ (i,j): i<j, ~~, i, j\in \mc{C} \}$. 
For any $(i,j)\in \mc{C}^2$, let $\mathbf{E}_{ij}=\mathbf{e}_{i}\mathbf{e}_{j}' + \mathbf{e}_{j} \mathbf{e}_{i}'$. 
Consider the following subset of $\mathbf{\Theta}_{n2}(b)$:
\begin{align}\label{eq:least favorable sparse}
\mathbf{\Theta}_{n2}^{*}(b)=&\left\{\mathbf{A}^{(ij)}: 
\mathbf{A}^{(ij)}_{\mc{H},\mc{H}}=\mathbf{I}_{p-q},\mathbf{A}^{(ij)}_{\mc{H},\mc{C}}=\left(\mathbf{A}^{(ij)}_{\mc{C},\mc{H}}\right)^\top=\mathbf{0}_{(p-q)\times q}, \right. \\
& 
\left. \qquad 
\mathbf{A}^{(ij)}_{\mc{C},\mc{C}}=\left[\mathbf{I}_{q}+\rho \mathbf{E}_{ij}
\right]^{-1}
,  (ij)\in \mc{C}^2 \right\}. \nonumber
\end{align}
For any $\mathbf{A}^{(ij)}\in \mathbf{\Theta}_{n2}^{*}(b)$, 
it is straightforward to see that $D(\mathbf{A}^{(ij)}_{\mc{C},\mc{C}},\infty)=\rho=b\sqrt{\log(q)/n}$.
Therefore, $\mathbf{\Theta}_{n2}^{*}(b)\subset \mathbf{\Theta}_{n2}(b)$. 

Let $P_0$ be the probability measure when the rows of $\mathbf{X}$ are i.i.d. $\mathbf{N}_p(\mathbf{0}, \mathbf{I}_p)$ and let $P_{ij}$ be the probability measure when the rows of $\mathbf{X}$ are i.i.d. $\mathbf{N}_p(\mathbf{0}, \left(\mathbf{A}^{(ij)}\right)^{-1} )$ for any $(i,j)\in \mc{C}^2$. 
Let $f_0$ and $f_{ij}$ be the respective Lebesgue density functions. 
A direct computation yields
\begin{align*}
\forall \mathbf{x}\in \mathbb{R}^{n\times p}, \qquad    \frac{f_{ij}(\mathbf{x})}{f_{0}(\mathbf{x})} 
   & = \left(\det \mathbf{A}^{(ij)}_{\mc{C},\mc{C}} \right)^{n/2}\exp \left( -\frac{1}{2}\sum_{k=1}^n\mathbf{x}_{k,\mc{C}}^\top\left( \mathbf{A}^{(ij)}_{\mc{C},\mc{C}}-\mathbf{I}_q\right)\mathbf{x}_{k,\mc{C}}  \right)\\
   & = \left(\det \mathbf{A}^{(ij)}_{\mc{C},\mc{C}} \right)^{n/2}\exp \left( -\frac{1}{2}\sum_{k=1}^n\mathbf{x}_{k,\mc{C}}^\top\left( \mathbf{A}^{(ij)}_{\mc{C},\mc{C}}-\mathbf{I}_q\right)\mathbf{x}_{k,\mc{C}}  \right)\\
  & =: g_{ij}(\mathbf{x}_{,(i,j)}), 
\end{align*}
since the only nonzero entries of  $\mathbf{A}^{(ij)}_{\mc{C},\mc{C}}-\mathbf{I}_q$ are on the $(i,j)$ row and the $(i,j)$ columns.

Let $P_1$ be the mixture measure $P_1=\frac{2}{q(q-1)}\sum_{(i,j)\in \mc{C}^2}P_{ij}$ and let $f_1$ be the Lebesgue density function. 
Since $f_1=\frac{2}{q(q-1)}\sum_{(i,j)\in \mc{C}^2}f_{ij}$, we have 
\begin{equation}\label{eq: lower bound sparse density ratio}
\forall \mathbf{x}\in \mathbb{R}^{n\times p}, \qquad 
\left(\frac{f_{1}(\mathbf{x})}{f_{0}(\mathbf{x})}\right)^2 
= \left(\frac{2}{q(q-1)}\right)^2 \sum_{(i,j), (i',j')}  g_{ij}(\mathbf{x}_{,(i,j)})g_{i'j'}(\mathbf{x}_{,(i',j')}). 
\end{equation}

Note that the marginal distribution of $\mathbf{X}_{,\mc{C}}$ under $P_0$ is $N_{n\times q}\left( \mathbf{0}_{q}, \mathbf{I}_{q}\right)$. We will compute the expectation under $P_0$ of each term in \eqref{eq: lower bound sparse density ratio}. 

\textbf{Case 1:}
For any $(i,j)$ and $(i',j')$ without overlap, the submatrices $\mathbf{X}_{,(i,j)}$ and $\mathbf{X}_{,(i',j')}$ are independent under $P_0$, and thus 
$$
\mathbb{E}_{P_0}[g_{ij}(\mathbf{X}_{,(i,j)})g_{i'j'}(\mathbf{X}_{,(i',j')})] = \mathbb{E}_{P_0}[g_{ij}(\mathbf{X}_{,(i,j)})]\mathbb{E}_{P_0}[g_{i'j'}(\mathbf{X}_{,(i',j')})]=1,
$$
since $\mathbb{E}_{P_0}[g_{ij}(\mathbf{X}_{,(i,j)})]=\mathbb{E}_{P_0}[\frac{f_{ij}(\mathbf{X})}{f_{0}(\mathbf{X})}]=1$ for any $(i,j)$.

\textbf{Case 2:}
For $(i,j)=(i',j')$, by property the normal distribution, we can compute  as
\begin{align*}
\mathbb{E}_{P_0}[g_{ij}(\mathbf{X}_{,ij})^2] 
&= \left( \frac{\det (\mathbf{A}^{(ij)}_{\mc{C},\mc{C}})^2}{\det(2\mathbf{A}^{(ij)}_{\mc{C},\mc{C}}-\mathbf{I}_q) } \right)^{n/2} \\
&= \left( \frac{\det (\mathbf{A}^{(ij)}_{\mc{C},\mc{C}})}{\det(2\mathbf{I}_q-\left[\mathbf{A}^{(ij)}_{\mc{C},\mc{C}}\right]^{-1} ) } \right)^{n/2}  \\
&=\left( \frac{1}{\det (\mathbf{I}_q+\rho\mathbf{E}_{ij}))\det(\mathbf{I}_q-\rho\mathbf{E}_{ij}) } \right)^{n/2}\\
&=(1-\rho^2)^{-n}.
\end{align*}

\textbf{Case 3:}
If $(i,j)$ and $(i',j')$ has one common element, (without loss of generality, assume $j=j'$ and $i\neq i'$), we have

\begin{align*}
& \qquad \mathbb{E}_{P_0}[g_{ij}(\mathbf{X}_{,\mc{C}})g_{i'j}(\mathbf{X}_{,\mc{C}})] \\
&= \int_{\mathbb{R}^{n\times q}} \left( (2\pi)^{-q/2}\det (\mathbf{A}^{(ij)}_{\mc{C},\mc{C}})\det (\mathbf{A}^{(i'j)}_{\mc{C},\mc{C}})  \right)^{n/2}\exp \left( -\frac{1}{2}\sum_{k=1}^n\mathbf{x}_{k,}^\top\left(\mathbf{A}^{(ij)}_{\mc{C},\mc{C}}+\mathbf{A}^{(i'j)}_{\mc{C},\mc{C}}-\mathbf{I}_q\right)\mathbf{x}_{k,}  \right)\mathbf{d}\mathbf{x} \\
&= \left( \frac{\det (\mathbf{A}^{(ij)}_{\mc{C},\mc{C}})\det (\mathbf{A}^{(i'j)}_{\mc{C},\mc{C}})}{\det\left(\mathbf{A}^{(ij)}_{\mc{C},\mc{C}}+\mathbf{A}^{(i'j)}_{\mc{C},\mc{C}}-\mathbf{I}_q \right) } \right)^{n/2}  \\
&=\left(\frac{1}{\det \left([\mathbf{A}^{(ij)}_{\mc{C},\mc{C}}]^{-1}+[\mathbf{A}^{(i'j)}_{\mc{C},\mc{C}}]^{-1}-[\mathbf{A}^{(ij)}_{\mc{C},\mc{C}}]^{-1}[\mathbf{A}^{(i'j)}_{\mc{C},\mc{C}}]^{-1}  \right)} \right)^{n/2}\\
&=\left(\det(\mathbf{I}_q-\rho^2\mathbf{E}_{ij}\mathbf{E}_{i'j}  ) \right)^{-n/2}\\
&= (\det\left(\mathbf{I}_q-\rho^2 \mathbf{e}_i \mathbf{e}_{i'}^{\prime}\right))^{-n/2}\\
&=1.
\end{align*}

We now combine the three cases together and note that if
$$
b^2\kappa<0.5, \text{and }  b^2<1, 
$$
we have $\rho^2<0.5$ and 
we can then compute
\begin{align}\label{eq: lower bound sparse chisq divergence}
&\mathbb{E}_{P_0}\left(\frac{f_{1}(\mathbf{X})}{f_{0}(\mathbf{X})}\right)^2 \\
= & 1 + \left(\frac{2}{q(q-1)}\right)^2\frac{q(q-1)}{2}  [(1-\rho^2)^{-n}-1]\nonumber \\
\leq & 1 + \frac{2}{q(q-1)} [(1+2\rho^2)^{n}-1]\nonumber\\
\leq & 1 +  \frac{2}{q(q-1)} \exp ( 2b^2 \log q )\nonumber \\
\leq & 1+ 4 q^{-2 (1- b^2)} ,\nonumber
\end{align}
where the first inequality is because $(1-\rho^2)^{-1}\leq 1+2\rho^2$ since $\rho^2<0.5$, the second inequality is because the inequality $1+x\leq e^{x}$, the last inequality is because $q>2$ implies $q-1>q/2$. 

Since $b^2<1$, we have 
$$
\limsup_{n,q\rightarrow\infty}\mathbb{E}_{P_0}\left(\frac{f_{1}(\mathbf{X})}{f_{0}(\mathbf{X})}\right)^2\leq 1. 
$$
Following the same argument at the end of the proof for Theorem~\ref{thm: lower bound dense}, we can derive an lower bound for the sum of error rates of two types for any test $\phi$, and obtain
$$
\liminf_{n,q\rightarrow\infty} ~~ \inf_{\phi} ~~ \sup_{(i,j)\in \mc{C}^2}  ~~ \left( \mathbb{E}_{P_0} \phi+\mathbb{E}_{P_{ij}}(1-\phi)  \right)\geq 1, 
$$
which implies 
\begin{align*}
\alpha
&\geq\limsup_{n,q\rightarrow\infty} ~~ \sup_{\phi} ~~ \inf_{(i,j)\in \mc{C}^2}  ~~ \left( \mathbb{E}_{P_{ij}}\phi \right)\\
&\geq \limsup_{n,q\rightarrow\infty} ~~ \sup_{\phi} ~~ \inf_{\mathbf{\Omega} \in \mathbf{\Theta}_{n2}^*(b)}  ~~ \left( \mathbb{E}\phi \right).
\end{align*}
This completes the proof. 
\end{proof}

%% file: 9AddSim.tex
In this section, we provide additional simulation results. 
Appendix~\ref{app: L} examines the effect of the parameter $L$ in Algorithm~\ref{alg: exchangeable}. 
Appendix~\ref{app: more graphs} explores more complicated graph structures. 
Appendix~\ref{app: simulation VV} conducts our methods in the setting of an experiment in \cite{verzelen_tests_2009}.
Appendix~\ref{app: PCR_W} explores the use of prior information for constructing test statistics in the framework of the MC-GoF test.

\subsection{Selection of $L$}\label{app: L}

The number of iterations $L$ in Algorithm \ref{alg: exchangeable} determines the length of the generated Markov chains. 
The main text argues that $L=1$ is sufficient for both theoretical and practical purposes. 
This appendix provides more empirical evidence to support this claim. 

Specifically, we present two experiments: 
\begin{enumerate}
    \item Analyze the dependence structure using the autocorrelation function for the generated Markov chain with $L=2,000$.
    \item Compare the power performance of the MC-GoF test with $L=1$, $L=3$, and $L=20$.
\end{enumerate}

In both experiments, we consider running the algorithm with $G_0$ being a cyclic graph with $p=120$ nodes, where each node is connected to its $d=20$ nearest neighbors ($10$ on each side). 
Concretely, the adjacency matrix of this graph is constructed such that:
$$
\mathbf{A}_{G_0}(i,j) = \begin{cases}1 & \text { if } \min (|i-j|, p-|i-j|) \leq d/2, i \neq j \\ 0 & \text { otherwise }\end{cases}
$$
The sample size $n$ is set to $50$, which is around the same magnitude as $d$--this is a high-dimensional scenarios with a dense null graph. 
The true precision matrix $\mathbf{\Omega}$ has nonzero off-diagonal entries at $(i,j)$ if $5\leq \min(|i-j|, p-|i-j|)\leq 14$. 
In other words, on the true graph, each node is connected to its $6$-th to $15$-th neighboring nodes on each side. 
These nonzero entries are randomly and independently assigned values of either $1$ or $-1$. 
The diagonal entries of $\mathbf{\Omega}$ are set to be a constant such that the smallest eigenvalue of $\mathbf{\Omega}$ is 1. We then normalize $\mathbf{\Omega}$ so that the resulting covariance matrix $\mathbf{\Sigma}$ has unit diagonal entries. 

\paragraph{Dependence of Markov Chains.}
Given $n$ samples from $N(\mathbf{0}, \mathbf{\Sigma})$, we implement Step 1 of Algorithm~\ref{alg: exchangeable} with $L=2000$ and obtain the Markov chain $\mc{C}=\left\{\mathbf{X}^{(t)}\right\}_{t=1}^{L}$ where each $\mathbf{X}^{(t)}$ is the output of Algorithm~\ref{alg: residual rotation} according to $\mc{I}=\{1, 2,\ldots, p\}$. Based on the algorithm,  $\mc{C}$ is a stationary process and all $\mathbf{X}^{(t)}$ share the same sufficient statistic for the GGM $\mc{M}_{G_0}$. 

We observe that the process $\mc{C}$ behaves as if it is i.i.d., rather than correlated. 
Specifically, we compute the Gram matrices $\mathbf{S}^{t}=\left(\mathbf{X}^{(t)}\right)^\top\mathbf{X}^{(t)}$ and for each fixed pair of $(i,j)$ absent from $G_0$, we examine the autocorrelation of the process $\{\mathbf{S}^{t}(i,j)\}_{t=1}^{L}$. 
See Figure~\ref{fig:L-trace-acf} for some examples of these processes. 
We note that there is no significant autocorrelation for $\{\mathbf{S}^{t}(i,j)\}_{t=1}^{L}$: we examine the first 1188 pairs of $(i,j)$ absent from $G_0$, among which only 2\% of the calculated lag-1 autocorrelation coefficients exceed 0.05 in absolute value.  
This suggests that the $\mathbf{S}^{t}$ and $\mathbf{S}^{t-1}$ exhibit minimal dependence and the generated data are almost independent.

\begin{figure}[htbp]
    \centering
\includegraphics[width=1\linewidth,angle=270,origin=c]{trace_autocorrelation_plots.pdf}
\caption{Generated values of $\left(\mathbf{X}^{(t)}\right)^\top\mathbf{X}^{(t)}(i,j)$ for selected pairs of $(i,j)$ absent from $G_0$. Left: Trace plots. Right: Autocorrelation functions. }
    \label{fig:L-trace-acf}
\end{figure}

\paragraph{Power performance.}
We further analyze the power properties of the resulting MC-GoF test with various values of $L$. 
For simplicity, we focus on the $F_{\Sigma}$ statistic. 
We implement the tests with $L\in \{1,3,20\}$ as well as the two baseline methods discussed in Section~\ref{sec: Simulation} of the main paper.

The results, computed based on 400 replications, are summarized in Table~\ref{tab:GoF_L}. 
It is clear that the MC-GoF test with any $L$ outperforms the baseline methods, and there is no significant difference in power performance across different values of $L$. 
These observations, along with additional unreported numerical results, suggest that choosing $L$ as small as 1 is often sufficient.

\input{GoF_L_power_se}

\subsection{More Complicated Graph Structures}\label{app: more graphs}  

We extend the experiments in Section~\ref{sec: Simulation} of the main paper to evaluate the power of the proposed MC-GoF tests on more diverse and challenging graph structures. 
Unlike the scenarios in Section~\ref{sec: Simulation}, where true graphs were supergraphs of the null, we consider cases where the null graph may contain many false edges not present in the true graph. 
This null graph could have a dense graph structure relative to the sample size which is challenging for GoF testing. 
Furthermore, we will explore various graph structures, including tree graphs, spatial graphs, small-world graphs, and scale-free graphs, which may be more commonly seen in real-world applications.

All experiments use $N_p(\bs{0}, \mathbf{\Omega}^{-1})$ as the data-generating population, with $p = 120$, $n = 50$, and are repeated for 400 times. 
The null graph $G_0$ and the true graph associated with the precision matrix $\mathbf{\Omega}$ are specified as follows so that they have different graphical structures.

    \begin{itemize}

\item \textbf{Tree vs. Star}

We begin with the case where $G_0$ is a tree graph where each parent has 3 children, and $G$ is a star graph. 
The tree graph is a block graph where each pair of a parent and one of its child form a block. 
A star graph is a special tree graph where there is only one parent and all other nodes are its children. 
Note that $G_0$ and $G$ are not nested: most of the edges in one graph are absent from the other graph. 
This experiment is designed to examine the ability of GoF tests to detect structural differences between fundamentally distinct graph topologies.
    
\item \textbf{Spatial Graph}

A spatial lattice graph captures local dependencies typical in physical systems. 
We consider a 2D spatial lattice graph, where nodes are arranged in a $10\times 12$ grid on $\mathbb{Z}^2$ and distance between nodes is defined as the $\ell_1$-distance. 
The null graph $G_0$ connects the pairs of nodes with distance equals to 1, and the true graph $G$ connects the pairs of nodes with distance equal to 2.

        \item \textbf{Small-World Graph}

A small-world graph is a graph with high clustering and short path lengths \citep{watts1998collective}, which are two of the characteristics shared by many real-world networks. 
We can use a small-world graph using the Watts-Strogatz algorithm with $K$ being an even number and $q$ being the rewiring probability: 1. Construct a regular ring lattice with each node connected to its $K$ nearest neighbors ($K/2$ on each side); 
2. At the $j$th iteration ($1\leq j\leq K/2$), consider each node $i$ and its $j$th rightmost neighbor $k$. With probability $q$, the edge $(i,k)$ is replaced by another edge $(i,k')$ where $k'$ is chosen uniformly random over all nodes except for $i$, $k$, and those connected to $i$. 
We refer to Figure 1 in \cite{watts1998collective} for the specific procedure. 

We implement this algorithm with $K=8$ and $q=0.3$ to get the null graph $G_0$. 
The true graph $G$ is obtained by removing edges in $G_0$ with probability 0.2 and adding edges absent in $G_0$ with probability 0.1.

        \item \textbf{Scale-free Graph}

A scale-free graph is a graph whose node degrees follow a power-law distribution \citep{barabasi1999emergence}. 
We can use the preferential attachment algorithm to generate a scale-free graph.  
For two integers $m$ and $m_0\geq m$, we start with a completely connected subgraph with $m_0$ nodes, and update the graph incrementally. 
At each time, we add a new node and $m$ edges connecting this node with $m$ of the existing nodes with probability proportional to the current node degrees. 
This preferential attachment ensures that high-degree nodes emerge so that the degree distribution is heavy tail. 
We implement this algorithm with $m=2$ and $m_0=3$ to get the null graph $G_0$. 
The true graph $G$ is obtained in the same way as in the setting of Small-World Graph.

\end{itemize}

Given the true graph $G$, the precision matrix $\mathbf{\Omega}$ is obtained by assigning symmetric Bernoulli random variables on the entries associated with $G$ and set the diagonal to be a positive constant so that the smallest eigenvalue of $\mathbf{\Omega}$ is 2. 
The results are summarized in Table~\ref{tab:moregraphs_power_table}. 
We recognize that the MC-GoF tests with any of the four test statistics significantly outperform the baseline methods across all types of graphs.
In particular, $F_\Sigma$ consistently achieves the highest power. These experiments consolidate the superior performance of our proposed method for various graphical structures.

\input{moreGraphs_power_table}

\subsection{Experiments in \citet{verzelen_tests_2009} }\label{app: simulation VV}

To test the GoF of GGMs,  we additionally consider an experimental setup in \citet{verzelen_tests_2009} for a balanced comparison with existing methods. 
According to the setting of the first simulation experiment in \citet[Section 3]{verzelen_tests_2009}, the parameter $\eta$ controls the percentages of edges in the true graph $G$, and the positions of edges are chosen uniformly. \footnote{\citet{verzelen_tests_2009} did not mention how a non-integer is rounded, so we decide to set the number of edges equal to $\lfloor \eta \times  C_{p}^2 \rfloor$.} 
Then we generate a random symmetric $p \times p$ matrix $U$, where for each pair of $i<j$, the element $U[i, j]$ is either sampled uniformly from the interval [-1, 1] if an edge connects nodes $i$ and $j$, or it is set to zero otherwise. 
Each diagonal entry of $U$ is set to the sum of the absolute values of its row's off-diagonal entries plus a small constant, ensuring that $U$ is diagonally dominant and positive definite. 
Finally, $U$ is standardized so that the diagonal entries all equal 1 to obtain the partial correlation matrix $\Pi$, which is then used for sample generation.

Their null graph $G_0$ is built from $G$ by randomly deleting the existing edges at the proportion of $q$. In other words, $\lceil q~ \eta~ C_p^2 \rceil$ edges are randomly removed from $G$ to define $G_0$. 
In their simulations, the dimension $p$ is fixed at $15$, the sample size $n$ varies among 10, 15, and 30,  $\eta$ is either $0.1$ or $0.15$, and $q$ ranges from 0.1 to 1.

In this experiment, the sample size is so small that the GLasso algorithm does not provide any meaningful estimates, so we do not consider using the GLR-$\ell_1$ statistic for the MC-GoF test here. 
Instead, we replace it with the $\text{F}_{\max}$ statistic, which resembles the test statistic used in the $M^1 P_1$ procedure proposed by \citet{verzelen_tests_2009}.

The significance level of each test is set to $\alpha=0.05$. 
We repeat the experiment 400 times and report the estimated power in Table~\ref{tab: verzelen} in a similar format as in Table~\ref{tab: band}. 
As the sample size $n$ increases, the power of every method improves as expected. 
Furthermore, we observe that when $q$ is $1$ (i.e., the null graph $G_0$ is an empty graph), all testing methods, including the Bonferroni adjustment of partial correlation tests (\nameDP), work exceptionally well even with a small sample size $n=10$. 
This is very different from the experiments in Section~\ref{sec: Simulation} of the main paper, where \nameDP substantially underperforms our methods regardless of the sample size. 

The parameter $q$ controls the number of different edges between the true graph $G$ and the null graph $G_0$. 
We observe that as $q$ increases, every method works better because the deviation between $G$ and $G_0$ increases. 
In addition, the $M^1 P_1$ procedure (VV) stands out in every setting and can achieve higher power than our MC-GoF test with the statistics PRC and ERC. 
The gap is especially large in the cases with $q=0.1$; in these cases, with either value of $\eta$, $G_0$ has two fewer edges than $G$. 
Although the number of edges omitted by $G_0$ is so small in these cases, the signal is strong enough for \nameVV to reject the null hypothesis with a probability of at least $36\%$ (the minimum appears when $\eta=0.15$ and $n=10$). 
We call this pattern as \textit{strong but sparse} to describe the deviation of the true distribution from the null hypothesis in these cases. 
This is in contrast to some of the experimental settings in Section~\ref{sec: Simulation}, where the pattern of the deviation is dense but weak, and both \nameVV and \nameDP underperform our methods by a large amount. 

Overall, the MC-GoF tests with the F$_{\Sigma}$ and F$_{\max}$ statistics remain competitive to \nameVV in every case. These comparisons further justify our recommendations in Appendix~\ref{sec: gof test F} of using F$_{\Sigma}$ as the default test statistic for the MC-GoF test, because the simulation results here and the ones in Section~\ref{sec: Simulation} of the main paper together demonstrate its efficiency in power and its robustness to signal patterns.

\begin{table}[ht]
    \centering
    \caption{Power of various GoF tests for the first experiment in \citet{verzelen_tests_2009}. $p=15$.  Significance level $\alpha=0.05$. }\label{tab: verzelen}
    \begin{minipage}{1\textwidth}
\caption*{(a) $\eta=0.1$}
 \centering
 \input{table-4-eta=0.1}
    \end{minipage}
    \hfill
    \begin{minipage}{1\textwidth}
\caption*{(b) $\eta=0.15$}
      \centering
      \input{table-4-eta=0.15}
\end{minipage}
 
\end{table}

\subsection{Additional simulations on GoF tests with prior information} \label{app: PCR_W}
We examine the power of the MC-GoF test using the statistic function PRC-w formulated in
Equation~\ref{eq: partial cor ss weight} by extending the simulation shown in Table \ref{tab: band}. 
Here we focus on a small signal magnitude $s=0.1$ since our methods already perform sufficiently well when $s$ takes a larger value as shown in Table~\ref{tab: band}. 
Assuming that we have the prior information that the true graph $G$ might be a band graph with $K=6$, we design a weight matrix $\bs{W}=(w_{ij})_{1\leq i,j\leq p}$, where $w_{ij}=0.8$ for $i \in \left\{1, \cdots, p \right\}$ and $1\leq |i-j|\leq 6$, and the other off-diagonal entries are equal to 0.2.  
Note that this prior information does not depend on $\mathbf{X}$, so the MC-GoF test using the PRC-w statistic remains valid. 

To investigate the assistance of the prior information, we repeat the experiment in Table~\ref{tab: band} using the MC-GoF test with the PRC-w statistic 400 times and report the results in Table~\ref{tab: weight}, where we keep the results of the other MC-GoF tests in Table~\ref{tab: band} for the same settings. We omit the results for \nameVV and \nameDP since they underperform in these settings. 
The results demonstrate that the MC-GoF tests using statistics incorporated with the weight matrix $\bs{W}$ achieve higher power than the original ones. 
In most cases, either PRC-w or ERC-w is the most powerful among these tests. 

\begin{table}[!btph]
    \caption{Boosting power by incorporating prior information in MC-GoF tests using the weighted statistics. The graphs are band graphs with $K=6, K_0=1, s=0.1$. Significance level $\alpha=0.05$. } \label{tab: weight}
\input{table-1-K=6_w}
\end{table} 

%% file: GoF_L_power_se.tex
\begin{table}[ht]
\centering
\begin{tabular}{|c|c|c|}
\hline
Method & Power & SE \\
\hline
$M^1P_1$ & 0.062 & 0.012 \\
Bonf & 0.043 & 0.010 \\
MC-GoF ($L=1$) & 0.762 & 0.021 \\
MC-GoF ($L=3$) & 0.770 & 0.021 \\
MC-GoF ($L=20$) & 0.748 & 0.022 \\
\hline
\end{tabular}
\caption{Power of different goodness-of-fit tests for the dense cycle graph $G_0$ with degree $20$. 
The dimension $p=120$ and the sample size $n=50$.
Power and standard errors are estimated based on 400 replications. 
}
\label{tab:GoF_L}
\end{table}

%% file: moreGraphs_power_table.tex
\begin{sidewaystable}[!bpth]
\centering
\begin{tabular}{lcccccc}
  \hline
Model & $M^1P_1$ & Bonf & PRC & ERC & F$_\Sigma$ & GLR-$\ell_1$ \\ 
  \hline
Tree vs. Star & 0.040 (0.010) & 0.058 (0.012) & 0.395 (0.024) & 0.385 (0.024) & \textbf{0.532} (0.025) & 0.362 (0.024) \\ 
  Lattice & 0.117 (0.016) & 0.117 (0.016) & 0.700 (0.023) & 0.672 (0.023) & \textbf{0.917} (0.014) & 0.782 (0.021) \\ 
  Small World & 0.117 (0.016) & 0.068 (0.013) & 0.590 (0.025) & 0.407 (0.025) & \textbf{0.995} (0.004) & 0.953 (0.011) \\ 
  Scale Free & 0.095 (0.015) & 0.020 (0.007) & 0.177 (0.019) & 0.877 (0.016) & \textbf{0.993} (0.004) & 0.948 (0.011) \\ 
   \hline
\end{tabular}
\caption{Power of various GoF tests in for different graphical models at the significance level $\alpha= 0.05$. 
The dimension $p=120$ and the sample size $n=50$. See the text in Appendix~\ref{app: more graphs} for detailed description of the null graphs and the true population. 
The first two columns correspond to the baseline methods and the last four columns correspond to the MC-GoF test with four test statistic functions. Standard errors are in parentheses.} 
\label{tab:moregraphs_power_table}
\end{sidewaystable}

%% file: table-4-eta=0.1.tex
\begin{tabular}{ll|cccccc}
\hline
& & \multicolumn{6}{c}{Method} \\ 
$q$ & $n$ & \nameVV & \nameDP & PRC & ERC & F-${\sum}$ & \multicolumn{1}{c}{F-max} \\ 
\hline
\hline 0.1 & 10  & \textbf{0.852 (.032)} & 0.352 (.042) & 0.289 (.040) & 0.367 (.043) & \textbf{0.852 (.032)} & 0.844 (.032) \\
 & 15  & 0.922 (.024) & 0.523 (.044) & 0.453 (.044) & 0.391 (.043) & 0.922 (.024) & \textbf{0.930 (.023)} \\
 & 30  & \textbf{0.977 (.013)} & 0.758 (.038) & 0.609 (.043) & 0.453 (.044) & 0.961 (.017) & \textbf{0.977 (.013)} \\
\hline 0.4 & 10  & 0.961 (.017) & 0.867 (.030) & 0.883 (.029) & 0.930 (.023) & \textbf{0.969 (.015)} & 0.961 (.017) \\
 & 15  & \textbf{0.992 (.008)} & 0.977 (.013) & 0.953 (.019) & 0.930 (.023) & \textbf{0.992 (.008)} & \textbf{0.992 (.008)} \\
 & 30  & 0.992 (.008) & \textbf{1.000 (.000)} & 0.977 (.013) & 0.945 (.020) & 0.992 (.008) & 0.992 (.008) \\
\hline 1 & 10  & 0.992 (.008) & 0.992 (.008) & \textbf{1.000 (.000)} & 0.992 (.008) & 0.992 (.008) & 0.992 (.008) \\
 & 15  & 0.992 (.008) & \textbf{1.000 (.000)} & \textbf{1.000 (.000)} & \textbf{1.000 (.000)} & \textbf{1.000 (.000)} & \textbf{1.000 (.000)} \\
 & 30  & \textbf{1.000 (.000)} & \textbf{1.000 (.000)} & \textbf{1.000 (.000)} & \textbf{1.000 (.000)} & \textbf{1.000 (.000)} & \textbf{1.000 (.000)} \\
\hline 
\end{tabular}

%% file: table-4-eta=0.15.tex
\begin{tabular}{ll|cccccc}
\hline
& & \multicolumn{6}{c}{Method} \\ 
$q$ & $n$ & \nameVV & \nameDP & PRC & ERC & F-${\sum}$ & \multicolumn{1}{c}{F-max} \\ 
\hline
\hline 0.1 & 10  & 0.367 (.043) & 0.031 (.015) & 0.055 (.020) & 0.125 (.029) & \textbf{0.414 (.044)} & 0.305 (.041) \\
 & 15  & 0.570 (.044) & 0.195 (.035) & 0.164 (.033) & 0.148 (.032) & \textbf{0.586 (.044)} & 0.578 (.044) \\
 & 30  & 0.781 (.037) & 0.406 (.044) & 0.258 (.039) & 0.180 (.034) & 0.789 (.036) & \textbf{0.797 (.036)} \\
\hline 0.4 & 10  & 0.750 (.038) & 0.469 (.044) & 0.430 (.044) & 0.500 (.044) & \textbf{0.852 (.032)} & 0.680 (.041) \\
 & 15  & 0.930 (.023) & 0.758 (.038) & 0.805 (.035) & 0.695 (.041) & \textbf{0.984 (.011)} & 0.930 (.023) \\
 & 30  & \textbf{0.992 (.008)} & 0.969 (.015) & 0.953 (.019) & 0.883 (.029) & \textbf{0.992 (.008)} & \textbf{0.992 (.008)} \\
\hline 1 & 10  & 0.883 (.029) & 0.914 (.025) & \textbf{1.000 (.000)} & \textbf{1.000 (.000)} & 0.992 (.008) & 0.898 (.027) \\
 & 15  & 0.969 (.015) & 0.992 (.008) & \textbf{1.000 (.000)} & \textbf{1.000 (.000)} & \textbf{1.000 (.000)} & \textbf{1.000 (.000)} \\
 & 30  & \textbf{1.000 (.000)} & \textbf{1.000 (.000)} & \textbf{1.000 (.000)} & \textbf{1.000 (.000)} & \textbf{1.000 (.000)} & \textbf{1.000 (.000)} \\
\hline 
\end{tabular}

%% file: table-1-K=6_w.tex
\begin{tabular}{ll|cccccc}
\hline
& & \multicolumn{6}{c}{Method} \\ 
$p$ & $n$ & PRC & PRC-w & ERC & ERC-w & F$_{\Sigma}$ & \multicolumn{1}{c}{GLR-${\ell_1}$} \\ 
\hline
\hline 20 & 20  & 0.062 (.012) & 0.085 (.014) & 0.060 (.012) & 0.095 (.015) & \textbf{0.098} (.015) & 0.085 (.014) \\
 & 40  & 0.217 (.021) & 0.280 (.022) & 0.190 (.020) & \textbf{0.320} (.023) & 0.212 (.020) & 0.165 (.019) \\
 & 80  & 0.542 (.025) & 0.725 (.022) & 0.517 (.025) & \textbf{0.785} (.020) & 0.510 (.025) & 0.352 (.024) \\
 
\hline 120 & 20  & 0.125 (.017) & \textbf{0.268} (.022) & 0.128 (.017) & 0.263 (.022) & 0.115 (.016) & 0.070 (.013) \\
 & 40  & 0.225 (.021) & 0.830 (.019) & 0.215 (.021) & \textbf{0.888} (.016) & 0.260 (.022) & 0.160 (.018) \\
 & 80  & 0.723 (.022) & \textbf{1.000} (.000) & 0.685 (.023) & \textbf{1.000} (.000) & 0.705 (.023) & 0.370 (.024) \\
\hline 
\end{tabular}

%% file: 9AppendixApplication.tex
In this section, we provide detailed information about the applications of our testing methods to real-world datasets. 

\subsection{Average Daily Precipitation in the United States}\label{app: US data}
Average daily precipitation, measured in mm, is collected from the North America Land Data Assimilation System (see \cite{NLDAS} or the system's homepage: \url{https://wonder.cdc.gov/NASA-Precipitation.html}). 
Our dataset spans from 1979 to 2011 and comprises 48 states, excluding Alaska, Hawaii, and Puerto Rico due to their geographical isolation, and the District of Columbia due to its encirclement by Maryland and Virginia. 
The data range has sample size $n=33$ and dimension $p=48$.  
Figure~\ref{fig:US} displays the map of the average daily precipitation in the contiguous United States in 2011 generated by \cite{NLDAS}. 

\begin{figure}[H]
\centering
    \includegraphics[width=0.8\textwidth]{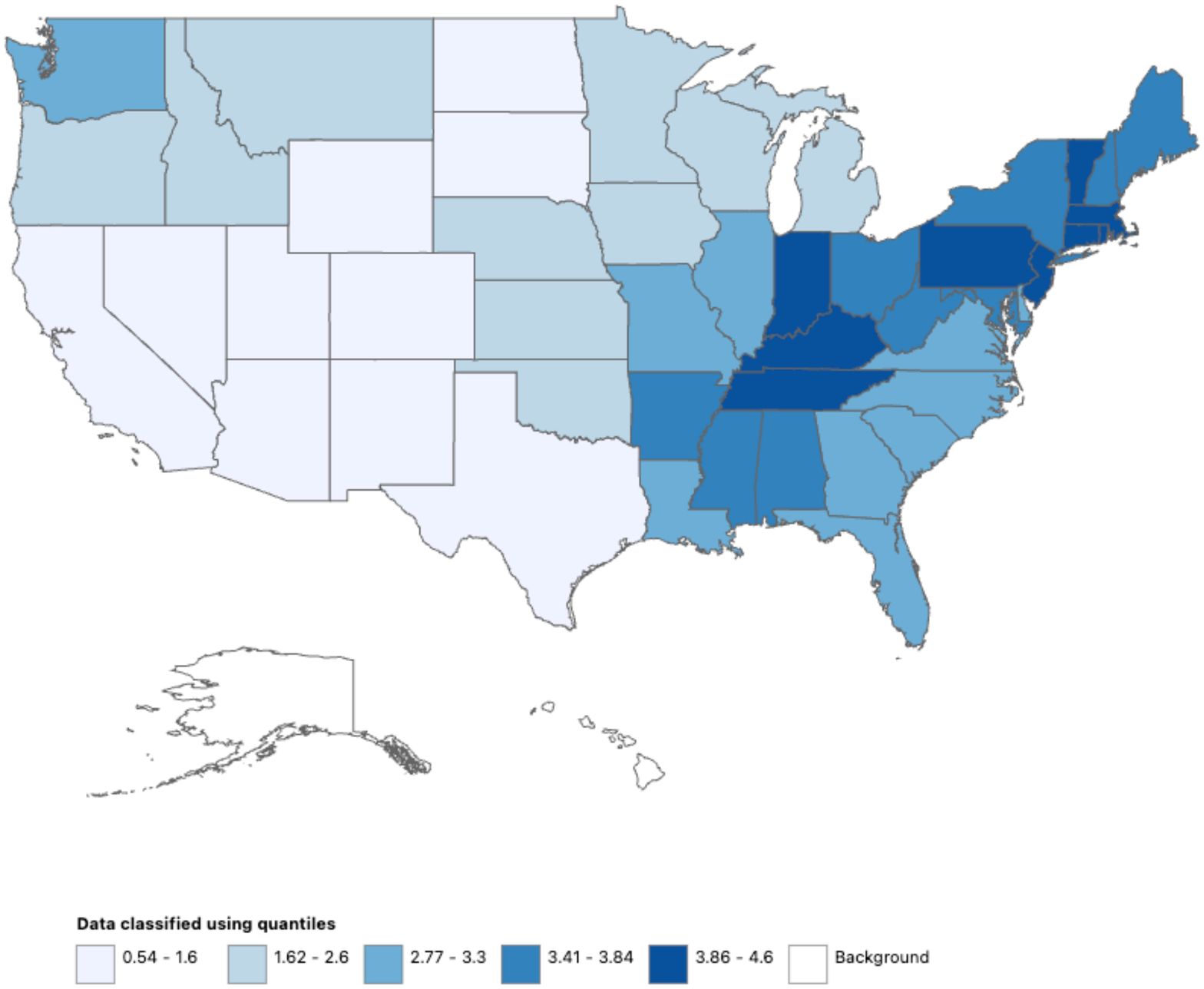}
    \caption{Map of Contiguous United States average daily precipitation in 2011 by \cite{NLDAS}}
    \label{fig:US}
\end{figure}

Assuming the yearly averaged data are independent and identically distributed, we conduct the Box--Cox transformation to normalize the data. 
For any state, suppose the data are $Z_{i}$ ($i\in [n]$). we first shift the data using $\tilde{Z}_i=Z_i - Z_{\min} + 0.01 \times (Z_{\max} -Z_{\min}) $, where $Z_{\max}$ and $Z_{\min}$ are the maximum and the minimum of $Z_{i}$'s. 
We then transform the shifted data through the following function 
\begin{equation}
\varphi(z) = 
\begin{cases}
    \frac{z^\lambda - 1}{\lambda}, & \text{if } \lambda \neq 0, \\
    \ln(z), & \text{if } \lambda = 0,
\end{cases}
\end{equation}
where the parameter $\lambda$ is tuned by using the \texttt{boxcox()} function in the \textbf{R} package \texttt{MASS}. 
For each state, the transformed data pass the Shapiro--Wilk test with a p-value larger than the significance level $\alpha=0.05$. 
This ensures that the distributions of the transformed data are approximately following some normal distributions.

Before our analysis, we perform a simulation study to assess the validity of these tests. 
We simulate data from the GGM with parameters estimated from the data under the null hypothesis. The precision matrix is estimated using a modified GLasso as in \eqref{modified glasso} but with $\lambda$ set to be extremely large. 
We repeat this experiment 400 times and report the results in Table \ref{tab: US error control}, which confirms that all methods control the Type-I error at the nominal level $\alpha=0.05$. 

\begin{table}[!btph]
\caption{Estimated sizes of GoF tests using simulated data from the fitted model for US precipitation data ($\alpha=0.05$).  }\label{tab: US error control}
 \centering

\input{US_type_i_err}
\end{table} 

The testing procedures considered here are identical to those in Section~\ref{sec: Simulation}, specifically, our MC-GoF tests with statistics F$_{\Sigma}$, GLR-${\ell_1}$, PRC, and ERC, as well as the benchmark methods \nameVV and \nameDP. 
Table~\ref{tab: US pvalue in appendix} summarizes the p-values of the GoF tests. 
The MC-GoF test using the F$_{\Sigma}$ and the GLR-${\ell_1}$ statistic reject the null hypothesis with p-values smaller than 0.05, which suggests that a GGM defined by the geographic adjacency of the states falls short of modeling the precipitation. 
A plausible explanation is that the data are strongly associated with precipitation in neighboring regions outside the U.S. such as Canada, Mexico, and the oceans, and the absence of the precipitation data in these regions breaks the conditional independence between non-adjacent states. 
Besides, the other tests fail to reject the null hypothesis. 
This is consistent with our simulation results in Section~\ref{sec: Simulation} that the MC-GoF test using F$_{\Sigma}$ or GLR-${\ell_1}$ can often achieve the highest power.

\begin{table}[hbtp]
    \centering
     \caption{Results of GoF tests for modeling precipitation using geographic adjacency. }
    \label{tab: US pvalue in appendix}
\input{US_pvalues}
\end{table}

We further investigate the influence of the sample size on the power of the GoF tests. To do this, we repeatedly draw a subsample of size $n_s$ from the full dataset and perform the tests on this smaller dataset. 
For each subsample size $n_s$, we repeat the experiment 400 times and compute the fraction of times each test rejects the null hypothesis. 
Since these subsamples are not independently drawn from the population, the rejection rate in this experiment cannot be interpreted as power. 
Nonetheless, the rejection rate is an informative indicator of the influence of the sample size. 
The results are shown in Table~\ref{tab: US subsample}. As $n_s$ increases, the rejection rate of the MC-GoF test using the F$_{\Sigma}$ statistic exhibits a clear increasing trend, whereas the MC-GoF test using the GLR-${\ell_1}$ statistic almost always rejects the null hypothesis. 
The rejection rate resulting from the PRC statistic increases to $22\%$ when $n_s$ increases to 21 but remains lower than 30\%. The other three tests all have rejection rates lower than $10\%$. In conclusion, the rejection decision given by the MC-GoF test using the GLR-${\ell_1}$ statistic is robust with respect to the sample size, and the MC-GoF test using the F$_{\Sigma}$ statistic is able to recover its decision more often when the subsample size increases.

\begin{table}[hbtp]
    \centering
     \caption{Rejection rates of GoF tests for modeling U.S. precipitation by geographic adjacency using a subsample of size $n_s$ across 400 replications ($\alpha=0.05$). Standard errors are placed in parentheses. }
    \label{tab: US subsample}
\input{US_subsample_power}
\end{table}

\subsection{Dependence of fund return}\label{app: stock detail}

In the context of stock market analysis, building a graphical model for stock returns can uncover dependencies among stocks and identify clusters or sector-specific interactions. 
Such insights are valuable for portfolio management, risk assessment, and market structure analysis.

We consider the weekly average returns of 103 large-cap U.S. stocks, comprising 101 stocks from the Standard \& Poor's 100 Index (SP100) and the Dow Jones Industrial Average (DJIA). 
The return is defined as $\left(P_t-P_{t-1}\right) / P_{t-1}$, where $P_t$ and $P_{t-1}$ are the adjusted closing prices of the current day and the last trading day on the stock market, respectively. 
We then apply a carefully chosen transformation to enhance the normality shown in the data (details on the transformation are deferred to the end). 
This results in $n=92$ observations of averaged stock return $X$ and the dimension $p$ of $X$ is 103.

Our analysis begins with a graphical modeling for the stock returns $X$. Suppose the observations of $X$ are i.i.d. sampled from a normal distribution, and we want to find a graph $G$ so that the distribution belongs to a GGM w.r.t. $G$.  
Here we aim at getting a super-graph $G$ of the true graph so that even if the $G$ includes some false edges, the model remains adequate for modeling the data. 
To facilitate the selection of such a super-graph, we incorporate the prior knowledge that the stocks can be classified into 11 sectors based on their industries 
(details on the sectors are deferred to the end). 
We estimate $G$ using GLasso with regularization parameter $\lambda=0.15$ and with no penalty on edges between stocks in the same sector. 
This approach yields a graph $\widehat{G}$ (Figure~\ref{fig: sp100_G} in Appendix~\ref{app: stock detail}), with a maximal degree of $30$, a median degree of 19, and around $18.4\%$ pairs of nodes are connected.

To validate the adequacy of the estimated graph $\widehat{G}$ for modeling the stock returns, we perform the MC-GoF test (Algorithm~\ref{alg:gof}) with $G_0=\widehat{G}$ and statistic F$_{\Sigma}$ (the overall winner in the simulation studies) as well as the two existing benchmarks \nameVV and \nameDP to test the goodness-of-fit of the GGM with $\widehat{G}$. 
The p-values are listed in Table~\ref{tab: GoF stock} and are all above 0.3, which indicates that the GGM with $\widehat{G}$ is sufficient for modeling the stock data. 
This validation supports the suitability of the estimated graph for downstream analysis. 
For instance, if we are interested in a stock-related quantity, such as the return of a particular fund, we can leverage the GGM to conduct model-X inference.

\begin{table}[h]
    \centering
    \caption{Results of GoF tests for modeling the stock returns by the GGM with the estimated graph. }
    \label{tab: GoF stock}
    \input{Stock_GoF}
\end{table}

\paragraph{Information of data collection, transformation, and sectors. }

An exchange-traded fund (ETF) is a basket of securities that tracks a specific index. 
The SPDR Dow Jones Industrial Average ETF Trust (DIA) was introduced in January 1998 and has been proven popular in the market \citep{HEGDE20041043}. 
The Standard and Poor's 100 Index (SP100) includes 101 stocks (there are two classes of stock corresponding to one of the component companies). The holdings of the SP100 can be found on the homepage of the index-linked product iShares S\&P 100 ETF (\url{https://www.ishares.com/us/products/239723/ishares-sp-100-etf}). We used the version of the holdings on September 29, 2023 for our analysis. We define sectors of these stocks almost the same as shown on the webpage but put the three stocks \verb|MA|, \verb|PYPL|, and \verb|V| of payment technology companies in the sector of Information Technology rather than Financial, and put \verb|TGT| in Consumer Discretionary rather than Consumer Staples. In our analysis, we have also included the two stocks \verb|TRV| and \verb|WBA| that are included by DJIA but not by the SP100. In total, we consider 103 stocks and they are divided into 11 sectors: 
\begin{enumerate}
\item Communication Services (10): CHTR, CMCSA, DIS, GOOG, GOOGL, META, NFLX, T, TMUS, VZ
\item Consumer Discretionary (12): AMZN, BKNG, F, GM, HD, LOW, MCD, NKE, SBUX, TGT, TRV, TSLA
\item Consumer Staples (10): CL, COST, KHC, KO, MDLZ, MO, PEP, PG, PM, WMT
\item Energy (3): COP, CVX, XOM
\item Financials (15): AIG, AXP, BAC, BK, BLK, BRK-B, C, COF, GS, JPM, MET, MS, SCHW, USB, WFC
\item Health Care (15): ABBV, ABT, AMGN, BMY, CVS, DHR, GILD, JNJ, LLY, MDT, MRK, PFE, TMO, UNH, WBA
\item Industrials (13): BA, CAT, DE, EMR, FDX, GD, GE, HON, LMT, MMM, RTX, UNP, UPS
\item Information Technology (17): AAPL, ACN, ADBE, AMD, AVGO, CRM, CSCO, IBM, INTC, MA, MSFT, NVDA, ORCL, PYPL, QCOM, TXN, V
\item Materials (2): DOW, LIN
\item Real Estate (2): AMT, SPG
\item Utilities (4): DUK, EXC, NEE, SO
\end{enumerate}

For any stock, denoted by $Z_{i}$'s the 5-day averages of returns. 
We then consider transformation using a piece-wise function 
\begin{equation}
\psi(z) = 
\begin{cases}
    \text{sign}(z) \cdot |z|^\lambda, & \text{if } |z| < c, \\
    \text{sign}(z) \cdot ( a  + 0.01 \cdot \log(|z|)), & \text{otherwise},
\end{cases}
\end{equation}
where $a = c^\lambda - 0.01 \cdot \log(c)$. The parameter $\lambda$ is tuned over a grid from 0.1 to 2, and the parameter $c$ is tuned over a grid between the 80\% quantile and the maximum of the observed values of $Z_i$. The criteria to be maximized is the p-value of the Shapiro--Wilk test on the transformed data. 
After this transformation, only two stocks have p-values from the Shapiro--Wilk tests smaller than 0.05, and we alternatively use the Box--Cox transformation for these two stocks as described in Appendix~\ref{app: US data}.

The graph $\widehat{G}$ we estimated using GLasso is shown in Figure \ref{fig: sp100_G}, where nodes are colored according to their sectors. 
\begin{figure}[!t]
\centering
    \caption{Estimated Graph with Color for Sector}
    \includegraphics[scale=0.9]{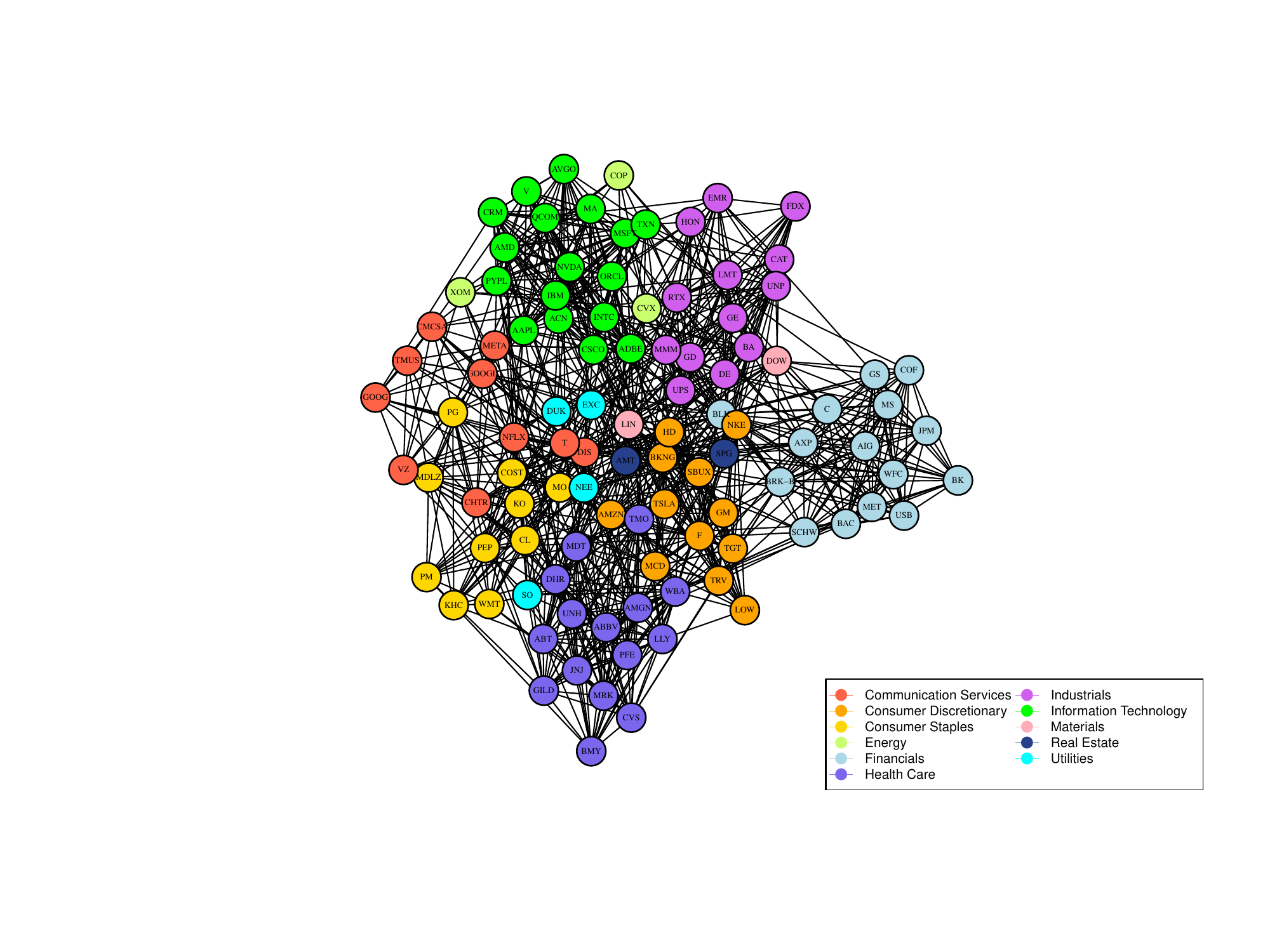}
    \label{fig: sp100_G}
\end{figure}

%% file: US_type_i_err.tex
\begin{tabular}{l|cccccc}
\hline
 & \multicolumn{6}{c}{Method} \\ 
  & \nameVV & \nameDP & PRC & ERC & F$_{\Sigma}$ & \multicolumn{1}{c}{GLR-${\ell_1}$} \\ 
\hline
\hline Size  & 0.062 (0.012) & 0.050 (0.011) & 0.052 (0.011) & 0.055 (0.011) & 0.055 (0.011) & 0.045 (0.010) \\
\hline 
\end{tabular}

%% file: US_pvalues.tex
\begin{tabular}{l|cccccc}
\hline
 & \multicolumn{6}{c}{Method} \\ 
  & \nameVV & \nameDP & PRC & ERC & F$_{\Sigma}$ & \multicolumn{1}{c}{GLR-${\ell_1}$} \\ 
\hline
\hline p-value  & 1.000 & 0.252 & 0.059 & 0.317 & 0.030 & 0.010 \\
\hline 
\end{tabular}

%% file: US_subsample_power.tex
\begin{tabular}{l|cccccc}
\hline
 & \multicolumn{6}{c}{Method} \\ 
$n_{s}$  & \nameVV & \nameDP & PRC & ERC & F$_{\Sigma}$ & \multicolumn{1}{c}{GLR-${\ell_1}$} \\ 
\hline
\hline 18  & 0.065 (.012) & 0.000 (.000) & 0.060 (.012) & 0.068 (.013) & 0.310 (.023) & \textbf{0.973} (.008) \\
21  & 0.085 (.014) & 0.022 (.007) & 0.215 (.021) & 0.065 (.012) & 0.420 (.025) & \textbf{0.993} (.004) \\
24  & 0.080 (.014) & 0.028 (.008) & 0.255 (.022) & 0.055 (.011) & 0.530 (.025) & \textbf{1.000} (.000) \\
27  & 0.072 (.013) & 0.043 (.010) & 0.278 (.022) & 0.030 (.009) & 0.645 (.024) & \textbf{1.000} (.000) \\
\hline 
\end{tabular}

%% file: Stock_GoF.tex
\begin{tabular}{l|ccc}
\hline
 & \multicolumn{3}{c}{Method} \\ 
  & \nameVV & \nameDP & \multicolumn{1}{c}{F$_{\Sigma}$} \\ 
\hline
\hline p-value  & 1.000 & 0.352 & 0.321 \\
\hline 
\end{tabular}